%% file: enriched-diagrams.tex
\documentclass[runningheads]{llncs}

\usepackage{todonotes}
\usepackage{amsmath}
\usepackage{amsfonts}
\usepackage{amssymb}
\usepackage{physics}
\let\tr\undefined
\usepackage{tikz}
\usepackage{mathtools}
\usepackage{stmaryrd}
\usepackage{quiver}
\usepackage{adjustbox}
\usepackage{soul}
\usepackage{mathtools}
\usepackage{tikzit}
\definecolor{zx_grey}{RGB}{211,211,211}
\input{zx2.tikzstyles}
\usepackage{subfig}
\usetikzlibrary{circuits.ee.IEC, positioning}
\usepackage{hyperref}
\usepackage{cleveref}
\usepackage{soul}
\usepackage[misc, geometry]{ifsym}

\tikzset{
  probabilistic split/.style = {trapezium, draw},
  probabilistic join/.style = {trapezium, draw, shape border rotate=180},
  gate/.style = {rectangle, draw, minimum height = 4ex, minimum width = 1.8em}
}

\usepackage{local-macros}

\makeatletter
\newcommand*{\relrelbarsep}{.386ex}
\newcommand*{\relrelbar}{%
  \mathrel{%
    \mathpalette\@relrelbar\relrelbarsep
  }%
}
\newcommand*{\@relrelbar}[2]{%
  \raise#2\hbox to 0pt{$\m@th#1\relbar$\hss}%
  \lower#2\hbox{$\m@th#1\relbar$}%
}
\providecommand*{\rightrightarrowsfill@}{%
  \arrowfill@\relrelbar\relrelbar\rightrightarrows
}
\providecommand*{\leftleftarrowsfill@}{%
  \arrowfill@\leftleftarrows\relrelbar\relrelbar
}
\providecommand*{\xrightrightarrows}[2][]{%
  \ext@arrow 0359\rightrightarrowsfill@{#1}{#2}%
}
\providecommand*{\xleftleftarrows}[2][]{%
  \ext@arrow 3095\leftleftarrowsfill@{#1}{#2}%
}

\newsavebox{\@brx}
\newcommand{\llangle}[1][]{\savebox{\@brx}{\(\m@th{#1\langle}\)}%
  \mathopen{\copy\@brx\kern-0.5\wd\@brx\usebox{\@brx}}}
\newcommand{\rrangle}[1][]{\savebox{\@brx}{\(\m@th{#1\rangle}\)}%
  \mathclose{\copy\@brx\kern-0.5\wd\@brx\usebox{\@brx}}}

\def\@seccntformat#1{\@ifundefined{#1@cntformat}%
   {\csname the#1\endcsname\quad}  % default
   {\csname #1@cntformat\endcsname}% enable individual control
}
\let\oldappendix\appendix %% save current definition of \appendix
\renewcommand\appendix{%
    \oldappendix
    \newcommand{\section@cntformat}{\appendixname~\thesection\quad}
}

\makeatother

\RequirePackage{array}
  {\begingroup
   \newcommand\estyle{}%
   \renewcommand\institute[1]%
     {\\\multicolumn{#1}{@{}c@{}}{\scriptsize\begin{tabular}[t]{@{}>{\footnotesize}c@{}}##1\end{tabular}}}%
   \renewcommand\email[1]%
     {\gdef\estyle{\footnotesize\ttfamily}\\##1\gdef\estyle{}}
   \begin{tabular}[t]{@{}*{#1}{>{\estyle}c}@{}}
  }%
  {\end{tabular}%
   \endgroup
  }
\usepackage{iftex}

\ifpdf
  \usepackage{underscore}         % Only needed if you use pdflatex.
  \usepackage[T1]{fontenc}        % Recommended with pdflatex
\else
  \usepackage{breakurl}           % Not needed if you use pdflatex only.
\fi
\begin{document}
\title{Enriching Diagrams with Algebraic Operations}

\author{Alejandro Villoria \and
Henning Basold \and
Alfons Laarman }
\authorrunning{Villoria et al.}
% First names are abbreviated in the running head.
% If there are more than two authors, 'et al.' is used.
%
\institute{Leiden Institute of Advanced Computer Science, The Netherlands\\
  \email{\{a.d.villoria.gonzalez, h.basold, a.w.laarman\}@liacs.leidenuniv.nl}
  }
\maketitle
\begin{abstract}
  In this paper, we extend diagrammatic reasoning in monoidal categories with algebraic
  operations and equations.
  We achieve this by considering monoidal categories that are enriched in the category
  of Eilenberg-Moore algebras for a monad.
  Under the condition that this monad is monoidal and there is an adjunction
  between the free algebra functor and the underlying category functor,
  we construct an adjunction between
  symmetric monoidal categories and symmetric monoidal categories enriched over
  algebras for the monad.
  This allows us to devise an extension, and its semantics, of the ZX-calculus with probabilistic
  choices by freely enriching over convex algebras, which are the algebras of the
  finite distribution monad.
  We show how this construction can be used for diagrammatic reasoning of noise in quantum systems.
\end{abstract}
\input{content/introduction}

\input{content/background}

\input{content/algebraic-enrichment}

\input{content/enriched-monoidal}

\input{content/zx}

\input{content/discussion}

\nocite{*}
\bibliographystyle{splncs04}
\bibliography{generic}

\appendix
\input{content/appendix}
\input{content/notation}

\end{document}

%% file: zx2.tikzstyles
\tikzstyle{Green Node}=[minimum size=5mm, font={\footnotesize\boldmath}, shape=rectangle, rounded corners=2mm, inner sep=0.2mm, outer sep=-2mm, scale=0.8, tikzit shape=circle, draw=black, fill={rgb,255: red,216; green,248; blue,216}, tikzit fill={rgb,255: red,216; green,248; blue,216}]
\tikzstyle{Empty green}=[inner sep=0mm, minimum size=2mm, shape=circle, draw=black, fill={rgb,255: red,216; green,248; blue,216}]
\tikzstyle{Red Node}=[minimum size=5mm, font={\footnotesize\boldmath}, shape=rectangle, rounded corners=2mm, inner sep=0.2mm, outer sep=-2mm, scale=0.8, tikzit shape=circle, draw=black, fill={rgb,255: red,232; green,165; blue,165}, tikzit fill={rgb,255: red,232; green,165; blue,165}]
\tikzstyle{Empty red}=[inner sep=0mm, minimum size=2mm, shape=circle, draw=black, fill={rgb,255: red,232; green,165; blue,165}]
\tikzstyle{wh}=[minimum size=5mm, font={\footnotesize\boldmath}, shape=rectangle, rounded corners=2mm, inner sep=0.2mm, outer sep=-2mm, scale=0.8, tikzit shape=circle, draw=black, fill={rgb,255: red,255; green,255; blue,255}, tikzit fill={rgb,255: red,255; green,255; blue,255}]
\tikzstyle{h}=[fill=white, draw=black, shape=rectangle, inner sep=0.6mm, minimum height=1.5mm, minimum width=1.5mm]
\tikzstyle{AutoCCZ}=[minimum size=6mm, font={\footnotesize}, fill=white, draw=black, shape=rectangle, scale=0.8]
\tikzstyle{txt}=[minimum size=5mm, font={\footnotesize\boldmath}, shape=rectangle, rounded corners=2mm, inner sep=0.2mm, outer sep=-2mm, scale=0.8, tikzit shape=circle, draw=none, fill=none, tikzit draw=none]
\tikzstyle{bigtxt}=[minimum size=5mm, font={\boldmath}, shape=rectangle, rounded corners=2mm, inner sep=0.2mm, outer sep=-2mm, scale=0.8, tikzit shape=circle, draw=none, fill=none, tikzit draw=none]
\tikzstyle{testt}=[fill=white, draw=black, shape=circle]
\tikzstyle{txtnobold}=[minimum size=5mm, font={\footnotesize}, shape=rectangle, rounded corners=2mm, inner sep=0.2mm, outer sep=-2mm, scale=0.8, tikzit shape=circle, draw=none, fill=none, tikzit draw=none]
\tikzstyle{divide}=[regular polygon, regular polygon sides=3, draw=black, fill={rgb,255: red,211; green,211; blue,211}, inner sep=2.3pt, tikzit category=scal, rounded corners=0.8mm]
\tikzstyle{black}=[fill=black, draw=black, shape=circle, tikzit fill=black, tikzit draw=black, tikzit shape=circle, tikzit category=IH, inner sep=2pt]
\tikzstyle{gather}=[fill={rgb,255: red,211; green,211; blue,211}, draw=black, tikzit category=scal, rounded corners=0.8mm, regular polygon, regular polygon sides=3, inner sep=2.3pt, rotate=180]
\tikzstyle{new style 0}=[fill=white, draw=black, shape=circle]
\tikzstyle{morphism}=[fill=white, draw=black, shape=rectangle, minimum height=5mm, font={\footnotesize}]
\tikzstyle{gnd}=[gnd]

\tikzstyle{hadamard edge}=[-, dashed, dash pattern=on 2pt off 1pt, thin, draw={rgb,255: red,190; green,190; blue,190}]
\tikzstyle{edge}=[-, thick]
\tikzstyle{delay}=[-, fill={rgb,255: red,232; green,165; blue,165}, tikzit fill={rgb,255: red,232; green,165; blue,165}]
\tikzstyle{delay2}=[-, draw=black, fill={rgb,255: red,216; green,248; blue,216}, tikzit draw=black, tikzit fill={rgb,255: red,216; green,248; blue,216}]
\tikzstyle{arrow}=[->]
\tikzstyle{distribution}=[-, fill={rgb,255: red,116; green,179; blue,231}, tikzit fill={rgb,255: red,37; green,77; blue,255}]
\tikzstyle{discontinuous}=[-, dashed, dash pattern=on 4pt off 1pt, thin, draw=black, fill=black, tikzit fill=black, tikzit draw=black]

%% file: content/introduction.tex
\section{Introduction}
\label{sec:introduction}

Monoidal categories are one way of generalizing algebraic reasoning and they can be used to draw
intuitive diagrams that encapsulate this reasoning graphically.
That monoidal categories are a powerful abstraction has been demonstrated in countless areas,
such as linear logic~\cite{dePaiva14:CategoricalSemanticsLinear} or
quantum mechanics~\cite{abramsky-categorical-2007}, just to name a few, and are amenable to
graphical reasoning~\cite{Selinger11:SurveyGraphicalLanguages} with diagrammatic languages such
as the ZX-calculus~\cite{cd11}.
Another abstraction of algebraic reasoning are monads~\cite{Barr85:ToposesTriples,%
  Moggi91:NotionsComputationMonads,PP02:NotionsComputationDetermine}
and their algebras, or
representations thereof~\cite{HP07:CategoryTheoreticUnderstanding,Manes76:AlgebraicTheories},
which are distinct from monoidal
categories in that identities (like associativity) always hold strictly and they allow rather
arbitrary algebraic operations.
In this paper, we set out to combine these two approaches into one framework, in which
monoidal category diagrams can be composed not only sequentially and in parallel with a tensor product
but also with additional algebraic operations.

\begin{figure}[ht]
  \centering
  \vspace*{-5mm}
  \begin{tikzpicture}[bend angle=20, x=0.5cm, y=0.5cm]
    \node[probabilistic split] (s) {};
    \node[gate, below left = 0.9 and 0.3 of s] (g) {$G$};
    \node[gate, below right = 0.9 and 0.3 of s] (e) {$E$};
    \node[probabilistic join, below right = 0.9 and 0.3 of g] (j) {};

    \coordinate[above = 0.3 of s] (i);
    \coordinate[below = 0.3 of j] (o);

    \node[right = 0.1 of i] {$A$};
    \node[right = 0.1 of o] {$B$};

    \draw (s) to[bend right] node[left] {$0.9$} (g);
    \draw (g) to[bend right] (j);
    \draw (s) to[bend left] node[right] {$0.1$} (e);
    \draw (e) to[bend left] (j);
    \draw (i) -- (s);
    \draw (j) -- (o);
  \end{tikzpicture}
  \hspace*{2cm}
  \begin{tikzpicture}[bend angle=20, x=0.5cm, y=0.5cm]
    %% First diagram
    \node[probabilistic split] (s) {};
    \node[gate, below left = 0.9 and 0.3 of s] (g1) {$G_{1}$};
    \node[gate, below right = 0.9 and 0.3 of s] (e) {$E$};
    \node[probabilistic join, below right = 0.9 and 0.3 of g] (j) {};

    \coordinate[above = 0.3 of s] (i1);
    \coordinate[below = 0.3 of j] (o1);

    \node[right = 0.1 of i1] {$A$};
    \node[right = 0.1 of o1] {$B$};

    \draw (s) to[bend right] node[left] {$0.9$} (g1);
    \draw (g1) to[bend right] (j);
    \draw (s) to[bend left] node[right] {$0.1$} (e);
    \draw (e) to[bend left] (j);
    \draw (i1) -- (s);
    \draw (j) -- (o1);

    %% Tensor
    \node[right = 0.1 of e] (tens) {$\otimes$};

    %% Second diagram
    \node[gate, right = 0.1 of tens] (g2) {$G_{2}$};
    \path
      let \p1 = (i1), \p2 = (o1), \p3 = (g2)
      in
        coordinate (i2) at (\x3,\y1)
        coordinate (o2) at (\x3,\y2)
    ;
    \draw (i2) -- (g2);
    \draw (g2) -- (o2);

    \node[right = 0.1 of i2] {$C$};
    \node[right = 0.1 of o2] {$D$};

    %% Equality
    \node[right = 0.1 of g2] (eq) {$=$};

    %% Third diagram
    \node[gate, right = 0.1 of eq] (g12) {$G_{1}$};
    \node[gate, right = 0.1 of g12] (g21) {$G_{2}$};
    \node[gate, right = 0.1 of g21] (e2) {$E$};
    \node[gate, right = 0.1 of e2] (g22) {$G_{2}$};
    \node[probabilistic split, above right = 0.9 and 0 of g21 ] (s2) {};
    \node[probabilistic join, below right = 0.9 and 0 of g21] (j2) {};

    % \coordinate[above left = 0.3 and 0.1 of s2] (i12);
    % \coordinate[above right = 0.3 and 0.1 of s2] (i22);
    \coordinate[above = 0.3 of s2] (i3);
    \coordinate[below = 0.3 of j2] (o3);

    % \node[left = 0.1 of i12] {$A$};
    % \node[right = 0.1 of i22] {$C$};
    \node[right = 0.1 of i3] {$A \otimes C$};
    \node[right = 0.1 of o3] {$B \otimes D$};

    %% Monoidal splitting
    \coordinate[above right = 0.5 and 0 of g12] (ms1);
    \coordinate[above left = 0.5 and 0 of g22] (ms2);
    \draw[double] (s2) -- node[left, inner sep = 10pt, pos = 0.1] {$0.9$} (ms1);
    \draw (ms1) to[bend right] node[left, inner sep = 4pt] {$A$} (g12);
    \draw (ms1) to[bend left] node[right, inner sep = 4pt] {$C$} (g21);
    \draw[double] (s2) -- node[right, inner sep = 10pt, pos = 0.1] {$0.1$} (ms2);
    \draw (ms2) to[bend right] node[left, inner sep = 4pt] {$A$} (e2);
    \draw (ms2) to[bend left] node[right, inner sep = 4pt] {$C$} (g22);

    %% Monoidal joining
    \coordinate[below right = 0.5 and 0 of g12] (mj1);
    \coordinate[below left = 0.5 and 0 of g22] (mj2);
    \draw[double] (mj1) -- (j2);
    \draw[double] (mj2) -- (j2);
    \draw (g12) to[bend right] node[left, inner sep = 4pt] {$B$} (mj1);
    \draw (g21) to[bend left] node[right, inner sep = 4pt] {$D$} (mj1);
    \draw (e2) to[bend right] node[left, inner sep = 4pt] {$B$} (mj2);
    \draw (g22) to[bend left] node[right, inner sep = 4pt] {$D$} (mj2);

    %% Input/Output
    \draw[double] (s2) -- (i3);
    \draw[double] (j2) -- (o3);
  \end{tikzpicture}
  \caption{\textbf{Left:} Probabilistic mix of a gate $G$ with an error $E$.
    \textbf{Right}: Interaction of tensor and convex sum, where double wires visually indicate a
    tensor product}
  \label{fig:probabilistic-mix}
\end{figure}
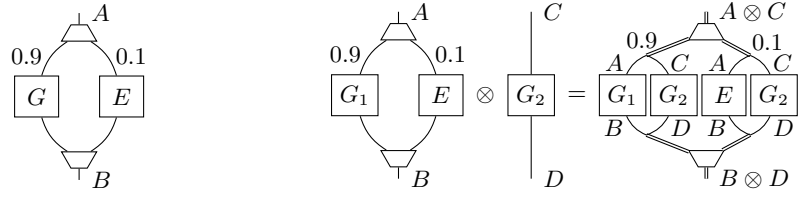
One such operation is the formation of convex combinations, which can be used to create
a probabilistic mix of two or more diagrams.
This occurs, for instance, when reasoning about the behaviour of noise in quantum circuits.
\Cref{fig:probabilistic-mix} shows on the left two quantum logic gates, one called $G$ and one called $E$ that,
respectively, model the wanted behaviour and a possible error.
These two gates are mixed, where $G$ gets a probability of $0.9$ and $E$ of $0.1$.
The trapezoids delimit the combination of the gates, and $A$ and $B$ are the input and output types
of the gates\footnote{We read diagrams from top to bottom.}.
In monoidal categories, the gates in the picture represent morphisms $G, E \from A \to B$
and our aim is to interpret the trapezoid block as a convex sum $G +_{0.9} E$ of these
morphisms, where we define $G +_{p} E = p G + (1-p) E$.
Such sums should also nicely interact with the tensor product.
For instance, if $G_{1} \from A \to B$ and $G_{2} \from C \to D$ are gates, then an identity such
as $(G_{1} +_{0.9} E) \otimes G_{2} = (G_{1}  \otimes G_{2}) +_{0.9} (E \otimes G_{2})$
should hold for these morphisms of type $A \otimes C \to B \otimes D$, see
\Cref{fig:probabilistic-mix} on the right.
Having an operation to form convex combinations together with intuitive identities enables
reasoning about, for example, probabilistic combination and noise in quantum circuits.

The difficulty lies in combining monoidal diagrams with algebraic operations such that the
algebraic identities and the monoidal identities interact coherently.
We will handle this difficulty by using enriched monoidal categories, where the enrichment
yields the algebraic operations and the monoidal structure the parallel composition.
More precisely, we will assume that the algebraic theory is given by a monad $T$ and that the
monoidal categories are enriched over the Eilenberg-Moore category $\emt$ of algebras for this
monad.
Our aim in this paper is to construct for an arbitrary monoidal category $\CC$ an $\emt$-enriched
monoidal category $F\CC$ that is free in the sense that there is an inclusion
$\iota_{\CC} \from \CC \to (F\CC)_0$ into the underlying category of $F\CC$ and for every
$\emt$-enriched monoidal category $\eD$ and monoidal functor $G \from \CC \to \eD_{0}$, there is a
unique $\emt$-enriched monoidal functor $\bar{G}_0 \from F\CC \to \eD$ that makes the following
diagram commute.
\begin{equation*}
  % https://q.uiver.app/#q=WzAsMyxbMCwwLCIoRlxcQ0MpXzAiXSxbMSwwLCJcXGVEXzAiXSxbMCwxLCJcXENDIl0sWzIsMSwiRyIsMl0sWzIsMCwiXFxpb3RhX3tcXENDfSJdLFswLDEsIlxcYmFye0d9XzAiXV0=&macro_url=https%3A%2F%2Fgist.githubusercontent.com%2Fhbasold%2F5b205ad0b469224b4b006c401b7872bd%2Fraw%2F349e29a93d61269cf30340b0fd265dd92be38e27%2Fgistfile1.txt
\begin{tikzcd}
	{(F\CC)_0} & {\eD_0} \\
	\CC
	\arrow["G"', from=2-1, to=1-2]
	\arrow["{\iota_{\CC}}", from=2-1, to=1-1]
	\arrow["{\bar{G}_0}", from=1-1, to=1-2]
\end{tikzcd}
\end{equation*}
This free construction does not work for all monads, but we show that the free enrichment
always exists for monoidal $\SetC$-monads whose free $T$-algebra functor is left adjoint to
the underlying category functor $\emt(I,-):\emt\to\SetC$ for $I$ the monoidal unit of $\emt$.

\paragraph{Contributions}\mbox{}\\
Specifically, we contribute a construction for free enrichment over algebras for some monoidal
monads in~\Cref{thm:underlyingAdjoint} and~\Cref{corollary:free-forget}.
We also show how the enrichment preserves symmetric monoidal structure
in~\Cref{thm:monoidal-enrichment} and~\Cref{co:free}.
Given this construction, we demonstrate how a graphical language for reasoning in monoidal
categories can be enriched with the free algebras for these monads, which enables diagrammatic reasoning of
the interaction between the sequential and parallel compositions with the algebraic structure.
We show how the theory can be applied to obtain convex combinations of ZX-diagrams and what the
resulting identities of diagrams are.
By exploiting the mapping property of the free enrichment, we automatically obtain sound
interpretations of these operations and identities.
Lastly, we describe how we can use the enrichment of ZX-diagrams to reason about noise in quantum systems.

\paragraph{Related Work}\mbox{}\\
ZX-diagrams are \emph{universal} in the sense that they can in principle represent any linear map
between Hilbert spaces of dimension $\bC^{2^n}$~\cite{cd11}.
Indeed, sums and linear/convex combinations~\cite{jeandel-addition-2023,shaikh-how-2022}
of ZX-diagrams can be encoded within the language, but in practice these representations
oftentimes lead to either very large diagrams or to diagrams that do not reveal upon visual
inspection the (linear/convex) structure that the diagram is representing.
This, in return, diminishes the advantages gained by reasoning in terms of abstract graphical
representations.
Our perspective of using enrichment keeps the abstraction barrier and thus makes reasoning about
convex combinations of diagrams tractable.
In general, our theory also covers the recently developed linear combinations of ZX-diagrams such as
\cite{stollenwerk-diagrammatic-2022,muuss-thesis}
and other, so far unexplored, algebraic operations such as those of join-semilattices.
Moreover, the identities that have to be crafted carefully by hand and proven to be sound
fall automatically out of our theory.
Other related work is that of \emph{Sheet diagrams}~\cite{comfort-sheet-2020}
and \emph{Tape diagrams}~\cite{bonchi-deconstructing-2022},
recently developed graphical languages for rig categories,
which are categories with two monoidal structures -- one for addition and one for multiplication.

\paragraph{Outline}\mbox{}\\
The paper is organised as follows.
We start by introducing notation and recalling some background of enriched and
monoidal categories in \Cref{sec:background}.
In \Cref{sec:algebras}, we establish the necessary theory to define categories enriched
over Eilenberg-Moore algebras and we construct a free enrichment over those algebras.
Our next step in \Cref{sec:v-monoidal} is to extend these definitions and the free construction to
also include monoidal structures on categories, which ensures that these enrichment and monoidal
structure coherently interact.
\Cref{sec:zx} is devoted to applying our theory to enrich ZX-diagrams with convex sums to reason
about probabilistic processes such as quantum noise.
We conclude the paper with directions for future work in \Cref{sec:discussion}.

%%% Local Variables:
%%% mode: latex
%%% TeX-master: "../enriched-diagrams"
%%% End:

%% file: content/background.tex
\section{Background}
\label{sec:background}

In this section, we recall some terminology from category theory~\cite{Borceux94:HandbookCategoricalAlgebraBasic,%
  Leinster14:BasicCategoryTheory,maclane,Riehl16:CatTheoryinContext} and introduce some notation.
We denote the collection of objects of a category \textbf{C} as $|\mathbf{C}|$,
and the morphisms from object $A$ to $B$ as $\mathbf{C}(A,B)$.
A \textit{monoidal category} $(\mathbf{C},\otimes,\mathit{I})$ is a category \textbf{C} together with
a functor $\otimes:\mathbf{C} \times \mathbf{C} \to \mathbf{C}$ called the \textit{tensor product}
and an object $\mathit{I} \in |\mathbf{C}|$ called the \textit{tensor unit} subject
to some conditions~\cite{Kelly82:ConceptsEnrichedCatTheory}.
We will often refer to a monoidal category $(\mathbf{C},\otimes,\mathit{I})$ as just \textbf{C}.
A monoidal category is a \textit{symmetric monoidal} category (SMC) when it also has a
\textit{braiding} $\sigma_{A,B}: A\otimes B \to B\otimes A$ such that
$\sigma_{B,A}\circ \sigma_{A,B} = \text{Id}_{A\otimes B}$
that is also subject to coherence conditions~\cite{Kelly82:ConceptsEnrichedCatTheory}.

Given a monoidal category $(\VC, \times, \timesUnit)$, a $\VC$-\emph{(enriched) category} $\eC$ consists
of
\begin{itemize}
\item a class $\obj{\eC}$ of objects,
\item for each pair $A,B \in \obj{\eC}$, an object $\eC(A,B) \in \obj{\VC}$ that we refer to as the
  \emph{hom-object},
\item for objects $A,B,C \in \obj{\eC}$, a \emph{composition} morphism
  $\enr{\circ} \from \eC(B,C) \times \eC(A,B) \to \eC(A,C)$
  in $\mathbf{V}$, and
\item for all $A \in \obj{\eC}$, an \emph{identity element} $j_A \from  \timesUnit \to \eC(A,A)$
\end{itemize}
subject to associativity and unit axioms~\cite{Kelly82:ConceptsEnrichedCatTheory}.
We say that $\VC$ is the \emph{base of enrichment} for $\eC$.
A way to look at the above definition is that we construct a $\VC$-enriched category
$\eC$ by identifying morphisms of some category $\mathbf{C}$ as objects from
$\VC$, which we are able to compose by using the tensor product of $\VC$.
The most well-known example is that of \emph{locally small} categories, in which the morphisms
between two objects form a set, and thus we can see them as objects in the monoidal category
($\mathbf{Set},\times,\timesUnit$) for $\times$ the Cartesian product and $\timesUnit$ the singleton
set.

With a suitable definition of $\VC$-functors and $\VC$-natural transformations, $\VC$-categories
organise themselves into a 2-category~\cite{Kelly82:ConceptsEnrichedCatTheory}, denoted by $\VCat$.
For an SMC $(\VC, \times, \timesUnit)$, $\VCat$ is also an SMC as follows.
We define for $\VC$-categories $\eC$ and $\eD$ a $\VC$-category $\eC \otimes \eD$ with objects
$\obj{\eC \otimes \eD}$ being the categorical product and hom-objects
$(\eC \otimes \eD)((A, B), (C, D)) = \eC(A, C) \times \eD(B, D)$.
The unit is given by
$j_{(A, B)}
= \timesUnit
\cong \timesUnit \times \timesUnit
\xrightarrow{u_{A} \times v_{B}} \eC(A, A) \times \eD(B,B)$
in terms of the units $u$ of $\eC$ and $v$ of $\eD$.
Similarly, one also defines the composition for $\eC \otimes \eD$ in terms of the composition
morphisms of $\eC$ and $\eD$, appealing to the symmetry in~$\VC$~\cite[Sec. 1.4]{Kelly82:ConceptsEnrichedCatTheory}.
The tensor product also extends to $\VC$-functors and $\VC$-natural transformations,
which makes it a 2-functor.
Finally, one defines $\otimesUnit$ to be the unit $\VC$-category with one object $0$ and
$I(0,0) = \timesUnit$ and we thus obtain, with suitable definitions of associators etc., a
symmetric monoidal 2-category $(\VCat, \otimes, \otimesUnit)$.

Most of the categories we are interested in are also \textit{dagger-compact}
categories ($\dagger\text{-}CC$).
These are SMCs $(\mathbf{C},\otimes,\mathit{I})$ with some additional structure.
First, they are equipped with an endofunctor
$\dagger: \mathbf{C}^{op} \to \mathbf{C}$, that satisfies $(\text{Id}_A)^\dagger = \text{Id}_A$,
$(g \circ f)^\dagger = f^\dagger \circ g^\dagger$, $(f^\dagger)^\dagger=f$, and
$(f \otimes g)^\dagger = f^\dagger \otimes g^\dagger$.
And secondly, for every object $A$ there exists a \textit{dual} $A^*$ such that there exists
\textit{unit} $\eta_A: I\to A \otimes A^*$ and \textit{counit} $\epsilon_A: A^*\otimes A \to I$
morphisms subject to some conditions~\cite{heunenvicary}.

We are interested in categories that let us reason about quantum mechanics.
One of them is \textbf{FdHilb}, the category of finite dimensional Hilbert spaces of the form
$\bC^n$ and linear maps as morphisms.
The category $\qbit$ is the (full) subcategory of \textbf{FdHilb} with objects Hilbert spaces of
the form $\bC^{2^n}$ and linear maps.
Similarly, the category $\cpm$~\cite{selinger-dagger-2007} has objects $\bC^{2^n}$ and
morphisms completely positive linear maps between them~\cite{caretteGround}.
We usually work in $\qbit$ when reasoning about pure quantum evolutions and in $\cpm$ when
impure quantum evolutions (such as noise) can take place.
All of these categories are $\dagger\text{-}CC$, with the monoidal structure $\otimes$ given by
the usual Kronecker product of vector spaces and the dagger $\dagger$ being the conjugate transpose.

%%% Local Variables:
%%% mode: latex
%%% TeX-master: "../enriched-diagrams"
%%% End:

%% file: content/algebraic-enrichment.tex
\section{Algebraic Enrichment}
\label{sec:algebras}

In this section, we are going to recall the concept of monoidal and affine monads,
and discuss some properties of the \emph{Eilenberg-Moore} category of a monad.
We also start applying the \emph{Distribution monad} and the \emph{Multiset monad} to
running examples that will be of interest in later sections.

The Distribution monad $(\md,\mu,\eta)$%
\footnote{We will often refer to a monad $(\mt,\mu,\eta)$ as $\mt$.} contains the functor
$\md:\SetC\to\SetC$ that maps a set $A$ to the set $\md(A)$ of (finitely supported)
probability distributions over elements of $A$.
We write probability distributions as formal convex sums: $\sum_{a} p_a [a]\in \md(A)$ such that
$a\in A, p_a \in [0,1]$, and $\sum_a p_a = 1$.
$\md$ acts on a morphism $f$ by simply applying $f$ to the underlying set:
$(\md f)(\sum_{a} p_a [a]) = \sum_{a} p_a [f(a)]$.
The unit of the monad is the map $\eta: A\to\md(A): a \mapsto 1[a]$
(the Dirac distribution), and the multiplication
$\mu$ ``flattens'' a distribution of distributions by multiplying the probabilities together:
$\mu: \md(\md(A)) \to \md(A): \sum_q p_q [\sum_a q_a [a]]\mapsto \sum_{a}r_a[a]$ where
$r_a = \sum_q p_q q_a$~\cite{jacobsNew}.

The functor $\md$ is also a monoidal functor, which makes $(\md,\mu,\eta)$ a \emph{monoidal monad}.
In particular, this means that there exists a map:
$$\nabla: \md(A)\times\md(B)\to\md(A\times B): \left( \sum_{a} p_a [a], \sum_{b} p_b [b] \right) \mapsto \sum_{a,b} p_ap_b [(a,b)]$$
for every $A,B\in |\mathbf{Set}|$.

A monad $\mt:\CC\to\CC$ is \emph{affine} if there is an isomorphism
$\mt(\timesUnit) \cong \timesUnit$ for $\timesUnit$ the terminal
object of $\CC$~\cite{jacobs-affine}.
This is the case for $\md$.

If $\md$ is a monad for expressing convex combinations of elements of a set,
the Multiset monad $\mm$ is its analogue for linear combinations with coefficients
over some semiring.

We recall that given any monad $\mt$ in $\CC$ we can construct its Eilenberg-Moore category,
with objects $\mt$-\emph{algebras} of the form $(A,\alpha_A)$ for $A\in\obj{\CC}$ and
$\mt$-\emph{action} $\alpha_A:\mt(A)\to A$ such that
$\alpha_A \circ T(\alpha_{A}) = \alpha_A \circ \mu_A$ and $\alpha_A \circ \eta_A = \text{Id}_A$.
Algebra homomorphisms $f:(A,\alpha_A)\to (B,\alpha_B)$ are morphisms of the underlying objects
$f: A\to B$ that commute with the action: $f \circ \alpha_A = \alpha_B \circ T(f)$.
The identity and composition follow from the ones for the underlying objects~\cite{maclane}.

For a monad $\mt$ on $\CC$, we have that $\emt$ is complete whenever $\CC$ is complete.
Cocompleteness is not as immediate, but if $\CC = \SetC$ then $\emt$ is also cocomplete
\cite{Barr85:ToposesTriples,jacobsConvex}.
This makes $\emt$ over monads on $\SetC$ a complete and cocomplete category and
in particular, $\emt$ has reflexive coequalizers, which we use to define the tensor
product of algebras.

When $\mt$ is a monoidal monad on a monoidal category $(\CC,\otimes,I)$,
the tensor product of $\mt$ algebras $(A,a), (B,b)$, denoted $(A,a)\otimes^\mt(B,b)$,
is (if it exists) defined as the coequalizer diagram~\cite{seal,brandenburg}
\begin{equation}
  \label{eq:coeq}
  % https://q.uiver.app/#q=WzAsMyxbMCwwLCJGKFQoQSlcXG90aW1lcyBUKEIpKSJdLFsyLDAsIkYoQVxcb3RpbWVzIEIpIl0sWzMsMCwiKEEsYSlcXG90aW1lc19UIChCLGIpIl0sWzAsMSwiXFxtdVxcY2RvdCBGKFxcbmFibGEpIiwwLHsib2Zmc2V0IjotMn1dLFswLDEsIkYoYVxcb3RpbWVzIGIpIiwyLHsib2Zmc2V0IjoyfV0sWzEsMiwicSJdXQ==
  \begin{tikzcd}
    {F(T(A)\otimes T(B))} && {F(A\otimes B)} & {(A,a)\otimes^\mt (B,b)}
    \arrow["{\mu\cdot F(\nabla)}", shift left=2, from=1-1, to=1-3]
    \arrow["{F(a\otimes b)}"', shift right=2, from=1-1, to=1-3]
    \arrow["q", from=1-3, to=1-4],
  \end{tikzcd}
\end{equation}

where $F:\CC\to \emt: A\mapsto (\mt(A),\mu)$ is the left adjoint to the
\emph{forgetful} functor $U:\emt\to\CC: (A,\alpha_A)\mapsto A$ that maps objects to
their free algebras over $\mt$.
Given that we need $\emt$ to be monoidal in order to use it as a base of enrichment,
diagram \eqref{eq:coeq} above is a convenient representation of the tensor product of algebras.
The rest of the structure to make $\emt$ a (symmetric) monoidal category follows under certain
conditions, in particular when $(\CC,\otimes,I)$ is a closed (S)MC and the coequalizer
\eqref{eq:coeq} exists for all algebras $(A,a), (B,b)$~\cite{seal}.

We can define (symmetric) monoidal structure in the category of \emph{free algebras} as follows.
Using \eqref{eq:coeq} for the category of free algebras over a monoidal monad, we have that
the following diagram forms a coequalizer.
\begin{equation}
  \label{eq:freecoeq}
  % https://q.uiver.app/#q=WzAsMyxbMCwwLCJGKFRUKEEpXFxvdGltZXMgVFQoQikpIl0sWzIsMCwiRihUQVxcb3RpbWVzIFRCKSJdLFs0LDAsIkYoQVxcb3RpbWVzIEIpIl0sWzAsMSwiXFxtdVxcY2RvdCBGKFxcbmFibGEpIiwwLHsib2Zmc2V0IjotMn1dLFswLDEsIkYoXFxtdVxcb3RpbWVzIFxcbXUpIiwyLHsib2Zmc2V0IjoyfV0sWzEsMiwiXFxtdVxcY2RvdCBGKFxcbmFibGEpIl1d
  \begin{tikzcd}[column sep=2.2em]
    {F(TT(A)\otimes TT(B))} && {F(T(A)\otimes T(B))} && {F(A\otimes B)}
    \arrow["{\mu\cdot F(\nabla)}", shift left=2, from=1-1, to=1-3]
    \arrow["{F(\mu\otimes \mu)}"', shift right=2, from=1-1, to=1-3]
    \arrow["{\mu\cdot F(\nabla)}", from=1-3, to=1-5]
  \end{tikzcd}
\end{equation}
Therefore, $(\mt(A),\mu)\otimes^\mt (\mt(B),\mu) \cong F(A \otimes B)$~\cite[Prop. 2.5.2]{seal}.
The monoidal unit is $I^\mt = (T(I),\mu)$, while the associator, unitor, and symmetry (if present) are
the images of the ones in $(\CC,\otimes,I)$ under $F$.
We then have that $(\emt, \otimes^\mt, I^\mt)$ is the (symmetric) monoidal category
of free $\mt$-algebras.

A functor $F: (\VC_1,\otimes_{\VC_1},I_{\VC_1}) \to (\VC_2,\otimes_{\VC_2},I_{\VC_2})$
between two monoidal categories can be lifted to a $2$-functor
$F_*: \ECatC{\VC_1}\to \ECatC{\VC_2}$~\cite{borceux1994}.
This is called a \emph{change of enriching}, where we turn a $\VC_1$-category into a
$\VC_2$-category.
Indeed, given a $\VC_1$-category $\eC$, we can construct the $\VC_2$-category
$F_*\eC$ by defining $\obj{F_*\eC}:= \obj{\eC}$ and, for every $A, B \in \obj{F_*\eC}$,
the hom-objects are $F_*\eC(A,B) := F(\eC(A,B))$ with composition and identity element
following from the ones in $\eC$ under $F$.
For a symmetric monoidal category $(\VC,\otimes,I)$, an important instance is the functor
$\VC(I, -)_{*} \from \VCat \to \CatC$ called the \emph{underlying category} functor and it
is denoted by $(-)_{0}$.

The following lemma states explicitly the case when one of the enriching categories is $\SetC$.
\begin{lemma}[{\cite[Prop. 6.4.7]{borceux1994}}]
  \label{lemma:underlyingAdjoint}
  Let $(\VC,\otimes,I)$ be a closed symmetric monoidal category with coproducts.
  Then the hom-functor $\VC (I,-) \from \VC \to \SetC$ has a left adjoint $F$ that
  sends a set $X$ to $F(X)=\coprod_X I$, the $X$-th fold copower of $I$.
  Moreover, $F$ is a strong morphism of symmetric monoidal categories and the induced $2$-functor
  $F_{*}$ is left-adjoint to the underlying category functor $(-)_{0}$.
\end{lemma}

\begin{theorem}
  \label{thm:underlyingAdjoint}
  A monoidal monad $\mt$ on $(\SetC,\times,\timesUnit)$ endows the category
  of $\mt$-algebras $\emt$ with a bicomplete (complete and cocomplete) closed SMC structure.
  This allows to lift the free-forgetful adjunction of~\Cref{lemma:underlyingAdjoint} as a
  change of enriching between $\emt$-categories and $\SetC$-categories for a monoidal $\mt$.
\end{theorem}
\begin{proof}
  The proof follows from \Cref{lemma:underlyingAdjoint} and previous arguments.
  Given that $\emt$ for $T$ a monad on $\SetC$ is bicomplete, then coequalizer~\eqref{eq:coeq}
  exists and we can define tensor products of algebras.
  We can then make $\emt$ a symmetric monoidal category given that $\SetC$ is closed symmetric
  monoidal.
  Finally, $\emt$ can be made into a closed category following~\cite{kock-1971} given
  that $\SetC$ has equalizers.
  We can then use~\Cref{lemma:underlyingAdjoint} to create a change of enriching between
  $\emt$-categories and $\SetC$-categories.
\end{proof}

\begin{corollary}
  \label{corollary:free-forget}
  Let $\mt$ be a monoidal monad on $\SetC$ defined by a free-forgetful adjunction
  $U: \emt \rightleftarrows \SetC: L$. If $L$ is naturally
  isomorphic to the functor $F$ from
  \Cref{lemma:underlyingAdjoint}, that is $(T(-),\mu)\cong\coprod_{(-)} (TI,\mu)$, then the induced
  $2$-functor $L_*$ is left adjoint to the underlying category functor $(-)_{0}$.
  This lets us
  use \Cref{thm:underlyingAdjoint} to enrich locally
  small categories with the free algebras over $T$.
  The condition $L \cong F$ holds, in particular, when $T$ is an affine monad.
\end{corollary}
\begin{proof}
  Whenever we have that $L \cong F$, the enrichment over free $T$-algebras comes simply from
  substituting $F$ with $L$ in \Cref{lemma:underlyingAdjoint} and \Cref{thm:underlyingAdjoint}.
  To see that this condition holds when $T$ is affine, we construct hom isomorphisms
  $\emt(L(X),Y)\cong\SetC(X,U(Y))\cong\SetC(X,\emt(I^\mt,Y))\cong\emt(F(X),Y)$
  for some $X\in \obj{\SetC}, Y\in\obj{\emt}$, with the second and last isomorphism coming
  from their respective adjunctions.
  The remaining one is due to $T(*)\cong *$, which allows us to get algebra homomorphisms
  $h: (T(*),\mu) \mapsto (T(Y),\mu)$ from maps $h': * \mapsto Y$, while the other direction just
  requires to forget the homomorphism structure.
\end{proof}

Let us construct an example for \Cref{thm:underlyingAdjoint} and relate it to graphical languages.
If we have a locally small monoidal category $\CC$ with morphisms $f,g:A\to B, h:B\to C$,
represented graphically as
$\scalebox{0.8}{\input{figures/f.tikz}},
\scalebox{0.8}{\input{figures/g.tikz}}, \scalebox{0.8}{\input{figures/h.tikz}}$,
we can freely enrich $\CC$ over $\emd$ following the change of enriching category method above.
Then, we can realize graphically a probabilistic process involving $f$ and $g$ with probability
$0.9$ and $0.1$ respectively, followed by applying $h$ deterministically
(that is, it occurs with probability $1$) afterwards as follows.

\begin{equation}\label{eq:example1}
  \scalebox{1}{\input{figures/example1.tikz}}
\end{equation}

Intuitively, we distinguish between probabilistic and deterministic processes by having the former
enclosed within \emph{distribution brackets}
(in the same way as we would represent them as a formal sum $0.9[f] + 0.1[g]$),
that we choose to depict as trapezoids in this paper.
Deterministic processes are depicted without the bracket enclosing mostly as syntactic sugar,
otherwise they would simply have a single choice with probability $1$.
We can see how wires can have weights inside this environment, and how each wire
represents a \emph{probabilistic choice}.
Intuition also tells us that we could for example rewrite the diagram above by distributing
$h$ over the two probabilistic branches.
We will discuss in later sections which graphical rules capture the interactions present
in these enriched categories.

It is natural to then ask if our enriched category $\CC$ maintained its monoidal structure,
and if other desired properties (such as braiding and symmetry, if present) would still hold too.
We will address this in the next section.

%%% Local Variables:
%%% mode: latex
%%% TeX-master: "../enriched-diagrams"
%%% End:

%% file: figures/f.tikz
\begin{tikzpicture}
	\begin{pgfonlayer}{nodelayer}
		\node [style=AutoCCZ] (24) at (0, 0) {$f$};
		\node [style=none] (25) at (0, 0.75) {};
		\node [style=none] (26) at (0, -0.75) {};
	\end{pgfonlayer}
	\begin{pgfonlayer}{edgelayer}
		\draw (25.center) to (24);
		\draw (24) to (26.center);
	\end{pgfonlayer}
\end{tikzpicture}

%% file: figures/g.tikz
\begin{tikzpicture}
	\begin{pgfonlayer}{nodelayer}
		\node [style=AutoCCZ] (0) at (0, 0) {$g$};
		\node [style=none] (1) at (0, 0.75) {};
		\node [style=none] (2) at (0, -0.75) {};
	\end{pgfonlayer}
	\begin{pgfonlayer}{edgelayer}
		\draw (1.center) to (0);
		\draw (0) to (2.center);
	\end{pgfonlayer}
\end{tikzpicture}

%% file: figures/h.tikz
\begin{tikzpicture}
	\begin{pgfonlayer}{nodelayer}
		\node [style=AutoCCZ] (0) at (0, 0) {$h$};
		\node [style=none] (1) at (0, 0.75) {};
		\node [style=none] (2) at (0, -0.75) {};
	\end{pgfonlayer}
	\begin{pgfonlayer}{edgelayer}
		\draw (1.center) to (0);
		\draw (0) to (2.center);
	\end{pgfonlayer}
\end{tikzpicture}

%% file: figures/example1.tikz
\begin{tikzpicture}
	\begin{pgfonlayer}{nodelayer}
		\node [style=none] (0) at (0, 2.75) {};
		\node [style=none] (1) at (-0.25, 2.75) {};
		\node [style=none] (2) at (0.25, 2.75) {};
		\node [style=none] (3) at (-0.5, 2.25) {};
		\node [style=none] (4) at (0.5, 2.25) {};
		\node [style=none] (5) at (-0.5, 2.25) {};
		\node [style=none] (6) at (0, 2.25) {};
		\node [style=none] (7) at (0, -0.75) {};
		\node [style=none] (8) at (-0.25, -0.75) {};
		\node [style=none] (9) at (0.25, -0.75) {};
		\node [style=none] (10) at (-0.5, -0.25) {};
		\node [style=none] (11) at (0.5, -0.25) {};
		\node [style=none] (12) at (-0.5, -0.25) {};
		\node [style=none] (13) at (0, -0.25) {};
		\node [style=none] (15) at (0, 3.25) {};
		\node [style=none] (17) at (0.25, 2.25) {};
		\node [style=none] (18) at (-0.25, 2.25) {};
		\node [style=none] (19) at (0.25, -0.25) {};
		\node [style=none] (20) at (-0.25, -0.25) {};
		\node [style=none] (22) at (0, -2.25) {};
		\node [style=morphism] (24) at (-0.5, 1) {$f$};
		\node [style=morphism] (25) at (0.5, 1) {$g$};
		\node [style=txtnobold] (27) at (-1, 2) {$0.9$};
		\node [style=txtnobold] (28) at (1, 2) {$0.1$};
		\node [style=morphism] (29) at (0, -1.5) {$h$};
	\end{pgfonlayer}
	\begin{pgfonlayer}{edgelayer}
		\draw [style=distribution] (6.center) to (4.center);
		\draw [style=distribution] (4.center) to (2.center);
		\draw [style=distribution] (2.center) to (0.center);
		\draw [style=distribution] (0.center) to (1.center);
		\draw [style=distribution] (1.center) to (3.center);
		\draw [style=distribution] (3.center) to (5.center);
		\draw [style=distribution] (5.center) to (6.center);
		\draw [style=distribution] (13.center) to (11.center);
		\draw [style=distribution] (11.center) to (9.center);
		\draw [style=distribution] (9.center) to (7.center);
		\draw [style=distribution] (7.center) to (8.center);
		\draw [style=distribution] (8.center) to (10.center);
		\draw [style=distribution] (10.center) to (12.center);
		\draw [style=distribution] (12.center) to (13.center);
		\draw [style=none, bend right=15] (24) to (20.center);
		\draw [style=none, bend right=15] (19.center) to (25);
		\draw [style=none, bend left=15] (24) to (18.center);
		\draw [style=none, bend left=15] (17.center) to (25);
		\draw (15.center) to (0.center);
		\draw (7.center) to (29);
		\draw (29) to (22.center);
	\end{pgfonlayer}
\end{tikzpicture}

%% file: content/enriched-monoidal.tex
\section{Enriched Monoidal Categories}
\label{sec:v-monoidal}

Recall from \Cref{sec:background} that a symmetric monoidal category $\VC$ gives rise to
a symmetric monoidal 2-category $(\VCat, \otimes, \otimesUnit)$ of $\VC$-categories.
This structure allows us to define an enriched (symmetric) monoidal category to be a (symmetric)
pseudo-monoid in
$\VCat$~\cite{DS97:MonoidalBicategoriesHopf}, which amounts to the following explicit
definition~\cite{KYZ+21:EnrichedMonoidalCategories,morrison}.
Let us denote by $S$ the symmetry isomorphism of $\VCat$.
A symmetric monoidal $\VC$-category is a tuple
$(\eC, \odot, \odotUnit, \alpha, \lambda, \rho, \sigma)$
consisting of:
\begin{itemize}
\item a $\VC$-enriched category $\eC$
\item a $\VC$-functor $\odotUnit \from \otimesUnit \to \eC$
\item a $\VC$-functor $\odot \from \eC \otimes \eC \to \eC$
\item $\VC$-natural isomorphisms
  $\alpha \from \odot \comp (\odot \otimes \Id_{\eC}) \to \odot \comp (\Id_{\eC} \otimes \odot)$
  (associator),
  $\lambda \from \odot \comp (U \otimes \Id_{\eC}) \to \Id_{\eC}$ (left unitor),
  $\rho \from \odot \comp (\Id_{\eC} \otimes U) \to \Id_{\eC}$ (right unitor), and
  $\sigma \from \odot \to \odot \comp S$ (symmetry)
\end{itemize}
subject to the expected coherence axioms~\cite{DS97:MonoidalBicategoriesHopf}.
A (symmetric) monoidal $\VC$-functor $(\eC, \odot_{1}, U_{1}) \to (\eD, \odot_{2}, U_{2})$
is a lax (symmetric) pseudo-monoid homomorphism, which means that it
consists of a $\VC$-functor $h \from \eC \to \eD$ and two $\VC$-natural transformations
$h^{0} \from U_{2} \to h \circ U_{1}$ and
$h^{2} \from \odot_{2} \comp (h \otimes h) \to h \comp \odot_{1}$
that are coherent with the associators, unitors and
symmetries~\cite{KYZ+21:EnrichedMonoidalCategories}.
Together, symmetric monoidal $\VC$-categories and functors form a category $\SMVCat$.

Our goal is now to lift the adjunction between enriched categories from~\Cref{thm:underlyingAdjoint}
to also include enriched monoidal structure.
To this end, we introduce \emph{lax monoidal strict $2$-functors}, which are
tuples $(G, G^{0}, G^{2})$ where $G \from \VC \to \WC$ is a strict functor of $2$-categories
and  $(G, G^{0}, G^{2}) \from (\VC, \otimes, \otimesUnit) \to (\WC, \times, \timesUnit)$
is a lax monoidal functor on the underlying $1$-categories.

\begin{theorem}
  \label{thm:pseudo-monoid-adjunction}
  Lax monoidal strict $2$-functors
  $(G, G^{0}, G^{2}) \from (\VC, \otimes, \otimesUnit) \to (\WC, \times, \timesUnit)$
  induce $2$-functors
  $\PMon(G) \from \PMon(\VC, \otimes, \otimesUnit) \to \PMon(\WC, \times, \timesUnit)$
  between $2$-categories of pseudo-monoids, lax homomorphisms and $2$-cells that are compatible
  with the homomorphism structures.
  If a $G$ has a monoidal left adjoint $F$, then $\PMon(F)$ is left adjoint to $\PMon(G)$.
  Finally, if the monoidal categories and functors are symmetric, then the adjunction
  can be improved to one between symmetric pseudo-monoids.
\end{theorem}
\begin{proof}
  The details and appropriate diagram chases are written in~\Cref{appProofThm2}, which go through
  the following steps.
  We begin by showing that $\PMon(G)$ maps  a pseudo-monoid
  $(\eC, \odot, \odotUnit, \alpha, \lambda, \rho, \sigma)$
  in $(\VCat, \otimes, I)$ to a pseudo-monoid
  $(G\eC, G(\odot) \circ G^{2}, G(\odotUnit) \circ G^{0}, G\alpha, G\lambda, G\rho, G\sigma)$
  in $(\WCat,\times,\timesUnit)$ by checking that it fulfills the pseudo-monoid
  axioms~\cite{DS97:MonoidalBicategoriesHopf}.

  Similarly, we check that $(G,G^{0},G^{2})$ maps
  a pseudo-monoid homomorphism $(h,h^{0},h^{2})$ to a pseudo-monoid homomorphism
  $(G(h),G(h^{0}),G(h^{2}))$.
  
  If $G$ has a strict left adjoint~$F$, which is also strong monoidal,
  we show that $(F, F^{0}, F^{2}) \dashv (G, G^{0}, G^{2})$ is a \emph{monoidal $2$-adjunction} if
  the mates~\cite{KS74:ReviewElements2categories} of $G^{0}$ and $G^{2}$ are the inverses of
  $F^{0}$ and $F^{2}$, respectively, as in the following equations, where $\beta_{A,B}$ is the
  natural isomorphism $\WC(A, GB) \xrightarrow{\cong} \VC(FA, B)$ and $\eta$ the unit of the
  adjunction:
  \begin{equation*}
  (F^{0})^{-1} = \beta(G^{0})
  \quad \text{and} \quad
  (F^{2})^{-1} = \beta(G^{2}) \comp F(\eta \times \eta).
\end{equation*}
\end{proof}

The following theorem, which shows that the change of enrichment extends to monoidal enriched
categories, follows from \Cref{thm:pseudo-monoid-adjunction} using that
$\SMVCat = \PMon(\ECatC{\VC}, \otimes, \otimesUnit)$
and that the change of enrichment gives a lax monoidal $2$-adjunction~\cite{borceux1994,Crutt}.

\begin{theorem}
  \label{thm:monoidal-enrichment}
  If $(G, G^{0}, G^{2}) \from (\VC, \otimes, \otimesUnit) \to (\WC, \times, \timesUnit)$
  is a symmetric monoidal functor between symmetric monoidal categories
  with a monoidal left adjoint $(F, F^{0}, F^{2})$, then there are adjunctions that commute
  with the forgetful functors as in the following diagram.
  \begin{equation*}
    % https://q.uiver.app/#q=WzAsNCxbMCwxLCJcXFZDYXQiXSxbMSwxLCJcXENhdEMiXSxbMSwwLCJcXFNNQ2F0QyJdLFswLDAsIlxcU01WQ2F0Il0sWzAsMSwiKC0pXzAiLDIseyJjdXJ2ZSI6MX1dLFszLDIsIigtKV8wIiwyLHsiY3VydmUiOjF9XSxbMiwzLCJGIiwyLHsiY3VydmUiOjF9XSxbMSwwLCJGIiwyLHsiY3VydmUiOjF9XSxbMywwXSxbMiwxXSxbNCw3LCIiLDIseyJsZXZlbCI6MSwic3R5bGUiOnsibmFtZSI6ImFkanVuY3Rpb24ifX1dLFs1LDYsIiIsMix7ImxldmVsIjoxLCJzdHlsZSI6eyJuYW1lIjoiYWRqdW5jdGlvbiJ9fV1d&macro_url=https%3A%2F%2Fgist.githubusercontent.com%2Fhbasold%2F5b205ad0b469224b4b006c401b7872bd%2Fraw%2F349e29a93d61269cf30340b0fd265dd92be38e27%2Fgistfile1.txt
\begin{tikzcd}[column sep=large]
	\SMVCat & \SMECatC{\WC} \\
  \VCat & \ECatC{\WC}
	\arrow[""{name=0, anchor=center, inner sep=0}, "{\chEnr{G}}"', curve={height=6pt}, from=2-1, to=2-2]
	\arrow[""{name=1, anchor=center, inner sep=0}, "{\chEnr{G}}"', curve={height=6pt}, from=1-1, to=1-2]
	\arrow[""{name=2, anchor=center, inner sep=0}, "\chEnr{F}"', curve={height=6pt}, from=1-2, to=1-1]
	\arrow[""{name=3, anchor=center, inner sep=0}, "\chEnr{F}"', curve={height=6pt}, from=2-2, to=2-1]
	\arrow[from=1-1, to=2-1]
	\arrow[from=1-2, to=2-2]
	\arrow["\dashv"{anchor=center, rotate=270}, draw=none, from=0, to=3]
	\arrow["\dashv"{anchor=center, rotate=270}, draw=none, from=1, to=2]
\end{tikzcd}
  \end{equation*}
\end{theorem}

From this theorem and combining the results from~\Cref{sec:algebras}
we can derive the following corollary, which is our main tool for building monoidal
diagrams that are enriched with algebraic operations.
\begin{corollary}
  \label{co:free}
  A monoidal monad $\mt$ on $\SetC$ with an adjunction between the free $T$-algebra functor
  and the underlying category functor (see \Cref{corollary:free-forget}) gives a free-underlying
  adjunction
  $(-)_0:\EmtCat \rightleftarrows \CatC:F_*$.
  This adjunction lifts to an adjunction between the $2$-categories
  of symmetric monoidal $\emt$-enriched and $\SetC$-enriched categories
  $(-)_0: \SMemtCat \rightleftarrows \SMCatC: F_*$.
  More explicitly, given such a monad $T$ on $\SetC$ and a SMC
  $(\CC, \otimes, I, \alpha, \lambda, \rho, \sigma)$ we can construct the freely $\emt$-enriched
  SMC
  $(\eC, \etens, \enr{I}, \enr{\alpha}, \enr{\lambda}, \enr{\rho}, \enr{\sigma})$.
\end{corollary}

Knowing that we keep the symmetric monoidal structure after doing the free enrichment,
we can justify drawing parallel composition of probabilistic operations in diagram form.
Continuing example~\eqref{eq:example1}, let us have another probabilistic process
in which $f'$ and $g'$ occur with probabilities $0.7$ and $0.3$ (respectively) parallelly
composed. Then we can draw the following picture.

\begin{equation}\label{eq:example2}
  \scalebox{1}{\input{figures/example2.tikz}}
\end{equation}

%%% Local Variables:
%%% mode: latex
%%% TeX-master: "../enriched-diagrams"
%%% End:

%% file: figures/example2.tikz
\begin{tikzpicture}
	\begin{pgfonlayer}{nodelayer}
		\node [style=none] (0) at (-2, 1.75) {};
		\node [style=none] (1) at (-2.25, 1.75) {};
		\node [style=none] (2) at (-1.75, 1.75) {};
		\node [style=none] (3) at (-2.5, 1.25) {};
		\node [style=none] (4) at (-1.5, 1.25) {};
		\node [style=none] (5) at (-2.5, 1.25) {};
		\node [style=none] (6) at (-2, 1.25) {};
		\node [style=none] (7) at (-2, -1.75) {};
		\node [style=none] (8) at (-2.25, -1.75) {};
		\node [style=none] (9) at (-1.75, -1.75) {};
		\node [style=none] (10) at (-2.5, -1.25) {};
		\node [style=none] (11) at (-1.5, -1.25) {};
		\node [style=none] (12) at (-2.5, -1.25) {};
		\node [style=none] (13) at (-2, -1.25) {};
		\node [style=none] (14) at (-2, 2.25) {};
		\node [style=none] (15) at (-1.75, 1.25) {};
		\node [style=none] (16) at (-2.25, 1.25) {};
		\node [style=none] (17) at (-1.75, -1.25) {};
		\node [style=none] (18) at (-2.25, -1.25) {};
		\node [style=none] (19) at (-2, -3.25) {};
		\node [style=morphism] (20) at (-2.5, 0) {$f$};
		\node [style=morphism] (21) at (-1.5, 0) {$g$};
		\node [style=txtnobold] (22) at (-3, 1) {$0.9$};
		\node [style=txtnobold] (23) at (-1, 1) {$0.1$};
		\node [style=morphism] (24) at (-2, -2.5) {$h$};
		\node [style=none] (25) at (2, 1.75) {};
		\node [style=none] (26) at (1.75, 1.75) {};
		\node [style=none] (27) at (2.25, 1.75) {};
		\node [style=none] (28) at (1.5, 1.25) {};
		\node [style=none] (29) at (2.5, 1.25) {};
		\node [style=none] (30) at (1.5, 1.25) {};
		\node [style=none] (31) at (2, 1.25) {};
		\node [style=none] (32) at (2, -1.75) {};
		\node [style=none] (33) at (1.75, -1.75) {};
		\node [style=none] (34) at (2.25, -1.75) {};
		\node [style=none] (35) at (1.5, -1.25) {};
		\node [style=none] (36) at (2.5, -1.25) {};
		\node [style=none] (37) at (1.5, -1.25) {};
		\node [style=none] (38) at (2, -1.25) {};
		\node [style=none] (39) at (2, 2.25) {};
		\node [style=none] (40) at (2.25, 1.25) {};
		\node [style=none] (41) at (1.75, 1.25) {};
		\node [style=none] (42) at (2.25, -1.25) {};
		\node [style=none] (43) at (1.75, -1.25) {};
		\node [style=none] (44) at (2, -3.25) {};
		\node [style=morphism] (45) at (1.25, 0) {$f'$};
		\node [style=morphism] (46) at (2.75, 0) {$g'$};
		\node [style=txtnobold] (47) at (1, 1) {$0.7$};
		\node [style=txtnobold] (48) at (3, 1) {$0.3$};
	\end{pgfonlayer}
	\begin{pgfonlayer}{edgelayer}
		\draw [style=distribution] (6.center) to (4.center);
		\draw [style=distribution] (4.center) to (2.center);
		\draw [style=distribution] (2.center) to (0.center);
		\draw [style=distribution] (0.center) to (1.center);
		\draw [style=distribution] (1.center) to (3.center);
		\draw [style=distribution] (3.center) to (5.center);
		\draw [style=distribution] (5.center) to (6.center);
		\draw [style=distribution] (13.center) to (11.center);
		\draw [style=distribution] (11.center) to (9.center);
		\draw [style=distribution] (9.center) to (7.center);
		\draw [style=distribution] (7.center) to (8.center);
		\draw [style=distribution] (8.center) to (10.center);
		\draw [style=distribution] (10.center) to (12.center);
		\draw [style=distribution] (12.center) to (13.center);
		\draw [style=none, bend right=15] (20) to (18.center);
		\draw [style=none, bend right=15] (17.center) to (21);
		\draw [style=none, bend left=15] (20) to (16.center);
		\draw [style=none, bend left=15] (15.center) to (21);
		\draw (14.center) to (0.center);
		\draw (7.center) to (24);
		\draw (24) to (19.center);
		\draw [style=distribution] (31.center) to (29.center);
		\draw [style=distribution] (29.center) to (27.center);
		\draw [style=distribution] (27.center) to (25.center);
		\draw [style=distribution] (25.center) to (26.center);
		\draw [style=distribution] (26.center) to (28.center);
		\draw [style=distribution] (28.center) to (30.center);
		\draw [style=distribution] (30.center) to (31.center);
		\draw [style=distribution] (38.center) to (36.center);
		\draw [style=distribution] (36.center) to (34.center);
		\draw [style=distribution] (34.center) to (32.center);
		\draw [style=distribution] (32.center) to (33.center);
		\draw [style=distribution] (33.center) to (35.center);
		\draw [style=distribution] (35.center) to (37.center);
		\draw [style=distribution] (37.center) to (38.center);
		\draw [style=none, bend right=15] (45) to (43.center);
		\draw [style=none, bend right=15] (42.center) to (46);
		\draw [style=none, bend left=15] (45) to (41.center);
		\draw [style=none, bend left=15] (40.center) to (46);
		\draw (39.center) to (25.center);
		\draw (32.center) to (44.center);
	\end{pgfonlayer}
\end{tikzpicture}

%% file: content/zx.tex
\section{Applications: ZX-calculus}
\label{sec:zx}
In this section, we show an example application of the categorical constructions
of the previous sections.
In particular, we are interested in demonstrating how we can take the Distribution monad and
enrich the quantum categories of interest for reasoning about probabilistic processes
in quantum systems.
Most importantly, we show how the \emph{ZX-calculus}, a graphical calculus for reasoning
about quantum processes, can be appropriately extended to accommodate the extra
structure on said categories and how additional graphical rewrite rules capture the
interaction of probabilistic and deterministic quantum operations.
We begin with a general introduction to quantum computing and ZX-calculus, and then follow
with the enrichment of our categories of interest, together with the introduction of the
extended notation, and we finish by giving an example of how we can use this for diagrammatic
reasoning of noise in quantum systems.

\subsection{Quantum Computing}

When referring to quantum systems and operations, we have to make a distinction whenever we take
\emph{impure} operations into account.
In the pure states formalisms, quantum states are normalized vectors in a Hilbert space of
dimension $\bC^{2^n}$, with $n$ the number of \emph{qubits} (quantum bits) of the system.
It is common to use Dirac bra-ket notation to represent states, for example, some important
single-qubit states are $\ket 0 = \begin{bsmallmatrix} 1\\0 \end{bsmallmatrix},
\ket1 = \begin{bsmallmatrix} 0\\1 \end{bsmallmatrix},
\ket+ = \frac{1}{\sqrt{2}} \cdot (\ket0 + \ket1),
\ket- = \frac{1}{\sqrt{2}}\cdot (\ket0 - \ket1)$.
We operate on qubits by performing unitary transformations $U$ on the quantum states.
A multi-qubit quantum system with states $\ket\psi$ and $\ket\phi$ corresponds
to the tensor (Kronecker) product of the quantum states: $\ket\psi \otimes \ket\phi$.
We will represent the $n$-fold tensor product of a state $\ket\psi$ by $\ket{\psi^{n}}$.
Simultaneous (but independent) operations also follow from tensoring unitaries.

When we take into consideration the possibility of applying non-unitary operations we require
a more general framework, which is the \emph{density matrix} and \textit{completely positive maps}
formalism.
In this case, quantum states are positive semi-definite Hermitian matrices $\rho$ of trace one.
We write them as $\rho = \sum_i p_i \ketbra{\psi_i}{\psi_i}$
(where $\bra{\psi_i} = \ket{\psi_i}^\dagger$, for $\dagger$ the conjugate transpose), that is,
a \emph{statistical ensemble} of quantum states $\ket{\psi_i}$ (as density matrices) with probability $p_i$.
Operations on density matrices are completely positive (CP) maps of the form
$\Phi: \rho \to \sum_i K_i\rho K_i^\dagger$ with the condition $\sum_i K_iK_i^\dagger \le 1$
(notice how unitary maps fall inside this description too).
When we want to reason about quantum systems in the presence of noise, we then have to use the density matrix
and CP map formalism. For more information on quantum computing, we refer the reader to~\cite{nc00},
and for a more categorical introduction to~\cite{heunenvicary}.

\subsection{The ZX-calculus}\label{subsec:zx}

The ZX-calculus~\cite{cd11} is a graphical language for reasoning about quantum
states and processes as diagrams.
The language consists of a set of \textit{generators}, which are the green and
red\footnote{Light and dark in grayscale, respectively.} \textit{spiders}
(also called $Z$ and $X$ spiders), the
\textit{Hadamard box}, the \textit{identity wire}, the \textit{swap}, the \textit{cup},
the \textit{cap}, and the \textit{empty diagram}.
In \autoref{fig:generators} we can see the generators of the ZX-calculus and their signature,
with input wire(s) coming from the top and outputs going to the bottom.
Spiders have a \textit{phase} $\alpha\in [0,2\pi)$,
which as we will see later is omitted when $\alpha=0$.
We can also see how to sequentially compose ($\circ$) arbitrary diagrams by connecting inputs
with outputs, and how to parallely compose diagrams
(as a tensor product $\otimes$) by placing them side by side. 

% Different ways of importing a tikzit figure
% \ctikzfig{paper figures/fig1}

%\begin{equation}\label{d1}
%  \tikzfig{d1}
%\end{equation}

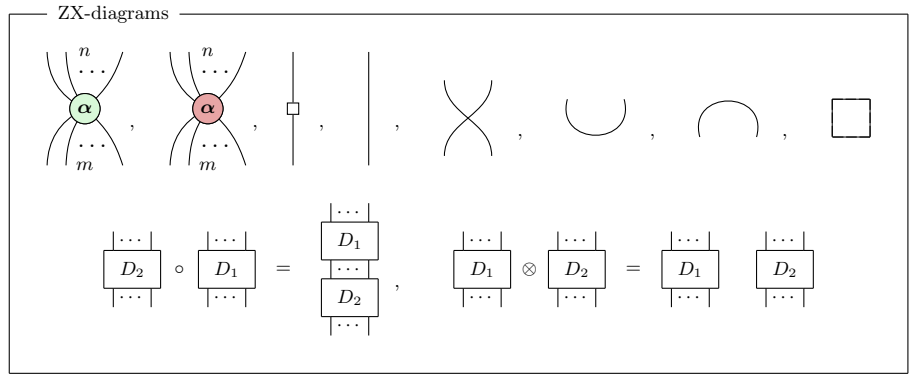
\begin{figure}
  \centering
  \resizebox{\columnwidth}{!}{%
    \input{figures/fig1.tikz}%
    }
  \caption{ZX-diagrams generators and how to compose them.}
  \label{fig:generators}
\end{figure}

Each of the generators has a \textit{standard interpretation} $\llbracket\cdot \rrbracket$
as a linear map in $\bC^{2^n}$ that we can find in \autoref{fig:zxInterpretation}.

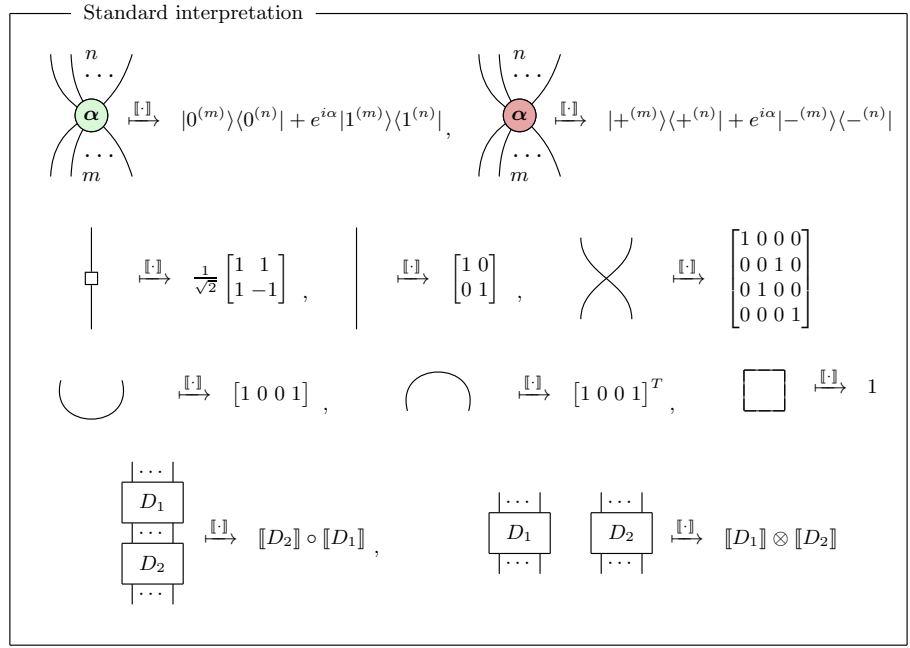
\begin{figure}
  \centering
  \resizebox{\columnwidth}{!}{%
    \input{figures/fig2.tikz}%
    }
  \caption{Standard interpretation of ZX-diagrams.}
  \label{fig:zxInterpretation}
\end{figure}

Categorically, ZX-diagrams form the category \textbf{ZX} with $|\mathbf{ZX}|=\mathbb{N}$
(where some $n\in\mathbb{N}$ is the number of wires,
which we can think of as an $n$-qubit quantum system) and morphisms being the generators.
The standard interpretation is a (monoidal) functor
$\llbracket \cdot \rrbracket: \mathbf{ZX}\to \qbit$ that acts on objects as
$\llbracket n \rrbracket=n$ and on morphisms as defined in
\autoref{fig:zxInterpretation}~\cite{johnSurvey}.

ZX-diagrams come with a set of \textit{rewrite rules} that form the ZX-calculus.
These rewrite rules let us transform a diagram into a different one while preserving the semantics
(i.e. the interpretation).
We have collected the rules in \autoref{fig:rules}.
There is also an important additional rule that can be summarized as the
\textit{only connectivity matters} rule, which states that we can deform diagrams at will without
changing their meaning, as long as we maintain the connectivity between the generators unchanged.
For a thorough explanation of each rule we refer the reader to~\cite{cd11,vilmartOptimal}.

\begin{figure}
  \centering
  \resizebox{\columnwidth}{!}{%
    \input{figures/ruleset.tikz}%
    }
    \caption{ZX-calculus ruleset.
      All rules also hold when swapping the colors of the spiders.
      In $(eu)$ we omit the calculation of the angles, which can be found
      in~\cite{vilmartOptimal}.}
  \label{fig:rules}
\end{figure}
The ZX-calculus satisfies important properties. ZX-diagrams are \textit{universal},
meaning that any linear map $f$ of the form $f: \bC^{2^n}\to\bC^{2^m}$
can be represented as a ZX-diagram.
The rewrite rules are \textit{sound}, meaning that they do not change the interpretation
of the diagram as a linear map.
They are also \textit{complete}, which ensures that if two diagrams have the same interpretation,
the ruleset is powerful enough to always let us transform one diagram into the other.
These properties ensure that the ZX-calculus can be used as a tool for reasoning about
quantum computing, as it has been already demonstrated in tasks such as quantum circuit
optimisation~\cite{tReduction}, verification of quantum circuits~\cite{Peham2022},
simulation~\cite{simulation}, and as a reasoning tool~\cite{deBeaudrap2020,kissinger2022phasefree}.

In~\eqref{eq:zxExample} we have example one- and two-qubit gates as ZX-diagrams.
We also see the computational basis $\{|0\rangle,|1\rangle\}$ and Hadamard basis
$\{|+\rangle,|-\rangle\}$ states.

\begin{equation}\label{eq:zxExample}
  \scalebox{0.8}{\input{figures/zxExample.tikz}}
\end{equation}

\subsection{Enriching the Categories $\qbit$ and $\cpm$}
\label{sec:eqbit}

Our motivation is to highlight certain types of relevant physical phenomena
(probabilistic processes) that are present in quantum systems within our categories.
It is then natural to use the Distribution monad $\md$ together with the construction
explained in the previous sections to enrich our categories for quantum reasoning.

Indeed, we take $\cpm$ and perform a free enrichment over $\md$.
What we get is the category $(F_*\cpm, \etens, I)$ consisting of the same objects as $\cpm$
and morphisms (incl. identity) for objects $A,B$ the free algebras over $\md$ of the
hom-set $\cpm(A,B)$.
Composites of morphisms are the free algebra over the composite in $\cpm$,
and the SMC structure is preserved thanks to \Cref{co:free}.

We also define the non-freely enriched category $\ecpm$ so we can
interpret probability distributions as CP maps.
For this, we define SMC-structure in the non-free $(\emd,\otimes^\md,\md(\timesUnit))$,
with a tensor product of algebras
defined by the coequalizer~\eqref{eq:coeq}, with a more detailed description
in \Cref{appMonoidalAlg}.
The category $(\ecpm,\odot,\odotUnit)$ has the same objects as $\cpm$ and for every pair of
objects $A,B$ the hom-object is an algebra $(\cpm(A,B),\alpha)$
with the $\md$-action turning a formal convex sum of linear maps into an actual sum by scalar
multiplication and addition.
Composition of hom-objects follows from composition in $\cpm$, and for an object $A$
the identity element is $j_A:(\timesUnit,\alpha)\to(\operatorname{Id}_A,\alpha)$.
We define now its symmetric monoidal structure following the definition of enriched SMC from
the beginning of \Cref{sec:v-monoidal}.
The tensor product $\odot$ on objects is the same as in $\cpm$, and on hom-objects it is
the tensor product in $\cpm$ to the underlying sets:
$\odot: (\cpm(A,A'),\alpha)\otimes^\md (\cpm(B,B'),\alpha)$ $\to$
$ (\cpm(A\otimes B, A'\otimes B'),\alpha)$.
The unit $U$ is the one in $\cpm$.
The associator, unitors, and symmetry all follow from applying the ones in $\cpm$.

In the following sections, we will interpret ZX-diagrams into $F_*\cpm$
as probability distributions of CP maps. From there, to interpret
probability distributions as CP maps, we define the functor
$\llangle\cdot\rrangle: F_*\cpm\to\ecpm$ that sends objects to themselves and applies
the monad algebra to hom-objects i.e. we ``evaluate'' a probability distribution
over CP maps by multiplying the probabilities with the corresponding map and then
adding all maps together.

Technically, we can also enrich $\qbit$ in the same way as we did with $\cpm$,
but density matrices and CP maps are the more sensible choices to talk about
probabilistic mixtures of operations.
On the other hand, enriching $\qbit$ (or $\cpm$) over algebras of the multiset monad $\mm$
leads to an enrichment over commutative monoids that exposes
addition of linear maps~\cite{heunenvicary}.
This was recently formulated in~\cite{muuss-thesis,stollenwerk-diagrammatic-2022}
as a way to ``split'' parameterised Pauli rotation
gates in ZX-calculus in such a way that the parameter relocates from its place inside the spider as a phase
to a scalar on a wire using the identity
$e^{i\alpha P} = \cos\alpha I + i \sin\alpha P$ for $P$ a Pauli matrix (or any matrix satisfying
$P^2=I$).

\subsection{Enriched ZX-diagrams and Their Interpretation}

In the same manner as the ZX-calculus is a language for reasoning in $\qbit$,
we can create a graphical language with extra structure to reason in our enriched categories.
Since we are going to be enriching $\cpm$, we first need to see how to turn the ZX-calculus
into a graphical language for CP maps.
This is done straightforwardly by adding a \emph{discard} operation $\gnd$
to the list of generators of \autoref{fig:generators} plus additional rewrite rules
(that we choose to omit here) stating that isometries can be discarded~\cite{caretteGround}.
The interpretation  of a ZX-diagram $D$ as a CP map is then a superoperator
$\rho \to \llbracket D \rrbracket\ \rho\ \llbracket D \rrbracket^\dagger$,
for $\llbracket D \rrbracket$ the standard interpretation of $D$
as in \autoref{fig:zxInterpretation}~\cite{borgna-hybrid-2021}.

We then construct an enriched graphical language for $\ecpm$ by building on top of
the ZX-calculus for CP maps. The notation will be similar to the running examples
we have given throughout the text
(cf.~\ref{fig:probabilistic-mix},\ref{eq:example1},\ref{eq:example2}).
The main idea is as follows.
We take the generators of the ZX-calculus and allow them to be freely wrapped between
opening $\scalebox{0.8}{\input{figures/openingbracket.tikz}}$
and closing $\scalebox{0.8}{\input{figures/closingbracket.tikz}}$
\emph{distribution brackets}.

Intuitively, we interpret diagrams that are within distribution brackets as a probabilistic
mixture of operations: diagrams placed side by side correspond to different probabilistic
choices with some weight attached to the corresponding wires.
Within each choice, sequential (and as we will see later, parallel) composition is allowed.
The main difference to the usual graphical languages for monoidal categories
is that the parallel composition of each choice does not correspond
to the tensor product.
In a way, we also subsume ZX-diagrams by drawing diagrams that are not enclosed by distribution
brackets, which are then interpreted as an operation that occurs with probability $1$.

For example, we can represent the single-qubit \emph{depolarizing channel}~\cite{nc00}
$\Phi:\rho\mapsto (1-p)\rho + \frac{p}{3}(X\rho X + Y\rho Y + Z\rho Z)$
that leaves a quantum state $\rho$ unchanged with probability $(1-p)$ or applies
an $X,Y$ or $Z$ error with probability $\frac{p}{3}$ each with the diagram on the left in
\Cref{fig:zx-examples}.

\begin{figure}
  \centering
  \scalebox{0.9}{%
    \input{figures/zx-examples.tikz}%
    }
    \caption{\textbf{Left:} Diagrammatic representation of the depolarizing channel.\\
      \textbf{Right:} Diagrammatic representation of a mixture of two-qubit gates.}
  \label{fig:zx-examples}
\end{figure}
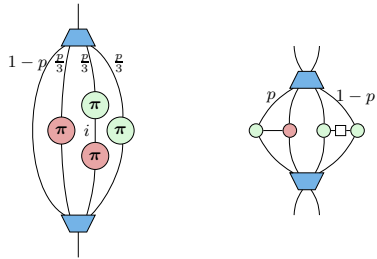

We need to take extra care when handling scalars inside the brackets.
Indeed, what we have inside distribution brackets is a formal convex sum of ZX-diagrams
(or, in the general case, string diagrams),
meaning that the SMC rewriting axioms apply to each summand independently.
Since summands are also juxtaposed, it might seem like this notation allows for the transfer
of scalars from one summand to another.
The crux is that, since what is enclosed by trapezoids is a formal sum, we cannot drag
scalars from one summand to another using those same monoidal category axioms.
This means that we can consider the probabilities
(and any scalar factor if present, such as the imaginary unit in~\Cref{fig:zx-examples})
to be bound to the wires themselves, and only interact with the ZX-diagrams
(or generally string diagrams) that belong to that summand.
An alternative is to encapsulate each summand in ``bubbles'' for stronger visual separation~\cite{stollenwerk-diagrammatic-2022,muuss-thesis}.

A probabilistic mixture of operations with multiple inputs or outputs looks similar to
the $1$-to-$1$ case, with the caveat that we need to be more careful in the positioning of the
wires as to distinguish between tensor product and probabilistic choice.\footnote{One could use \emph{scalable notation}~\cite{carette-szx-calculus}
  to allow wires to be multi-qubit quantum registers.
  This would help with the distinction when diagrams are larger in practice. }
For example, if we want to represent applying the CNOT gate and the $CZ$ gate with probabilities
$p$ and $1-p$ respectively we would get the diagram on the right in~\Cref{fig:zx-examples}.

We can consider this extra notation as the result of a free enrichment of $\zx$ over $\emd$
giving us the category of enriched ZX-diagrams $F_*\zx$.
We can then define the interpretation $\llbracket\cdot\rrbracket_\md$ of an $\emd$-enriched
diagram as a monoidal functor from the category of enriched ZX-diagrams to $\ecpm$.
This functor factors through $F_*\cpm$ as follows:
% https://q.uiver.app/#q=WzAsMyxbMCwwLCJGXypcXHp4Il0sWzEsMSwiXFxlY3BtIl0sWzEsMCwiRl8qXFxjcG0iXSxbMCwxLCJcXGxsYnJhY2tldCBcXGNkb3RcXHJyYnJhY2tldF9cXG1kIiwyXSxbMiwxLCJcXGxsYW5nbGUgXFxjZG90XFxycmFuZ2xlIl0sWzAsMiwiXFxsbGFuZ2xlIFxcY2RvdFxccnJhbmdsZV8qIl1d
\[\begin{tikzcd}[sep=small]
    {F_*\zx} & {F_*\cpm} \\
    & \ecpm
    \arrow["{\llbracket \cdot\rrbracket_\md}"', from=1-1, to=2-2]
    \arrow["{\llangle \cdot\rrangle}", from=1-2, to=2-2]
    \arrow["{\llangle \cdot\rrangle_*}", from=1-1, to=1-2]
  \end{tikzcd}\]

Where $\llangle \cdot\rrangle_*$ interprets an enriched ZX-diagram as a probabilistic
mixture of operations which is then evaluated by $\llangle \cdot\rrangle$ as explained
in \Cref{sec:eqbit}.
An example of the interpretation of an arbitrary distribution of ZX-diagrams of arbitrary
size can be seen in~\eqref{eq:interpretation}.
When using the multiset monad $\mm$ instead the interpretation $\llbracket \cdot\rrbracket_\mm$
is similar.

\begin{equation}\label{eq:interpretation}
  \scalebox{0.85}{\input{figures/interpretation.tikz}}
\end{equation}

Enriched ZX-diagrams are universal, that is, any morphism in $\ecpm$ can be represented
by an enriched ZX-diagram.
Indeed, since $\ecpm$ is still made of CP maps between Hilbert spaces,
we can use universality of the ZX-calculus alone to represent any morphism in $\ecpm$.

\subsection{Additional Rules for Enriched ZX-diagrams}

With the new notation we can have new rewrite rules too, some of which were already
introduced in~\cite{stollenwerk-diagrammatic-2022,muuss-thesis}
(\textbf{(es)}, \textbf{(ep)}, \textbf{(ec)}, and \textbf{(e$\delta$)}) for the case of
linear combinations.
We will display them here, including additional rules.
The ruleset of the enriched ZX-calculus for the distribution and multiset monads is the same as
the one for ZX-calculus plus additional rules that capture the interaction between sums,
products, tensor products, and scalars.
We can see the additional rules arising from the enrichment in~\Cref{fig:enriched-rules}, which
intuitively state:
\begin{figure}
  \centering
  \resizebox{\columnwidth}{!}{%
    \input{figures/enrichedRules.tikz}%
    }
    \caption{Additional rules for the enriched ZX-calculus, alongside the ones of \Cref{fig:rules}.
    Diagrams $D, D'$ are arbitrary ZX-diagrams and weights $p,q$ are probabilities.}
  \label{fig:enriched-rules}
\end{figure}
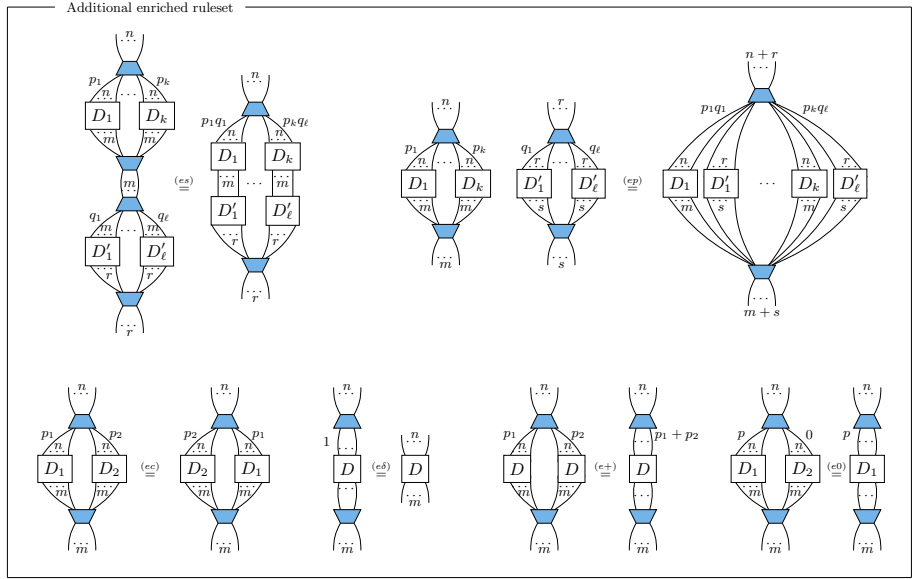
\begin{itemize}
\item \textbf{(es):} The enriched sequential composition rule shows how to sequentially
  compose distributions. Intuitively this rule follows from products distributing over addition.
\item \textbf{(ep):} The enriched parallel composition rule is the same as \textbf{(es)}, but
  for parallel composition instead of sequential.
\item \textbf{(ec):} The enriched commutativity rule shows that bracketed diagrams are invariant
  under permutation of the branches.
\item \textbf{(e$\delta$):} The enriched Dirac delta distribution rule provides a shorthand for the trivial Dirac delta distribution.
\item \textbf{(e$+$):} The enriched addition rule shows that we can remove a branch if it is identical
  to some other by adding the probabilities.
\item \textbf{(e$0$):} The enriched $0$-probability rule allows us to remove branches with $0$ probability .
\end{itemize}

Rules \textbf{(es)},\textbf{(ep)},\textbf{(ec)} and \textbf{(e$\delta$)} were proven to be sound
in~\cite{muuss-thesis} but in the context of linear combinations of diagrams interpreted in $\qbit$.
We show that these rules still hold as an enrichment in $\emd$ and interpreted in $\ecpm$
in~\Cref{sec:appendixSoundness}.
Finding a complete ruleset (i.e. one that can show $D_1 = D_2$ whenever
$\llbracket D_1 \rrbracket_\md = \llbracket D_2 \rrbracket_\md$) for enriched diagrams remains
to be done.
A possible direction to tackle this problem would be to translate enriched diagrams into
ZXW~\cite{shaikh-how-2022} diagrams, which is a complete diagrammatic language with a
\emph{W-spider} that can encode addition of phases.
Another alternative would be to translate into the controlled form of~\cite{jeandel-addition-2023}.

We conclude with a demonstration of how we can use this extension of the ZX-calculus
to study the effectiveness of Quantum Error Mitigation (QEM) techniques for different noise models.
Quantum Error Mitigation~\cite{cai-quantum-2023} are the series of techniques that are used to reduce the
effects of noise in near-term quantum systems.
One such technique is \emph{Symmetry Verification}~\cite{bonet-monroig-2018},
which states that given a Hamiltonian
(Hermitian operator that determines the evolution of a system) $\hat{H}$,
and a symmetry $S$ (an operator that commutes with $\hat{H}$ i.e. $[\hat{H},S] = \hat{H}S - S\hat{H} = 0$),
one can perform measurements of $S$ to verify if the state that (ideally) evolves under $\hat{H}$
was affected by errors.
Indeed, under the assumption that the initial state is a $(+1)$ eigenvector of $S$, then it will
stay that way under ideal evolution under $\hat{H}$.
This implies that if there is an error $E$ that anti-commutes with $S$ (i.e. $\{E,S\} = ES + SE = 0$)
at some point in the computation, we can measure $S$ to detect a change in the eigenvalue.
Symmetry verification then proposes to perform a postselection on the result $(+1)$,
meaning that we discard computations that give a $(-1)$ outcome when measuring $S$.

Given a noisy state $\rho_{\text{noisy}}$ and a symmetry $S$, the probability of outcome $(+1)$
when measuring $S$ is given by $p(+1)=\tr(P_{+1}\rho_{noisy})$, for $P_{+1} = \frac{I+S}{2}$ the projector
onto the $(+1)$ eigenspace of $S$ and $\tr$ the trace operator.
This value tells us then with which probability the measurement ``accepts'' a noisy state, and can
be used to compare the effectiveness of different choices of $S$ given a certain noise
model~\cite{Kakkar-2022}.
Let us consider $\rho_{\text{noisy}} = \Phi(U\ket{0})$ for $\Phi$ the depolarizing noise channel
and $U$ some single-qubit unitary -- in other words, we have a single layer of depolarizing
noise at the end of our computation.
For simplicity, let us further assume that our state before the depolarizing channel
is the $(+1)$ eigenvector of some Pauli operator e.g. the Pauli $X$, then we have $U=H$
(the Hadamard gate) and we can draw $p(+1)=\tr(\frac{I+S}{2}\rho_{noisy})$ diagrammatically
(up to scalar factor, see~\Cref{app:SV}) as the following diagram:

\begin{equation*}\label{eq:exampleSV}
  \scalebox{0.85}{\input{figures/exampleSV.tikz}}
\end{equation*}

From top to bottom, the diagram represents applying $H$ to the $\ket{0}$ state, followed by
a depolarizing noise channel and the verification of $X$ in the form of a CNOT gate controlled
on an auxiliary qubit.
The auxiliary qubit on the right starts in the $\ket{0}$ state and has a Hadamard gate applied
to it before and after the CNOT.
It is then postselected into $\bra{0}$, which is the corresponding state for the $(+1)$
outcome.
The last operation in the form of $\gnd$ corresponds to the trace.
The full diagrammatic calculation is in~\Cref{app:SV}.
With similar diagrams, we can study diagrammatically how well different QEM techniques mitigate
certain noise models, and apply them to representations of quantum algorithms that, for example,
have one layer of errors for every time step.

%%% Local Variables:
%%% mode: latex
%%% TeX-master: "../enriched-diagrams"
%%% End:

%% file: figures/fig1.tikz
\begin{tikzpicture}
	\begin{pgfonlayer}{nodelayer}
		\node [style=none] (198) at (2.5, 2.75) {};
		\node [style=none] (199) at (4, 2.75) {};
		\node [style=none] (200) at (-11.25, 4) {};
		\node [style=none] (201) at (-9.25, 4) {};
		\node [style=none] (202) at (-10.75, 4) {};
		\node [style=none] (203) at (-10, 3.475) {$\dots$};
		\node [style=none] (204) at (-11.25, 1) {};
		\node [style=none] (205) at (-9.25, 1) {};
		\node [style=none] (206) at (-10.75, 1) {};
		\node [style=none] (207) at (-10, 1.525) {$\dots$};
		\node [style=Green Node] (208) at (-10.25, 2.5) {};
		\node [style=Green Node] (211) at (-10.25, 2.5) {$\alpha$};
		\node [style=none] (212) at (-8, 1) {};
		\node [style=none] (213) at (-6, 1) {};
		\node [style=none] (214) at (-7.5, 1) {};
		\node [style=none] (215) at (-6.75, 1.525) {$\dots$};
		\node [style=none] (216) at (-8, 4) {};
		\node [style=none] (217) at (-6, 4) {};
		\node [style=none] (218) at (-7.5, 4) {};
		\node [style=none] (219) at (-6.75, 3.475) {$\dots$};
		\node [style=Red Node] (220) at (-7, 2.5) {};
		\node [style=Red Node] (221) at (-7, 2.5) {$\alpha$};
		\node [style=none] (222) at (-4.75, 1) {};
		\node [style=none] (223) at (-4.75, 4) {};
		\node [style=h] (224) at (-4.75, 2.5) {};
		\node [style=none] (225) at (6, 1.75) {};
		\node [style=none] (226) at (7.5, 1.75) {};
		\node [style=none] (227) at (-2.75, 1) {};
		\node [style=none] (228) at (-2.75, 4) {};
		\node [style=none] (229) at (-0.75, 3.25) {};
		\node [style=none] (230) at (-0.75, 1.25) {};
		\node [style=none] (231) at (0.5, 3.25) {};
		\node [style=none] (232) at (0.5, 1.25) {};
		\node [style=bigtxt] (234) at (-9, 2) {,};
		\node [style=bigtxt] (235) at (-5.75, 2) {,};
		\node [style=bigtxt] (236) at (-4, 2) {,};
		\node [style=bigtxt] (237) at (1.25, 1.75) {,};
		\node [style=bigtxt] (238) at (4.75, 1.75) {,};
		\node [style=bigtxt] (241) at (8.25, 1.75) {,};
		\node [style=none] (242) at (9.5, 2.75) {};
		\node [style=none] (243) at (10.5, 2.75) {};
		\node [style=none] (244) at (9.5, 1.75) {};
		\node [style=none] (245) at (10.5, 1.75) {};
		\node [style=txtnobold] (247) at (-10.25, 1) {$m$};
		\node [style=txtnobold] (249) at (-7, 1) {$m$};
		\node [style=txtnobold] (251) at (-10.25, 4) {$n$};
		\node [style=txtnobold] (253) at (-7, 4) {$n$};
		\node [style=bigtxt] (254) at (-9.5, 5) {ZX-diagrams};
		\node [style=none] (255) at (-12.25, 5) {};
		\node [style=none] (256) at (-12.25, -4.5) {};
		\node [style=none] (257) at (-11.25, 5) {};
		\node [style=none] (258) at (11.5, 5) {};
		\node [style=none] (259) at (11.5, -4.5) {};
		\node [style=none] (260) at (-7.75, 5) {};
		\node [style=bigtxt] (261) at (-2, 2) {,};
		\node [style=none] (262) at (-3.75, 0) {};
		\node [style=none] (263) at (-2.75, 0) {};
		\node [style=txtnobold] (264) at (-3.25, -0.25) {$\cdots$};
		\node [style=none] (265) at (-2.75, -0.5) {};
		\node [style=none] (266) at (-3.75, -0.5) {};
		\node [style=none] (267) at (-4, -0.5) {};
		\node [style=none] (268) at (-2.5, -0.5) {};
		\node [style=none] (269) at (-4, -1.5) {};
		\node [style=none] (270) at (-2.5, -1.5) {};
		\node [style=none] (271) at (-3.75, -1.5) {};
		\node [style=none] (272) at (-2.75, -1.5) {};
		\node [style=none] (273) at (-3.75, -1.5) {};
		\node [style=none] (274) at (-2.75, -1.5) {};
		\node [style=txtnobold] (275) at (-3.25, -1.75) {$\cdots$};
		\node [style=none] (276) at (-2.75, -2) {};
		\node [style=none] (277) at (-3.75, -2) {};
		\node [style=none] (278) at (-4, -2) {};
		\node [style=none] (279) at (-2.5, -2) {};
		\node [style=none] (280) at (-4, -3) {};
		\node [style=none] (281) at (-2.5, -3) {};
		\node [style=none] (282) at (-3.75, -3) {};
		\node [style=none] (283) at (-2.75, -3) {};
		\node [style=none] (284) at (-3.75, -3.5) {};
		\node [style=none] (285) at (-2.75, -3.5) {};
		\node [style=txtnobold] (286) at (-3.25, -3.25) {$\cdots$};
		\node [style=txtnobold] (287) at (-3.25, -1) {$D_1$};
		\node [style=txtnobold] (288) at (-3.25, -2.5) {$D_2$};
		\node [style=none] (289) at (-0.25, -0.75) {};
		\node [style=none] (290) at (0.75, -0.75) {};
		\node [style=txtnobold] (291) at (0.25, -1) {$\cdots$};
		\node [style=none] (292) at (0.75, -1.25) {};
		\node [style=none] (293) at (-0.25, -1.25) {};
		\node [style=none] (294) at (-0.5, -1.25) {};
		\node [style=none] (295) at (1, -1.25) {};
		\node [style=none] (296) at (-0.5, -2.25) {};
		\node [style=none] (297) at (1, -2.25) {};
		\node [style=none] (298) at (-0.25, -2.25) {};
		\node [style=none] (299) at (0.75, -2.25) {};
		\node [style=none] (300) at (-0.25, -2.25) {};
		\node [style=none] (301) at (0.75, -2.25) {};
		\node [style=txtnobold] (302) at (0.25, -2.5) {$\cdots$};
		\node [style=none] (303) at (3.25, -1.25) {};
		\node [style=none] (304) at (2.25, -1.25) {};
		\node [style=none] (305) at (2, -1.25) {};
		\node [style=none] (306) at (3.5, -1.25) {};
		\node [style=none] (307) at (2, -2.25) {};
		\node [style=none] (308) at (3.5, -2.25) {};
		\node [style=none] (309) at (2.25, -2.25) {};
		\node [style=none] (310) at (3.25, -2.25) {};
		\node [style=none] (311) at (2.25, -2.75) {};
		\node [style=none] (312) at (3.25, -2.75) {};
		\node [style=txtnobold] (313) at (2.75, -2.5) {$\cdots$};
		\node [style=txtnobold] (314) at (0.25, -1.75) {$D_1$};
		\node [style=txtnobold] (315) at (2.75, -1.75) {$D_2$};
		\node [style=none] (316) at (-0.25, -2.75) {};
		\node [style=none] (317) at (0.75, -2.75) {};
		\node [style=none] (318) at (2.25, -0.75) {};
		\node [style=none] (319) at (3.25, -0.75) {};
		\node [style=txtnobold] (320) at (2.75, -1) {$\cdots$};
		\node [style=none] (321) at (-7, -0.75) {};
		\node [style=none] (322) at (-6, -0.75) {};
		\node [style=txtnobold] (323) at (-6.5, -1) {$\cdots$};
		\node [style=none] (324) at (-6, -1.25) {};
		\node [style=none] (325) at (-7, -1.25) {};
		\node [style=none] (326) at (-7.25, -1.25) {};
		\node [style=none] (327) at (-5.75, -1.25) {};
		\node [style=none] (328) at (-7.25, -2.25) {};
		\node [style=none] (329) at (-5.75, -2.25) {};
		\node [style=none] (330) at (-7, -2.25) {};
		\node [style=none] (331) at (-6, -2.25) {};
		\node [style=none] (332) at (-7, -2.25) {};
		\node [style=none] (333) at (-6, -2.25) {};
		\node [style=txtnobold] (334) at (-6.5, -2.5) {$\cdots$};
		\node [style=none] (335) at (-8.5, -1.25) {};
		\node [style=none] (336) at (-9.5, -1.25) {};
		\node [style=none] (337) at (-9.75, -1.25) {};
		\node [style=none] (338) at (-8.25, -1.25) {};
		\node [style=none] (339) at (-9.75, -2.25) {};
		\node [style=none] (340) at (-8.25, -2.25) {};
		\node [style=none] (341) at (-9.5, -2.25) {};
		\node [style=none] (342) at (-8.5, -2.25) {};
		\node [style=none] (343) at (-9.5, -2.75) {};
		\node [style=none] (344) at (-8.5, -2.75) {};
		\node [style=txtnobold] (345) at (-9, -2.5) {$\cdots$};
		\node [style=txtnobold] (346) at (-6.5, -1.75) {$D_1$};
		\node [style=txtnobold] (347) at (-9, -1.75) {$D_2$};
		\node [style=none] (348) at (-6, -2.75) {};
		\node [style=none] (349) at (-8.5, -0.75) {};
		\node [style=none] (350) at (-7, -2.75) {};
		\node [style=none] (351) at (-9.5, -0.75) {};
		\node [style=txtnobold] (352) at (-7.75, -1.75) {$\circ$};
		\node [style=txtnobold] (353) at (1.5, -1.75) {$\otimes$};
		\node [style=txtnobold] (354) at (-5, -1.75) {$=$};
		\node [style=txtnobold] (355) at (4.25, -1.75) {$=$};
		\node [style=none] (356) at (5.25, -0.75) {};
		\node [style=none] (357) at (6.25, -0.75) {};
		\node [style=txtnobold] (358) at (5.75, -1) {$\cdots$};
		\node [style=none] (359) at (6.25, -1.25) {};
		\node [style=none] (360) at (5.25, -1.25) {};
		\node [style=none] (361) at (5, -1.25) {};
		\node [style=none] (362) at (6.5, -1.25) {};
		\node [style=none] (363) at (5, -2.25) {};
		\node [style=none] (364) at (6.5, -2.25) {};
		\node [style=none] (365) at (5.25, -2.25) {};
		\node [style=none] (366) at (6.25, -2.25) {};
		\node [style=none] (367) at (5.25, -2.25) {};
		\node [style=none] (368) at (6.25, -2.25) {};
		\node [style=txtnobold] (369) at (5.75, -2.5) {$\cdots$};
		\node [style=none] (370) at (8.75, -1.25) {};
		\node [style=none] (371) at (7.75, -1.25) {};
		\node [style=none] (372) at (7.5, -1.25) {};
		\node [style=none] (373) at (9, -1.25) {};
		\node [style=none] (374) at (7.5, -2.25) {};
		\node [style=none] (375) at (9, -2.25) {};
		\node [style=none] (376) at (7.75, -2.25) {};
		\node [style=none] (377) at (8.75, -2.25) {};
		\node [style=none] (378) at (7.75, -2.75) {};
		\node [style=none] (379) at (8.75, -2.75) {};
		\node [style=txtnobold] (380) at (8.25, -2.5) {$\cdots$};
		\node [style=txtnobold] (381) at (5.75, -1.75) {$D_1$};
		\node [style=txtnobold] (382) at (8.25, -1.75) {$D_2$};
		\node [style=none] (383) at (5.25, -2.75) {};
		\node [style=none] (384) at (6.25, -2.75) {};
		\node [style=none] (385) at (7.75, -0.75) {};
		\node [style=none] (386) at (8.75, -0.75) {};
		\node [style=txtnobold] (387) at (8.25, -1) {$\cdots$};
		\node [style=txtnobold] (388) at (-9, -1) {$\cdots$};
		\node [style=bigtxt] (389) at (-2, -2.25) {,};
	\end{pgfonlayer}
	\begin{pgfonlayer}{edgelayer}
		\draw [style=none, bend right=105, looseness=2.25] (198.center) to (199.center);
		\draw [style=none, bend left] (208) to (200.center);
		\draw [style=none, bend left=15] (208) to (202.center);
		\draw [style=none, bend right=15] (208) to (201.center);
		\draw [style=none, bend right] (208) to (204.center);
		\draw [style=none, bend right=15] (208) to (206.center);
		\draw [style=none, bend left=15] (208) to (205.center);
		\draw [style=none, bend right] (220) to (212.center);
		\draw [style=none, bend right=15] (220) to (214.center);
		\draw [style=none, bend left=15] (220) to (213.center);
		\draw [style=none, bend left] (220) to (216.center);
		\draw [style=none, bend left=15] (220) to (218.center);
		\draw [style=none, bend right=15] (220) to (217.center);
		\draw [style=none] (222.center) to (224);
		\draw [style=none] (224) to (223.center);
		\draw [style=none, bend left=105, looseness=2.25] (225.center) to (226.center);
		\draw [style=none] (227.center) to (228.center);
		\draw [style=none, in=-90, out=90] (230.center) to (231.center);
		\draw [style=none, in=-90, out=90] (232.center) to (229.center);
		\draw [style=discontinuous] (242.center) to (243.center);
		\draw [style=discontinuous] (243.center) to (245.center);
		\draw [style=discontinuous] (245.center) to (244.center);
		\draw [style=discontinuous] (242.center) to (244.center);
		\draw (255.center) to (256.center);
		\draw (255.center) to (257.center);
		\draw (258.center) to (259.center);
		\draw (258.center) to (260.center);
		\draw (256.center) to (259.center);
		\draw [style=none] (262.center) to (266.center);
		\draw [style=none] (263.center) to (265.center);
		\draw [style=none] (267.center) to (266.center);
		\draw [style=none] (266.center) to (265.center);
		\draw [style=none] (265.center) to (268.center);
		\draw [style=none] (268.center) to (270.center);
		\draw [style=none] (270.center) to (272.center);
		\draw [style=none] (272.center) to (271.center);
		\draw [style=none] (271.center) to (269.center);
		\draw [style=none] (269.center) to (267.center);
		\draw [style=none] (273.center) to (277.center);
		\draw [style=none] (274.center) to (276.center);
		\draw [style=none] (278.center) to (277.center);
		\draw [style=none] (277.center) to (276.center);
		\draw [style=none] (276.center) to (279.center);
		\draw [style=none] (279.center) to (281.center);
		\draw [style=none] (281.center) to (283.center);
		\draw [style=none] (283.center) to (282.center);
		\draw [style=none] (282.center) to (280.center);
		\draw [style=none] (280.center) to (278.center);
		\draw [style=none] (283.center) to (285.center);
		\draw [style=none] (282.center) to (284.center);
		\draw [style=none] (289.center) to (293.center);
		\draw [style=none] (290.center) to (292.center);
		\draw [style=none] (294.center) to (293.center);
		\draw [style=none] (293.center) to (292.center);
		\draw [style=none] (292.center) to (295.center);
		\draw [style=none] (295.center) to (297.center);
		\draw [style=none] (297.center) to (299.center);
		\draw [style=none] (299.center) to (298.center);
		\draw [style=none] (298.center) to (296.center);
		\draw [style=none] (296.center) to (294.center);
		\draw [style=none] (305.center) to (304.center);
		\draw [style=none] (304.center) to (303.center);
		\draw [style=none] (303.center) to (306.center);
		\draw [style=none] (306.center) to (308.center);
		\draw [style=none] (308.center) to (310.center);
		\draw [style=none] (310.center) to (309.center);
		\draw [style=none] (309.center) to (307.center);
		\draw [style=none] (307.center) to (305.center);
		\draw [style=none] (310.center) to (312.center);
		\draw [style=none] (309.center) to (311.center);
		\draw [style=none] (318.center) to (304.center);
		\draw [style=none] (319.center) to (303.center);
		\draw [style=none] (300.center) to (316.center);
		\draw [style=none] (301.center) to (317.center);
		\draw [style=none] (321.center) to (325.center);
		\draw [style=none] (322.center) to (324.center);
		\draw [style=none] (326.center) to (325.center);
		\draw [style=none] (325.center) to (324.center);
		\draw [style=none] (324.center) to (327.center);
		\draw [style=none] (327.center) to (329.center);
		\draw [style=none] (329.center) to (331.center);
		\draw [style=none] (331.center) to (330.center);
		\draw [style=none] (330.center) to (328.center);
		\draw [style=none] (328.center) to (326.center);
		\draw [style=none] (337.center) to (336.center);
		\draw [style=none] (336.center) to (335.center);
		\draw [style=none] (335.center) to (338.center);
		\draw [style=none] (338.center) to (340.center);
		\draw [style=none] (340.center) to (342.center);
		\draw [style=none] (342.center) to (341.center);
		\draw [style=none] (341.center) to (339.center);
		\draw [style=none] (339.center) to (337.center);
		\draw [style=none] (342.center) to (344.center);
		\draw [style=none] (341.center) to (343.center);
		\draw (333.center) to (348.center);
		\draw (332.center) to (350.center);
		\draw (351.center) to (336.center);
		\draw (349.center) to (335.center);
		\draw [style=none] (356.center) to (360.center);
		\draw [style=none] (357.center) to (359.center);
		\draw [style=none] (361.center) to (360.center);
		\draw [style=none] (360.center) to (359.center);
		\draw [style=none] (359.center) to (362.center);
		\draw [style=none] (362.center) to (364.center);
		\draw [style=none] (364.center) to (366.center);
		\draw [style=none] (366.center) to (365.center);
		\draw [style=none] (365.center) to (363.center);
		\draw [style=none] (363.center) to (361.center);
		\draw [style=none] (372.center) to (371.center);
		\draw [style=none] (371.center) to (370.center);
		\draw [style=none] (370.center) to (373.center);
		\draw [style=none] (373.center) to (375.center);
		\draw [style=none] (375.center) to (377.center);
		\draw [style=none] (377.center) to (376.center);
		\draw [style=none] (376.center) to (374.center);
		\draw [style=none] (374.center) to (372.center);
		\draw [style=none] (377.center) to (379.center);
		\draw [style=none] (376.center) to (378.center);
		\draw [style=none] (385.center) to (371.center);
		\draw [style=none] (386.center) to (370.center);
		\draw [style=none] (367.center) to (383.center);
		\draw [style=none] (368.center) to (384.center);
	\end{pgfonlayer}
\end{tikzpicture}

%% file: figures/fig2.tikz
\begin{tikzpicture}
	\begin{pgfonlayer}{nodelayer}
		\node [style=none] (198) at (-9.75, -1) {};
		\node [style=none] (199) at (-8.25, -1) {};
		\node [style=none] (200) at (-10, 7) {};
		\node [style=none] (201) at (-8, 7) {};
		\node [style=none] (202) at (-9.5, 7) {};
		\node [style=none] (203) at (-8.75, 6.475) {$\dots$};
		\node [style=none] (204) at (-10, 4) {};
		\node [style=none] (205) at (-8, 4) {};
		\node [style=none] (206) at (-9.5, 4) {};
		\node [style=none] (207) at (-8.75, 4.525) {$\dots$};
		\node [style=Green Node] (208) at (-9, 5.5) {};
		\node [style=Green Node] (211) at (-9, 5.5) {$\alpha$};
		\node [style=none] (212) at (0.5, 4) {};
		\node [style=none] (213) at (2.5, 4) {};
		\node [style=none] (214) at (1, 4) {};
		\node [style=none] (215) at (1.75, 4.525) {$\dots$};
		\node [style=none] (216) at (0.5, 7) {};
		\node [style=none] (217) at (2.5, 7) {};
		\node [style=none] (218) at (1, 7) {};
		\node [style=none] (219) at (1.75, 6.475) {$\dots$};
		\node [style=Red Node] (220) at (1.5, 5.5) {};
		\node [style=Red Node] (221) at (1.5, 5.5) {$\alpha$};
		\node [style=none] (222) at (-9, 0.25) {};
		\node [style=none] (223) at (-9, 2.75) {};
		\node [style=h] (224) at (-9, 1.5) {};
		\node [style=none] (225) at (-1.25, -1.75) {};
		\node [style=none] (226) at (0.25, -1.75) {};
		\node [style=none] (227) at (-2.5, 0.25) {};
		\node [style=none] (228) at (-2.5, 2.75) {};
		\node [style=none] (229) at (3, 2.5) {};
		\node [style=none] (230) at (3, 0.5) {};
		\node [style=none] (231) at (4.25, 2.5) {};
		\node [style=none] (232) at (4.25, 0.5) {};
		\node [style=bigtxt] (234) at (-0.25, 5) {,};
		\node [style=bigtxt] (236) at (-3.75, 1) {,};
		\node [style=bigtxt] (237) at (1.5, 1) {,};
		\node [style=bigtxt] (238) at (-3.25, -1.75) {,};
		\node [style=bigtxt] (241) at (5.25, -1.75) {,};
		\node [style=none] (242) at (7, -0.75) {};
		\node [style=none] (243) at (8, -0.75) {};
		\node [style=none] (244) at (7, -1.75) {};
		\node [style=none] (245) at (8, -1.75) {};
		\node [style=txtnobold] (247) at (-9, 4) {$m$};
		\node [style=txtnobold] (249) at (1.5, 4) {$m$};
		\node [style=txtnobold] (251) at (-9, 7) {$n$};
		\node [style=txtnobold] (253) at (1.5, 7) {$n$};
		\node [style=bigtxt] (254) at (-6.5, 8) {Standard interpretation};
		\node [style=none] (255) at (-11, 8) {};
		\node [style=none] (256) at (-11, -7.5) {};
		\node [style=none] (257) at (-9.5, 8) {};
		\node [style=none] (258) at (11, 8) {};
		\node [style=none] (259) at (11, -7.5) {};
		\node [style=none] (260) at (-3.5, 8) {};
		\node [style=txtnobold] (263) at (-4.25, 5.5) {$\xmapsto{\llbracket\cdot \rrbracket}\ \ |0^{(m)}\rangle\langle0^{(n)}| + e^{i\alpha}|1^{(m)}\rangle\langle 1^{(n)}|$};
		\node [style=txtnobold] (264) at (6.5, 5.5) {$\xmapsto{\llbracket\cdot \rrbracket}\ \ |+^{(m)}\rangle\langle +^{(n)}| + e^{i\alpha}|-^{(m)}\rangle\langle -^{(n)}|$};
		\node [style=txtnobold] (265) at (-6, 1.5) {$\xmapsto{\llbracket\cdot \rrbracket}\ \ \frac{1}{\sqrt{2}}\begin{bmatrix}1 & 1\\ 1 & -1 \end{bmatrix}$};
		\node [style=txtnobold] (266) at (-0.25, 1.5) {$\xmapsto{\llbracket\cdot \rrbracket}\ \ \begin{bmatrix}1 & 0\\ 0 & 1 \end{bmatrix}$};
		\node [style=txtnobold] (267) at (7, 1.5) {$\xmapsto{\llbracket\cdot \rrbracket}\ \ \begin{bmatrix}1 & 0 & 0 & 0\\ 0 & 0 & 1& 0\\  0 & 1 & 0& 0 \\  0 & 0 & 0& 1  \end{bmatrix}$};
		\node [style=txtnobold] (268) at (-5.25, -1.25) {$\xmapsto{\llbracket\cdot \rrbracket}\ \ \begin{bmatrix} 1 & 0 & 0 & 1  \end{bmatrix}$};
		\node [style=txtnobold] (269) at (3.25, -1.25) {$\xmapsto{\llbracket\cdot \rrbracket}\ \ \begin{bmatrix} 1 & 0 & 0 & 1  \end{bmatrix}^T$};
		\node [style=txtnobold] (270) at (9.5, -1) {$\xmapsto{\llbracket\cdot \rrbracket}\ \ 1$};
		\node [style=none] (271) at (-8, -3) {};
		\node [style=none] (272) at (-7, -3) {};
		\node [style=txtnobold] (273) at (-7.5, -3.25) {$\cdots$};
		\node [style=none] (274) at (-7, -3.5) {};
		\node [style=none] (275) at (-8, -3.5) {};
		\node [style=none] (276) at (-8.25, -3.5) {};
		\node [style=none] (277) at (-6.75, -3.5) {};
		\node [style=none] (278) at (-8.25, -4.5) {};
		\node [style=none] (279) at (-6.75, -4.5) {};
		\node [style=none] (280) at (-8, -4.5) {};
		\node [style=none] (281) at (-7, -4.5) {};
		\node [style=none] (282) at (-8, -4.5) {};
		\node [style=none] (283) at (-7, -4.5) {};
		\node [style=txtnobold] (284) at (-7.5, -4.75) {$\cdots$};
		\node [style=none] (285) at (-7, -5) {};
		\node [style=none] (286) at (-8, -5) {};
		\node [style=none] (287) at (-8.25, -5) {};
		\node [style=none] (288) at (-6.75, -5) {};
		\node [style=none] (289) at (-8.25, -6) {};
		\node [style=none] (290) at (-6.75, -6) {};
		\node [style=none] (291) at (-8, -6) {};
		\node [style=none] (292) at (-7, -6) {};
		\node [style=none] (293) at (-8, -6.5) {};
		\node [style=none] (294) at (-7, -6.5) {};
		\node [style=txtnobold] (295) at (-7.5, -6.25) {$\cdots$};
		\node [style=txtnobold] (296) at (-7.5, -4) {$D_1$};
		\node [style=txtnobold] (297) at (-7.5, -5.5) {$D_2$};
		\node [style=none] (298) at (1, -3.75) {};
		\node [style=none] (299) at (2, -3.75) {};
		\node [style=txtnobold] (300) at (1.5, -4) {$\cdots$};
		\node [style=none] (301) at (2, -4.25) {};
		\node [style=none] (302) at (1, -4.25) {};
		\node [style=none] (303) at (0.75, -4.25) {};
		\node [style=none] (304) at (2.25, -4.25) {};
		\node [style=none] (305) at (0.75, -5.25) {};
		\node [style=none] (306) at (2.25, -5.25) {};
		\node [style=none] (307) at (1, -5.25) {};
		\node [style=none] (308) at (2, -5.25) {};
		\node [style=none] (309) at (1, -5.25) {};
		\node [style=none] (310) at (2, -5.25) {};
		\node [style=txtnobold] (311) at (1.5, -5.5) {$\cdots$};
		\node [style=none] (312) at (4.5, -4.25) {};
		\node [style=none] (313) at (3.5, -4.25) {};
		\node [style=none] (314) at (3.25, -4.25) {};
		\node [style=none] (315) at (4.75, -4.25) {};
		\node [style=none] (316) at (3.25, -5.25) {};
		\node [style=none] (317) at (4.75, -5.25) {};
		\node [style=none] (318) at (3.5, -5.25) {};
		\node [style=none] (319) at (4.5, -5.25) {};
		\node [style=none] (320) at (3.5, -5.75) {};
		\node [style=none] (321) at (4.5, -5.75) {};
		\node [style=txtnobold] (322) at (4, -5.5) {$\cdots$};
		\node [style=txtnobold] (323) at (1.5, -4.75) {$D_1$};
		\node [style=txtnobold] (324) at (4, -4.75) {$D_2$};
		\node [style=none] (325) at (1, -5.75) {};
		\node [style=none] (326) at (2, -5.75) {};
		\node [style=none] (327) at (3.5, -3.75) {};
		\node [style=none] (328) at (4.5, -3.75) {};
		\node [style=txtnobold] (329) at (4, -4) {$\cdots$};
		\node [style=txtnobold] (330) at (-4.25, -4.75) {$\xmapsto{\llbracket\cdot \rrbracket}\ \ \llbracket D_2\rrbracket \circ \llbracket D_1\rrbracket$};
		\node [style=txtnobold] (331) at (7.25, -4.75) {$\xmapsto{\llbracket\cdot \rrbracket}\ \ \llbracket D_1\rrbracket \otimes \llbracket D_2\rrbracket$};
		\node [style=bigtxt] (332) at (-2, -5.25) {,};
	\end{pgfonlayer}
	\begin{pgfonlayer}{edgelayer}
		\draw [style=none, bend right=105, looseness=2.25] (198.center) to (199.center);
		\draw [style=none, bend left] (208) to (200.center);
		\draw [style=none, bend left=15] (208) to (202.center);
		\draw [style=none, bend right=15] (208) to (201.center);
		\draw [style=none, bend right] (208) to (204.center);
		\draw [style=none, bend right=15] (208) to (206.center);
		\draw [style=none, bend left=15] (208) to (205.center);
		\draw [style=none, bend right] (220) to (212.center);
		\draw [style=none, bend right=15] (220) to (214.center);
		\draw [style=none, bend left=15] (220) to (213.center);
		\draw [style=none, bend left] (220) to (216.center);
		\draw [style=none, bend left=15] (220) to (218.center);
		\draw [style=none, bend right=15] (220) to (217.center);
		\draw [style=none] (222.center) to (224);
		\draw [style=none] (224) to (223.center);
		\draw [style=none, bend left=105, looseness=2.25] (225.center) to (226.center);
		\draw [style=none] (227.center) to (228.center);
		\draw [style=none, in=-90, out=90] (230.center) to (231.center);
		\draw [style=none, in=-90, out=90] (232.center) to (229.center);
		\draw [style=discontinuous] (242.center) to (243.center);
		\draw [style=discontinuous] (243.center) to (245.center);
		\draw [style=discontinuous] (245.center) to (244.center);
		\draw [style=discontinuous] (242.center) to (244.center);
		\draw (255.center) to (256.center);
		\draw (255.center) to (257.center);
		\draw (258.center) to (259.center);
		\draw (258.center) to (260.center);
		\draw (256.center) to (259.center);
		\draw [style=none] (271.center) to (275.center);
		\draw [style=none] (272.center) to (274.center);
		\draw [style=none] (276.center) to (275.center);
		\draw [style=none] (275.center) to (274.center);
		\draw [style=none] (274.center) to (277.center);
		\draw [style=none] (277.center) to (279.center);
		\draw [style=none] (279.center) to (281.center);
		\draw [style=none] (281.center) to (280.center);
		\draw [style=none] (280.center) to (278.center);
		\draw [style=none] (278.center) to (276.center);
		\draw [style=none] (282.center) to (286.center);
		\draw [style=none] (283.center) to (285.center);
		\draw [style=none] (287.center) to (286.center);
		\draw [style=none] (286.center) to (285.center);
		\draw [style=none] (285.center) to (288.center);
		\draw [style=none] (288.center) to (290.center);
		\draw [style=none] (290.center) to (292.center);
		\draw [style=none] (292.center) to (291.center);
		\draw [style=none] (291.center) to (289.center);
		\draw [style=none] (289.center) to (287.center);
		\draw [style=none] (292.center) to (294.center);
		\draw [style=none] (291.center) to (293.center);
		\draw [style=none] (298.center) to (302.center);
		\draw [style=none] (299.center) to (301.center);
		\draw [style=none] (303.center) to (302.center);
		\draw [style=none] (302.center) to (301.center);
		\draw [style=none] (301.center) to (304.center);
		\draw [style=none] (304.center) to (306.center);
		\draw [style=none] (306.center) to (308.center);
		\draw [style=none] (308.center) to (307.center);
		\draw [style=none] (307.center) to (305.center);
		\draw [style=none] (305.center) to (303.center);
		\draw [style=none] (314.center) to (313.center);
		\draw [style=none] (313.center) to (312.center);
		\draw [style=none] (312.center) to (315.center);
		\draw [style=none] (315.center) to (317.center);
		\draw [style=none] (317.center) to (319.center);
		\draw [style=none] (319.center) to (318.center);
		\draw [style=none] (318.center) to (316.center);
		\draw [style=none] (316.center) to (314.center);
		\draw [style=none] (319.center) to (321.center);
		\draw [style=none] (318.center) to (320.center);
		\draw [style=none] (327.center) to (313.center);
		\draw [style=none] (328.center) to (312.center);
		\draw [style=none] (309.center) to (325.center);
		\draw [style=none] (310.center) to (326.center);
	\end{pgfonlayer}
\end{tikzpicture}

%% file: figures/ruleset.tikz
\begin{tikzpicture}
	\begin{pgfonlayer}{nodelayer}
		\node [style=none] (2) at (-10.5, 4.75) {};
		\node [style=none] (3) at (-8.5, 4.75) {};
		\node [style=none] (4) at (-10, 4.75) {};
		\node [style=none] (5) at (-9.25, 4.225) {$\dots$};
		\node [style=none] (6) at (-10.5, 1.75) {};
		\node [style=none] (7) at (-8.5, 1.75) {};
		\node [style=none] (8) at (-10, 1.75) {};
		\node [style=none] (9) at (-9.25, 2.275) {$\dots$};
		\node [style=Green Node] (10) at (-9.5, 3.25) {};
		\node [style=Green Node] (11) at (-9.5, 3.25) {$\alpha$};
		\node [style=Green Node] (20) at (3.75, 4.25) {$\frac{\pi}{4}$};
		\node [style=none] (22) at (8.5, 1.75) {};
		\node [style=none] (23) at (8.5, 4.25) {};
		\node [style=h] (24) at (8.5, 3) {};
		\node [style=bigtxt] (33) at (2, 2.5) {,};
		\node [style=bigtxt] (34) at (-5, -2.25) {,};
		\node [style=none] (38) at (5.25, 3.5) {};
		\node [style=none] (39) at (6.25, 3.5) {};
		\node [style=none] (40) at (5.25, 2.5) {};
		\node [style=none] (41) at (6.25, 2.5) {};
		\node [style=bigtxt] (46) at (-7.5, 5.75) {ZX-calculus ruleset};
		\node [style=none] (47) at (-11.5, 5.75) {};
		\node [style=none] (48) at (-11.5, -4) {};
		\node [style=none] (49) at (-10, 5.75) {};
		\node [style=none] (50) at (13.25, 5.75) {};
		\node [style=none] (51) at (13.25, -4) {};
		\node [style=none] (52) at (-5, 5.75) {};
		\node [style=none] (123) at (-7.5, 4) {};
		\node [style=none] (124) at (-5.5, 4) {};
		\node [style=none] (125) at (-7, 4) {};
		\node [style=none] (126) at (-6.25, 3.475) {$\dots$};
		\node [style=none] (127) at (-7.5, 1) {};
		\node [style=none] (128) at (-5.5, 1) {};
		\node [style=none] (129) at (-7, 1) {};
		\node [style=none] (130) at (-6.25, 1.525) {$\dots$};
		\node [style=Green Node] (131) at (-6.5, 2.5) {};
		\node [style=Green Node] (132) at (-6.5, 2.5) {$\beta$};
		\node [style=none] (133) at (-8, 3.025) {$\vdots$};
		\node [style=none] (134) at (-6.75, 2.5) {};
		\node [style=none] (135) at (-4.25, 4.5) {};
		\node [style=none] (136) at (-2.25, 4.5) {};
		\node [style=none] (137) at (-3.75, 4.5) {};
		\node [style=none] (138) at (-3, 3.975) {$\dots$};
		\node [style=none] (139) at (-4.25, 1.5) {};
		\node [style=none] (140) at (-2.25, 1.5) {};
		\node [style=none] (141) at (-3.75, 1.5) {};
		\node [style=none] (142) at (-3, 2.025) {$\dots$};
		\node [style=Green Node] (143) at (-3.25, 3) {};
		\node [style=Green Node] (144) at (-3.25, 3) {$\alpha+\beta$};
		\node [style=txtnobold] (145) at (-5, 3) {$\stackrel{(f)}{=}$};
		\node [style=none] (147) at (0, 4.5) {};
		\node [style=none] (148) at (0, 1.5) {};
		\node [style=txtnobold] (149) at (0.75, 3) {$\stackrel{(id)}{=}$};
		\node [style=none] (150) at (1.5, 4.5) {};
		\node [style=none] (151) at (1.5, 1.5) {};
		\node [style=Red Node] (152) at (3.75, 1.75) {$\frac{\text{-}\pi}{4}$};
		\node [style=bigtxt] (153) at (-2, 2.5) {,};
		\node [style=txtnobold] (154) at (4.5, 3) {$\stackrel{(e)}{=}$};
		\node [style=Red Node] (155) at (-10, -0.5) {$a\pi$};
		\node [style=Green Node] (156) at (-10, -1.5) {$\alpha$};
		\node [style=none] (157) at (-10.5, -2.5) {};
		\node [style=none] (158) at (-9.5, -2.5) {};
		\node [style=txtnobold] (159) at (-9, -1.25) {$\stackrel{(\pi)}{=}$};
		\node [style=Red Node] (160) at (-8, -1) {$a\pi$};
		\node [style=Red Node] (161) at (-7, -1) {$a\pi$};
		\node [style=none] (162) at (-8, -2) {};
		\node [style=none] (163) at (-7, -2) {};
		\node [style=txtnobold] (164) at (-5.75, -1.25) {$\frac{e^{ia\alpha}}{\sqrt{2}}$};
		\node [style=none] (165) at (-3.5, 0) {};
		\node [style=none] (166) at (-2, 0) {};
		\node [style=none] (167) at (-3.5, -3) {};
		\node [style=none] (168) at (-2, -3) {};
		\node [style=Empty green] (173) at (-3.5, -1) {};
		\node [style=Empty green] (174) at (-2, -1) {};
		\node [style=Empty red] (175) at (-3.5, -2) {};
		\node [style=Empty red] (176) at (-2, -2) {};
		\node [style=Empty green] (177) at (0, 3) {};
		\node [style=txtnobold] (178) at (-1, -1.5) {$\stackrel{(b)}{=}$};
		\node [style=Empty red] (179) at (0, -1) {};
		\node [style=Empty green] (180) at (0, -2) {};
		\node [style=none] (181) at (-0.5, -3) {};
		\node [style=none] (182) at (0.5, -3) {};
		\node [style=none] (183) at (-0.5, 0) {};
		\node [style=none] (184) at (0.5, 0) {};
		\node [style=txtnobold] (185) at (1, -1.25) {$\frac{1}{\sqrt{2}}$};
		\node [style=bigtxt] (186) at (1.5, -2.25) {,};
		\node [style=bigtxt] (187) at (6.75, 2.5) {,};
		\node [style=txtnobold] (188) at (9.25, 3) {$\stackrel{(hd)}{=}$};
		\node [style=none] (189) at (10.25, 4.75) {};
		\node [style=none] (190) at (10.25, 1.25) {};
		\node [style=Green Node] (191) at (10.25, 4) {$\frac{\pi}{2}$};
		\node [style=Green Node] (192) at (10.25, 2) {$\frac{\pi}{2}$};
		\node [style=Empty red] (193) at (10.25, 3) {};
		\node [style=Green Node] (194) at (11.25, 3) {$\frac{\text{-}\pi}{2}$};
		\node [style=none] (195) at (2.5, 0) {};
		\node [style=none] (196) at (4.5, 0) {};
		\node [style=none] (197) at (3, 0) {};
		\node [style=none] (198) at (3.75, -0.525) {$\dots$};
		\node [style=none] (199) at (2.5, -3) {};
		\node [style=none] (200) at (4.5, -3) {};
		\node [style=none] (201) at (3, -3) {};
		\node [style=none] (202) at (3.75, -2.475) {$\dots$};
		\node [style=Green Node] (203) at (3.5, -1.5) {};
		\node [style=Green Node] (204) at (3.5, -1.5) {$\alpha$};
		\node [style=h] (205) at (2.5, -0.5) {};
		\node [style=h] (206) at (3, -0.5) {};
		\node [style=h] (207) at (4.5, -0.5) {};
		\node [style=h] (208) at (2.5, -2.5) {};
		\node [style=h] (209) at (4.5, -2.5) {};
		\node [style=h] (210) at (3, -2.5) {};
		\node [style=none] (211) at (5.5, 0) {};
		\node [style=none] (212) at (7.5, 0) {};
		\node [style=none] (213) at (6, 0) {};
		\node [style=none] (214) at (6.75, -0.525) {$\dots$};
		\node [style=none] (215) at (5.5, -3) {};
		\node [style=none] (216) at (7.5, -3) {};
		\node [style=none] (217) at (6, -3) {};
		\node [style=none] (218) at (6.75, -2.475) {$\dots$};
		\node [style=Green Node] (219) at (6.5, -1.5) {};
		\node [style=Red Node] (220) at (6.5, -1.5) {$\alpha$};
		\node [style=txtnobold] (221) at (5, -1.5) {$\stackrel{(h)}{=}$};
		\node [style=bigtxt] (222) at (8.25, -3) {,};
		\node [style=Green Node] (223) at (9.25, -0.75) {$\alpha_1$};
		\node [style=Green Node] (224) at (9.25, -2.75) {$\alpha_3$};
		\node [style=Green Node] (225) at (11.25, -1.75) {$\beta_2$};
		\node [style=Red Node] (228) at (9.25, -1.75) {$\alpha_2$};
		\node [style=Red Node] (229) at (11.25, -0.75) {$\beta_1$};
		\node [style=Red Node] (230) at (11.25, -2.75) {$\beta_3$};
		\node [style=none] (233) at (9.25, 0) {};
		\node [style=none] (234) at (11.25, 0) {};
		\node [style=none] (235) at (9.25, -3.5) {};
		\node [style=none] (236) at (11.25, -3.5) {};
		\node [style=txtnobold] (237) at (12.25, -1) {$e^{i\gamma}$};
		\node [style=txtnobold] (238) at (10.25, -1.5) {$\stackrel{(eu)}{=}$};
	\end{pgfonlayer}
	\begin{pgfonlayer}{edgelayer}
		\draw [style=none, bend left] (10) to (2.center);
		\draw [style=none, bend left=15] (10) to (4.center);
		\draw [style=none, bend right=15] (10) to (3.center);
		\draw [style=none, bend right] (10) to (6.center);
		\draw [style=none, bend right=15] (10) to (8.center);
		\draw [style=none, bend left=15] (10) to (7.center);
		\draw [style=none] (22.center) to (24);
		\draw [style=none] (24) to (23.center);
		\draw [style=discontinuous] (38.center) to (39.center);
		\draw [style=discontinuous] (39.center) to (41.center);
		\draw [style=discontinuous] (41.center) to (40.center);
		\draw [style=discontinuous] (38.center) to (40.center);
		\draw (47.center) to (48.center);
		\draw (47.center) to (49.center);
		\draw (50.center) to (51.center);
		\draw (50.center) to (52.center);
		\draw (48.center) to (51.center);
		\draw [style=none, bend left] (131) to (123.center);
		\draw [style=none, bend left=15] (131) to (125.center);
		\draw [style=none, bend right=15] (131) to (124.center);
		\draw [style=none, bend right] (131) to (127.center);
		\draw [style=none, bend right=15] (131) to (129.center);
		\draw [style=none, bend left=15] (131) to (128.center);
		\draw [bend right] (11) to (132);
		\draw [bend left] (11) to (134.center);
		\draw [style=none, bend left] (143) to (135.center);
		\draw [style=none, bend left=15] (143) to (137.center);
		\draw [style=none, bend right=15] (143) to (136.center);
		\draw [style=none, bend right] (143) to (139.center);
		\draw [style=none, bend right=15] (143) to (141.center);
		\draw [style=none, bend left=15] (143) to (140.center);
		\draw (150.center) to (151.center);
		\draw (20) to (152);
		\draw (155) to (156);
		\draw [bend right=15] (156) to (157.center);
		\draw [bend left=15] (156) to (158.center);
		\draw (160) to (162.center);
		\draw (161) to (163.center);
		\draw (165.center) to (173);
		\draw (173) to (176);
		\draw (175) to (174);
		\draw (174) to (176);
		\draw (176) to (168.center);
		\draw (175) to (167.center);
		\draw (175) to (173);
		\draw (166.center) to (174);
		\draw (147.center) to (177);
		\draw (177) to (148.center);
		\draw [bend right=15] (183.center) to (179);
		\draw [bend left=15] (184.center) to (179);
		\draw (179) to (180);
		\draw [bend right=15] (180) to (181.center);
		\draw [bend left=15] (180) to (182.center);
		\draw (189.center) to (191);
		\draw (191) to (193);
		\draw (193) to (194);
		\draw (190.center) to (192);
		\draw (192) to (193);
		\draw [bend right] (205) to (204);
		\draw (205) to (195.center);
		\draw (206) to (197.center);
		\draw [bend right=15] (206) to (204);
		\draw [bend right] (204) to (207);
		\draw (207) to (196.center);
		\draw [bend right] (204) to (208);
		\draw (208) to (199.center);
		\draw (210) to (201.center);
		\draw [bend left] (204) to (209);
		\draw (209) to (200.center);
		\draw [bend right=15] (204) to (210);
		\draw [style=none, bend left] (219) to (211.center);
		\draw [style=none, bend left=15] (219) to (213.center);
		\draw [style=none, bend right=15] (219) to (212.center);
		\draw [style=none, bend right] (219) to (215.center);
		\draw [style=none, bend right=15] (219) to (217.center);
		\draw [style=none, bend left=15] (219) to (216.center);
		\draw (233.center) to (223);
		\draw (223) to (228);
		\draw (228) to (224);
		\draw (224) to (235.center);
		\draw (234.center) to (229);
		\draw (229) to (225);
		\draw (225) to (230);
		\draw (230) to (236.center);
	\end{pgfonlayer}
\end{tikzpicture}

%% file: figures/zxExample.tikz
\begin{tikzpicture}
	\begin{pgfonlayer}{nodelayer}
		\node [style=txtnobold] (0) at (-3, 1) {$Z = \begin{bmatrix}1 & 0\\ 0 & -1 \end{bmatrix} = $};
		\node [style=Green Node] (1) at (-0.75, 1) {$\pi$};
		\node [style=none] (2) at (-0.75, 2) {};
		\node [style=none] (3) at (-0.75, 0) {};
		\node [style=txtnobold] (4) at (2, 1) {$X = \begin{bmatrix}0 & 1\\ 1 & 0 \end{bmatrix} = $};
		\node [style=Red Node] (5) at (4, 1) {$\pi$};
		\node [style=none] (6) at (4, 2) {};
		\node [style=none] (7) at (4, 0) {};
		\node [style=txtnobold] (8) at (7.75, 1) {$Y = i XZ = \begin{bmatrix}0 & -i\\ i & 0 \end{bmatrix} = $};
		\node [style=Green Node] (9) at (11.25, 1.5) {$\pi$};
		\node [style=none] (10) at (11.25, 2.5) {};
		\node [style=none] (11) at (11.25, -0.5) {};
		\node [style=Red Node] (12) at (11.25, 0.5) {$\pi$};
		\node [style=txtnobold] (13) at (10.75, 1) {$i$};
		\node [style=txtnobold] (14) at (-2.5, -3) {CNOT $ = \begin{bmatrix}1 & 0 & 0 & 0\\ 0 & 1 & 0 & 0\\ 0 & 0 & 0 & 1\\ 0 & 0 & 1 & 0 \end{bmatrix} = $};
		\node [style=none] (16) at (1.25, -2) {};
		\node [style=none] (17) at (1.25, -4) {};
		\node [style=Empty green] (18) at (1.25, -3) {};
		\node [style=none] (19) at (2.25, -2) {};
		\node [style=none] (20) at (2.25, -4) {};
		\node [style=Empty red] (21) at (2.25, -3) {};
		\node [style=txtnobold] (22) at (0.5, -3) {$\sqrt{2}$};
		\node [style=txtnobold] (23) at (6, -3) {CZ $ = \begin{bmatrix}1 & 0 & 0 & 0\\ 0 & 1 & 0 & 0\\ 0 & 0 & 1 & 0\\ 0 & 0 & 0 & -1 \end{bmatrix} = $};
		\node [style=none] (24) at (9.75, -2) {};
		\node [style=none] (25) at (9.75, -4) {};
		\node [style=Empty green] (26) at (9.75, -3) {};
		\node [style=none] (27) at (10.75, -2) {};
		\node [style=none] (28) at (10.75, -4) {};
		\node [style=Empty green] (29) at (10.75, -3) {};
		\node [style=txtnobold] (30) at (9, -3) {$\sqrt{2}$};
		\node [style=h] (31) at (10.25, -3) {};
		\node [style=txtnobold] (32) at (-3.75, -6.25) {$|0\rangle=$};
		\node [style=none] (35) at (-2, -7) {};
		\node [style=txtnobold] (36) at (0, -6.25) {$|1\rangle=$};
		\node [style=Red Node] (37) at (2, -6) {$\pi$};
		\node [style=none] (38) at (2, -7) {};
		\node [style=Empty red] (39) at (-2, -6) {};
		\node [style=txtnobold] (40) at (-2.75, -6.25) {$\frac{1}{\sqrt{2}}$};
		\node [style=txtnobold] (41) at (1, -6.25) {$\frac{1}{\sqrt{2}}$};
		\node [style=txtnobold] (42) at (4.25, -6.25) {$|+\rangle=$};
		\node [style=none] (43) at (6, -7) {};
		\node [style=txtnobold] (44) at (8, -6.25) {$|-\rangle=$};
		\node [style=Green Node] (45) at (10, -6) {$\pi$};
		\node [style=none] (46) at (10, -7) {};
		\node [style=Empty green] (47) at (6, -6) {};
		\node [style=txtnobold] (48) at (5.25, -6.25) {$\frac{1}{\sqrt{2}}$};
		\node [style=txtnobold] (49) at (9, -6.25) {$\frac{1}{\sqrt{2}}$};
	\end{pgfonlayer}
	\begin{pgfonlayer}{edgelayer}
		\draw (2.center) to (1);
		\draw (1) to (3.center);
		\draw (6.center) to (5);
		\draw (5) to (7.center);
		\draw (10.center) to (9);
		\draw (9) to (12);
		\draw (12) to (11.center);
		\draw (16.center) to (18);
		\draw (18) to (17.center);
		\draw (18) to (21);
		\draw (21) to (19.center);
		\draw (21) to (20.center);
		\draw (24.center) to (26);
		\draw (26) to (25.center);
		\draw (29) to (27.center);
		\draw (29) to (28.center);
		\draw (26) to (31);
		\draw (31) to (29);
		\draw (37) to (38.center);
		\draw (39) to (35.center);
		\draw (45) to (46.center);
		\draw (47) to (43.center);
	\end{pgfonlayer}
\end{tikzpicture}

%% file: figures/openingbracket.tikz
\begin{tikzpicture}
	\begin{pgfonlayer}{nodelayer}
		\node [style=none] (14) at (0, 0.25) {};
		\node [style=none] (15) at (-0.25, 0.25) {};
		\node [style=none] (16) at (0.25, 0.25) {};
		\node [style=none] (17) at (-0.5, -0.25) {};
		\node [style=none] (18) at (0.5, -0.25) {};
		\node [style=none] (19) at (-0.5, -0.25) {};
		\node [style=none] (20) at (0, -0.25) {};
		\node [style=none] (21) at (-0.5, 1) {};
		\node [style=none] (22) at (0.5, 1) {};
		\node [style=none] (23) at (0, 0.75) {$\cdots$};
		\node [style=none] (24) at (0.25, -0.25) {};
		\node [style=none] (25) at (-0.25, -0.25) {};
		\node [style=none] (38) at (-0.5, -1) {};
		\node [style=none] (39) at (0.5, -1) {};
		\node [style=none] (40) at (0, -0.75) {$\cdots$};
	\end{pgfonlayer}
	\begin{pgfonlayer}{edgelayer}
		\draw [style=distribution] (20.center)
			 to (18.center)
			 to (16.center)
			 to (14.center)
			 to (15.center)
			 to (17.center)
			 to (19.center)
			 to cycle;
		\draw [style=edge, bend right=15] (21.center) to (15.center);
		\draw [style=edge, bend left=15] (22.center) to (16.center);
		\draw [style=edge, bend left=15] (38.center) to (25.center);
		\draw [style=edge, bend right=15] (39.center) to (24.center);
	\end{pgfonlayer}
\end{tikzpicture}

%% file: figures/closingbracket.tikz
\begin{tikzpicture}
	\begin{pgfonlayer}{nodelayer}
		\node [style=none] (0) at (0, -0.25) {};
		\node [style=none] (1) at (-0.25, -0.25) {};
		\node [style=none] (2) at (0.25, -0.25) {};
		\node [style=none] (3) at (-0.5, 0.25) {};
		\node [style=none] (4) at (0.5, 0.25) {};
		\node [style=none] (5) at (-0.5, 0.25) {};
		\node [style=none] (6) at (0, 0.25) {};
		\node [style=none] (7) at (0.25, 0.25) {};
		\node [style=none] (8) at (-0.25, 0.25) {};
		\node [style=none] (9) at (-0.5, -1) {};
		\node [style=none] (10) at (0.5, -1) {};
		\node [style=none] (11) at (0, -0.75) {$\cdots$};
		\node [style=none] (12) at (-0.5, 1) {};
		\node [style=none] (13) at (0.5, 1) {};
		\node [style=none] (14) at (0, 0.75) {$\cdots$};
	\end{pgfonlayer}
	\begin{pgfonlayer}{edgelayer}
		\draw [style=distribution] (3.center)
			 to (5.center)
			 to (6.center)
			 to (4.center)
			 to (2.center)
			 to (0.center)
			 to (1.center)
			 to cycle;
		\draw [style=edge, bend right=15] (1.center) to (9.center);
		\draw [style=edge, bend left=15, looseness=0.75] (2.center) to (10.center);
		\draw [style=edge, bend left=15] (13.center) to (7.center);
		\draw [style=edge, bend right=15] (12.center) to (8.center);
	\end{pgfonlayer}
\end{tikzpicture}

%% file: figures/zx-examples.tikz
\begin{tikzpicture}
	\begin{pgfonlayer}{nodelayer}
		\node [style=none] (0) at (3.875, 2.5) {};
		\node [style=txtnobold] (1) at (2.45, 1) {$p$};
		\node [style=txtnobold] (2) at (4.94, 1) {$1-p$};
		\node [style=none] (3) at (3.25, 1.75) {};
		\node [style=none] (4) at (3.75, 1.75) {};
		\node [style=none] (5) at (3, 1.25) {};
		\node [style=none] (6) at (4, 1.25) {};
		\node [style=none] (7) at (3, 1.25) {};
		\node [style=none] (8) at (3.5, 1.25) {};
		\node [style=none] (9) at (3.75, 1.25) {};
		\node [style=none] (10) at (3.25, 1.25) {};
		\node [style=none] (11) at (3.25, -1.75) {};
		\node [style=none] (12) at (3.75, -1.75) {};
		\node [style=none] (13) at (3, -1.25) {};
		\node [style=none] (14) at (4, -1.25) {};
		\node [style=none] (15) at (3, -1.25) {};
		\node [style=none] (16) at (3.5, -1.25) {};
		\node [style=none] (17) at (3.75, -1.25) {};
		\node [style=none] (18) at (3.25, -1.25) {};
		\node [style=none] (19) at (3.1, -1.25) {};
		\node [style=none] (20) at (3.1, 1.25) {};
		\node [style=none] (21) at (3.875, 1.25) {};
		\node [style=none] (22) at (3.875, -1.25) {};
		\node [style=none] (23) at (3.375, 1.75) {};
		\node [style=none] (24) at (3.625, 1.75) {};
		\node [style=none] (25) at (3.125, 2.5) {};
		\node [style=none] (26) at (3.375, -1.75) {};
		\node [style=none] (27) at (3.125, -2.5) {};
		\node [style=none] (28) at (3.625, -1.75) {};
		\node [style=none] (29) at (3.875, -2.5) {};
		\node [style=Empty green] (30) at (2, 0) {};
		\node [style=Empty red] (31) at (3, 0) {};
		\node [style=Empty green] (32) at (5, 0) {};
		\node [style=Empty green] (33) at (4, 0) {};
		\node [style=h] (34) at (4.5, 0) {};
		\node [style=none] (35) at (-3.25, 3.75) {};
		\node [style=Green Node] (36) at (-2.75, 0.75) {$\pi$};
		\node [style=Red Node] (37) at (-3.75, 0) {$\pi$};
		\node [style=Red Node] (38) at (-2.75, -0.75) {$\pi$};
		\node [style=Green Node] (39) at (-2, 0) {$\pi$};
		\node [style=txtnobold] (40) at (-4.75, 2) {$1-p$};
		\node [style=txtnobold] (41) at (-2.05, 2) {$\frac{p}{3}$};
		\node [style=txtnobold] (42) at (-3.06, 2) {$\frac{p}{3}$};
		\node [style=txtnobold] (43) at (-3.8, 2) {$\frac{p}{3}$};
		\node [style=none] (44) at (-3.25, 3) {};
		\node [style=none] (45) at (-3.5, 3) {};
		\node [style=none] (46) at (-3, 3) {};
		\node [style=none] (47) at (-3.75, 2.5) {};
		\node [style=none] (48) at (-2.75, 2.5) {};
		\node [style=none] (49) at (-3.75, 2.5) {};
		\node [style=none] (50) at (-3.25, 2.5) {};
		\node [style=none] (51) at (-3, 2.5) {};
		\node [style=none] (52) at (-3.5, 2.5) {};
		\node [style=none] (53) at (-3.25, -3) {};
		\node [style=none] (54) at (-3.5, -3) {};
		\node [style=none] (55) at (-3, -3) {};
		\node [style=none] (56) at (-3.75, -2.5) {};
		\node [style=none] (57) at (-2.75, -2.5) {};
		\node [style=none] (58) at (-3.75, -2.5) {};
		\node [style=none] (59) at (-3.25, -2.5) {};
		\node [style=none] (60) at (-3, -2.5) {};
		\node [style=none] (61) at (-3.5, -2.5) {};
		\node [style=none] (62) at (-3.65, -2.5) {};
		\node [style=none] (63) at (-3.65, 2.5) {};
		\node [style=none] (64) at (-2.875, 2.5) {};
		\node [style=none] (65) at (-2.875, -2.5) {};
		\node [style=none] (66) at (-3.25, -3.75) {};
		\node [style=txtnobold] (67) at (-3, 0) {$i$};
	\end{pgfonlayer}
	\begin{pgfonlayer}{edgelayer}
		\draw [style=distribution] (3.center)
			 to (5.center)
			 to (7.center)
			 to (8.center)
			 to (6.center)
			 to (4.center)
			 to cycle;
		\draw [style=distribution] (13.center)
			 to (11.center)
			 to (12.center)
			 to (14.center)
			 to (16.center)
			 to (15.center)
			 to cycle;
		\draw [bend right=15] (25.center) to (23.center);
		\draw [bend left=15] (0.center) to (24.center);
		\draw [style=none, bend right=15] (26.center) to (27.center);
		\draw [style=none, bend left=15] (28.center) to (29.center);
		\draw [style=none] (30) to (31);
		\draw [style=none] (33) to (34);
		\draw [style=none] (34) to (32);
		\draw (19.center) to (15.center);
		\draw [bend left=15] (19.center) to (30);
		\draw [bend left=15] (18.center) to (31);
		\draw [bend left=15] (30) to (20.center);
		\draw [bend left=15] (31) to (10.center);
		\draw [bend right=15] (33) to (9.center);
		\draw [bend right=15] (32) to (21.center);
		\draw [bend left=15] (33) to (17.center);
		\draw [bend left=15] (32) to (22.center);
		\draw (38) to (36);
		\draw [style=distribution] (46.center)
			 to (44.center)
			 to (45.center)
			 to (47.center)
			 to (49.center)
			 to (50.center)
			 to (48.center)
			 to cycle;
		\draw [style=distribution] (59.center)
			 to (57.center)
			 to (55.center)
			 to (53.center)
			 to (54.center)
			 to (56.center)
			 to (58.center)
			 to cycle;
		\draw [in=75, out=-90] (38) to (60.center);
		\draw [in=105, out=-90, looseness=0.75] (37) to (61.center);
		\draw (35.center) to (44.center);
		\draw [in=90, out=-105, looseness=0.75] (52.center) to (37);
		\draw [in=90, out=-75] (51.center) to (36);
		\draw [in=150, out=-150, looseness=0.75] (63.center) to (62.center);
		\draw [in=75, out=-30, looseness=0.75] (64.center) to (39);
		\draw [in=45, out=-75, looseness=0.75] (39) to (65.center);
		\draw (53.center) to (66.center);
	\end{pgfonlayer}
\end{tikzpicture}

%% file: figures/interpretation.tikz
\begin{tikzpicture}
	\begin{pgfonlayer}{nodelayer}
		\node [style=none] (29) at (4.5, 0) {$\xrightarrow{\llbracket\cdot \rrbracket_\md}\ \ \rho\mapsto\sum_i p_i \llbracket D_i \rrbracket\ \rho\ \llbracket D_i \rrbracket^\dagger$};
		\node [style=none] (30) at (-2.875, 3) {};
		\node [style=none] (31) at (-1.875, 3) {};
		\node [style=txtnobold] (32) at (-2.375, 2.75) {$\cdots$};
		\node [style=none] (33) at (-2.875, -3) {};
		\node [style=none] (34) at (-1.875, -3) {};
		\node [style=txtnobold] (35) at (-2.375, -2.75) {$\cdots$};
		\node [style=none] (36) at (-2.375, 2) {};
		\node [style=none] (37) at (-2.625, 2) {};
		\node [style=none] (38) at (-2.125, 2) {};
		\node [style=none] (39) at (-2.875, 1.5) {};
		\node [style=none] (40) at (-1.875, 1.5) {};
		\node [style=none] (41) at (-2.875, 1.5) {};
		\node [style=none] (42) at (-2.375, 1.5) {};
		\node [style=none] (43) at (-2.375, -2) {};
		\node [style=none] (44) at (-2.625, -2) {};
		\node [style=none] (45) at (-2.125, -2) {};
		\node [style=none] (46) at (-2.875, -1.5) {};
		\node [style=none] (47) at (-1.875, -1.5) {};
		\node [style=none] (48) at (-2.875, -1.5) {};
		\node [style=none] (49) at (-2.375, -1.5) {};
		\node [style=none] (50) at (-2.125, 1.5) {};
		\node [style=none] (51) at (-2.625, 1.5) {};
		\node [style=none] (52) at (-2.125, -1.5) {};
		\node [style=none] (53) at (-2.625, -1.5) {};
		\node [style=morphism] (54) at (-3.375, 0) {$D_1$};
		\node [style=morphism] (55) at (-1.375, 0) {$D_k$};
		\node [style=txtnobold] (56) at (-2.375, 0.75) {$\cdots$};
		\node [style=txtnobold] (57) at (-3.625, 1.25) {$p_1$};
		\node [style=txtnobold] (58) at (-1.125, 1.25) {$p_k$};
		\node [style=none] (59) at (-2.7, 1.5) {};
		\node [style=none] (60) at (-4, 0.25) {};
		\node [style=none] (61) at (-2.875, 0.25) {};
		\node [style=txtnobold] (62) at (-3.25, 0.625) {$\cdots$};
		\node [style=txtnobold] (63) at (-3.25, 0.85) {$n$};
		\node [style=none] (64) at (-1.875, 0.25) {};
		\node [style=none] (65) at (-0.875, 0.25) {};
		\node [style=none] (66) at (-2, 1.5) {};
		\node [style=txtnobold] (67) at (-1.5, 0.625) {$\cdots$};
		\node [style=txtnobold] (68) at (-1.5, 0.85) {$n$};
		\node [style=none] (69) at (-2.7, -1.5) {};
		\node [style=none] (70) at (-2, -1.5) {};
		\node [style=none] (71) at (-2.875, -0.25) {};
		\node [style=none] (72) at (-4, -0.25) {};
		\node [style=txtnobold] (73) at (-3.25, -0.65) {$\cdots$};
		\node [style=txtnobold] (74) at (-3.125, -0.9) {$m$};
		\node [style=none] (75) at (-0.875, -0.25) {};
		\node [style=none] (76) at (-1.875, -0.25) {};
		\node [style=txtnobold] (77) at (-1.5, -0.65) {$\cdots$};
		\node [style=txtnobold] (78) at (-1.625, -0.9) {$m$};
		\node [style=txtnobold] (79) at (-2.375, -3) {$m$};
		\node [style=txtnobold] (80) at (-2.375, 3) {$n$};
	\end{pgfonlayer}
	\begin{pgfonlayer}{edgelayer}
		\draw [style=distribution] (39.center)
			 to (41.center)
			 to (42.center)
			 to (40.center)
			 to (38.center)
			 to (36.center)
			 to (37.center)
			 to cycle;
		\draw [style=distribution] (47.center)
			 to (45.center)
			 to (43.center)
			 to (44.center)
			 to (46.center)
			 to (48.center)
			 to (49.center)
			 to cycle;
		\draw [bend left=15] (60.center) to (59.center);
		\draw [bend left=15] (61.center) to (51.center);
		\draw [bend left=15] (66.center) to (65.center);
		\draw [bend right=15] (64.center) to (50.center);
		\draw [bend right=15] (72.center) to (69.center);
		\draw [bend right=15] (71.center) to (53.center);
		\draw [bend left=15] (75.center) to (70.center);
		\draw [bend left=15] (76.center) to (52.center);
		\draw [bend right=15] (44.center) to (33.center);
		\draw [bend left=15] (45.center) to (34.center);
		\draw [bend right=15] (30.center) to (37.center);
		\draw [bend left=15] (31.center) to (38.center);
	\end{pgfonlayer}
\end{tikzpicture}

%% file: figures/enrichedRules.tikz
\begin{tikzpicture}
	\begin{pgfonlayer}{nodelayer}
		\node [style=txtnobold] (43) at (-12.75, 6) {$\cdots$};
		\node [style=none] (48) at (-13.25, 0.75) {};
		\node [style=none] (49) at (-12.25, 0.75) {};
		\node [style=txtnobold] (50) at (-12.75, 1) {$\cdots$};
		\node [style=none] (56) at (-12.75, 10.75) {};
		\node [style=none] (57) at (-13, 10.75) {};
		\node [style=none] (58) at (-12.5, 10.75) {};
		\node [style=none] (59) at (-13.25, 10.25) {};
		\node [style=none] (60) at (-12.25, 10.25) {};
		\node [style=none] (61) at (-13.25, 10.25) {};
		\node [style=none] (62) at (-12.75, 10.25) {};
		\node [style=none] (63) at (-12.75, 6.75) {};
		\node [style=none] (64) at (-13, 6.75) {};
		\node [style=none] (65) at (-12.5, 6.75) {};
		\node [style=none] (66) at (-13.25, 7.25) {};
		\node [style=none] (67) at (-12.25, 7.25) {};
		\node [style=none] (68) at (-13.25, 7.25) {};
		\node [style=none] (69) at (-12.75, 7.25) {};
		\node [style=none] (70) at (-13.25, 11.75) {};
		\node [style=none] (71) at (-12.25, 11.75) {};
		\node [style=txtnobold] (72) at (-12.75, 11.5) {$\cdots$};
		\node [style=none] (73) at (-12.5, 10.25) {};
		\node [style=none] (74) at (-13, 10.25) {};
		\node [style=none] (75) at (-12.5, 7.25) {};
		\node [style=none] (76) at (-13, 7.25) {};
		\node [style=morphism] (77) at (-13.75, 8.75) {$D_1$};
		\node [style=morphism] (78) at (-11.75, 8.75) {$D_k$};
		\node [style=txtnobold] (79) at (-12.75, 9.5) {$\cdots$};
		\node [style=txtnobold] (80) at (-14, 10) {$p_1$};
		\node [style=txtnobold] (81) at (-11.5, 10) {$p_k$};
		\node [style=none] (96) at (-8.625, 10.25) {};
		\node [style=none] (97) at (-7.625, 10.25) {};
		\node [style=txtnobold] (98) at (-8.125, 10) {$\cdots$};
		\node [style=none] (103) at (-8.625, 2) {};
		\node [style=none] (104) at (-7.625, 2) {};
		\node [style=txtnobold] (105) at (-8.125, 2.25) {$\cdots$};
		\node [style=txtnobold] (123) at (-10.75, 6.25) {$\stackrel{(es)}{=}$};
		\node [style=none] (164) at (-1.625, 9.25) {};
		\node [style=none] (165) at (-0.625, 9.25) {};
		\node [style=txtnobold] (166) at (-1.125, 9) {$\cdots$};
		\node [style=txtnobold] (207) at (5.75, 6.25) {$\stackrel{(ep)}{=}$};
		\node [style=none] (210) at (-1.625, 3.25) {};
		\node [style=none] (211) at (-0.625, 3.25) {};
		\node [style=txtnobold] (212) at (-1.125, 3.5) {$\cdots$};
		\node [style=txtnobold] (247) at (-12, -4.25) {$\stackrel{(ec)}{=}$};
		\node [style=txtnobold] (304) at (-12.75, 11.75) {$n$};
		\node [style=none] (314) at (-13.075, 10.25) {};
		\node [style=none] (315) at (-14.375, 9) {};
		\node [style=none] (316) at (-13.25, 9) {};
		\node [style=txtnobold] (317) at (-13.625, 9.375) {$\cdots$};
		\node [style=txtnobold] (318) at (-13.625, 9.6) {$n$};
		\node [style=none] (319) at (-12.25, 9) {};
		\node [style=none] (320) at (-11.25, 9) {};
		\node [style=none] (321) at (-12.375, 10.25) {};
		\node [style=txtnobold] (322) at (-11.875, 9.375) {$\cdots$};
		\node [style=txtnobold] (323) at (-11.875, 9.6) {$n$};
		\node [style=none] (324) at (-13.075, 7.25) {};
		\node [style=none] (325) at (-12.375, 7.25) {};
		\node [style=none] (326) at (-13.25, 8.5) {};
		\node [style=none] (327) at (-14.375, 8.5) {};
		\node [style=txtnobold] (328) at (-13.625, 8.1) {$\cdots$};
		\node [style=txtnobold] (329) at (-13.5, 7.85) {$m$};
		\node [style=none] (330) at (-11.25, 8.5) {};
		\node [style=none] (331) at (-12.25, 8.5) {};
		\node [style=txtnobold] (332) at (-11.875, 8.1) {$\cdots$};
		\node [style=txtnobold] (333) at (-12, 7.85) {$m$};
		\node [style=txtnobold] (334) at (-12.75, 6.25) {$m$};
		\node [style=none] (335) at (-8.125, 9.25) {};
		\node [style=none] (336) at (-8.375, 9.25) {};
		\node [style=none] (337) at (-7.875, 9.25) {};
		\node [style=none] (338) at (-8.625, 8.75) {};
		\node [style=none] (339) at (-7.625, 8.75) {};
		\node [style=none] (340) at (-8.625, 8.75) {};
		\node [style=none] (341) at (-8.125, 8.75) {};
		\node [style=none] (342) at (-8.125, 3) {};
		\node [style=none] (343) at (-8.375, 3) {};
		\node [style=none] (344) at (-7.875, 3) {};
		\node [style=none] (345) at (-8.625, 3.5) {};
		\node [style=none] (346) at (-7.625, 3.5) {};
		\node [style=none] (347) at (-8.625, 3.5) {};
		\node [style=none] (348) at (-8.125, 3.5) {};
		\node [style=none] (349) at (-7.875, 8.75) {};
		\node [style=none] (350) at (-8.375, 8.75) {};
		\node [style=none] (351) at (-7.875, 3.5) {};
		\node [style=none] (352) at (-8.375, 3.5) {};
		\node [style=morphism] (353) at (-9.125, 7.25) {$D_1$};
		\node [style=morphism] (354) at (-7.125, 7.25) {$D_k$};
		\node [style=txtnobold] (355) at (-8.125, 6.25) {$\cdots$};
		\node [style=txtnobold] (356) at (-9.625, 8.5) {$p_1q_1$};
		\node [style=txtnobold] (357) at (-6.625, 8.5) {$p_kq_\ell$};
		\node [style=none] (358) at (-8.45, 8.75) {};
		\node [style=none] (359) at (-9.75, 7.5) {};
		\node [style=none] (360) at (-8.625, 7.5) {};
		\node [style=txtnobold] (361) at (-9, 7.875) {$\cdots$};
		\node [style=txtnobold] (362) at (-9, 8.1) {$n$};
		\node [style=none] (363) at (-7.625, 7.5) {};
		\node [style=none] (364) at (-6.625, 7.5) {};
		\node [style=none] (365) at (-7.75, 8.75) {};
		\node [style=txtnobold] (366) at (-7.25, 7.875) {$\cdots$};
		\node [style=txtnobold] (367) at (-7.25, 8.1) {$n$};
		\node [style=none] (368) at (-8.45, 3.5) {};
		\node [style=none] (369) at (-7.75, 3.5) {};
		\node [style=none] (370) at (-8.625, 4.75) {};
		\node [style=none] (371) at (-9.75, 4.75) {};
		\node [style=txtnobold] (372) at (-9, 4.35) {$\cdots$};
		\node [style=txtnobold] (373) at (-8.875, 4.1) {$r$};
		\node [style=none] (374) at (-6.625, 4.75) {};
		\node [style=none] (375) at (-7.625, 4.75) {};
		\node [style=txtnobold] (376) at (-7.25, 4.35) {$\cdots$};
		\node [style=txtnobold] (377) at (-7.375, 4.1) {$r$};
		\node [style=txtnobold] (378) at (-12.75, 0.75) {$r$};
		\node [style=morphism] (397) at (-9.125, 5.25) {$D_1'$};
		\node [style=morphism] (398) at (-7.125, 5.25) {$D_\ell'$};
		\node [style=none] (422) at (-1.125, 8.25) {};
		\node [style=none] (423) at (-1.375, 8.25) {};
		\node [style=none] (424) at (-0.875, 8.25) {};
		\node [style=none] (425) at (-1.625, 7.75) {};
		\node [style=none] (426) at (-0.625, 7.75) {};
		\node [style=none] (427) at (-1.625, 7.75) {};
		\node [style=none] (428) at (-1.125, 7.75) {};
		\node [style=none] (429) at (-1.125, 4.25) {};
		\node [style=none] (430) at (-1.375, 4.25) {};
		\node [style=none] (431) at (-0.875, 4.25) {};
		\node [style=none] (432) at (-1.625, 4.75) {};
		\node [style=none] (433) at (-0.625, 4.75) {};
		\node [style=none] (434) at (-1.625, 4.75) {};
		\node [style=none] (435) at (-1.125, 4.75) {};
		\node [style=none] (436) at (-0.875, 7.75) {};
		\node [style=none] (437) at (-1.375, 7.75) {};
		\node [style=none] (438) at (-0.875, 4.75) {};
		\node [style=none] (439) at (-1.375, 4.75) {};
		\node [style=morphism] (440) at (-2.125, 6.25) {$D_1$};
		\node [style=morphism] (441) at (-0.125, 6.25) {$D_k$};
		\node [style=txtnobold] (442) at (-1.125, 7) {$\cdots$};
		\node [style=txtnobold] (443) at (-2.375, 7.5) {$p_1$};
		\node [style=txtnobold] (444) at (0.125, 7.5) {$p_k$};
		\node [style=none] (445) at (-1.45, 7.75) {};
		\node [style=none] (446) at (-2.75, 6.5) {};
		\node [style=none] (447) at (-1.625, 6.5) {};
		\node [style=txtnobold] (448) at (-2, 6.875) {$\cdots$};
		\node [style=txtnobold] (449) at (-2, 7.1) {$n$};
		\node [style=none] (450) at (-0.625, 6.5) {};
		\node [style=none] (451) at (0.375, 6.5) {};
		\node [style=none] (452) at (-0.75, 7.75) {};
		\node [style=txtnobold] (453) at (-0.25, 6.875) {$\cdots$};
		\node [style=txtnobold] (454) at (-0.25, 7.1) {$n$};
		\node [style=none] (455) at (-1.45, 4.75) {};
		\node [style=none] (456) at (-0.75, 4.75) {};
		\node [style=none] (457) at (-1.625, 6) {};
		\node [style=none] (458) at (-2.75, 6) {};
		\node [style=txtnobold] (459) at (-2, 5.6) {$\cdots$};
		\node [style=txtnobold] (460) at (-1.875, 5.35) {$m$};
		\node [style=none] (461) at (0.375, 6) {};
		\node [style=none] (462) at (-0.625, 6) {};
		\node [style=txtnobold] (463) at (-0.25, 5.6) {$\cdots$};
		\node [style=txtnobold] (464) at (-0.375, 5.35) {$m$};
		\node [style=none] (465) at (-12.75, 5.75) {};
		\node [style=none] (466) at (-13, 5.75) {};
		\node [style=none] (467) at (-12.5, 5.75) {};
		\node [style=none] (468) at (-13.25, 5.25) {};
		\node [style=none] (469) at (-12.25, 5.25) {};
		\node [style=none] (470) at (-13.25, 5.25) {};
		\node [style=none] (471) at (-12.75, 5.25) {};
		\node [style=none] (472) at (-12.75, 1.75) {};
		\node [style=none] (473) at (-13, 1.75) {};
		\node [style=none] (474) at (-12.5, 1.75) {};
		\node [style=none] (475) at (-13.25, 2.25) {};
		\node [style=none] (476) at (-12.25, 2.25) {};
		\node [style=none] (477) at (-13.25, 2.25) {};
		\node [style=none] (478) at (-12.75, 2.25) {};
		\node [style=none] (479) at (-12.5, 5.25) {};
		\node [style=none] (480) at (-13, 5.25) {};
		\node [style=none] (481) at (-12.5, 2.25) {};
		\node [style=none] (482) at (-13, 2.25) {};
		\node [style=morphism] (483) at (-13.75, 3.75) {$D_1'$};
		\node [style=morphism] (484) at (-11.75, 3.75) {$D_\ell'$};
		\node [style=txtnobold] (485) at (-12.75, 4.5) {$\cdots$};
		\node [style=txtnobold] (486) at (-14, 5) {$q_1$};
		\node [style=txtnobold] (487) at (-11.5, 5) {$q_\ell$};
		\node [style=none] (488) at (-13.075, 5.25) {};
		\node [style=none] (489) at (-14.375, 4) {};
		\node [style=none] (490) at (-13.25, 4) {};
		\node [style=txtnobold] (491) at (-13.625, 4.375) {$\cdots$};
		\node [style=txtnobold] (492) at (-13.625, 4.6) {$m$};
		\node [style=none] (493) at (-12.25, 4) {};
		\node [style=none] (494) at (-11.25, 4) {};
		\node [style=none] (495) at (-12.375, 5.25) {};
		\node [style=txtnobold] (496) at (-11.875, 4.375) {$\cdots$};
		\node [style=txtnobold] (497) at (-11.875, 4.6) {$m$};
		\node [style=none] (498) at (-13.075, 2.25) {};
		\node [style=none] (499) at (-12.375, 2.25) {};
		\node [style=none] (500) at (-13.25, 3.5) {};
		\node [style=none] (501) at (-14.375, 3.5) {};
		\node [style=txtnobold] (502) at (-13.625, 3.1) {$\cdots$};
		\node [style=txtnobold] (503) at (-13.5, 2.85) {$r$};
		\node [style=none] (504) at (-11.25, 3.5) {};
		\node [style=none] (505) at (-12.25, 3.5) {};
		\node [style=txtnobold] (506) at (-11.875, 3.1) {$\cdots$};
		\node [style=txtnobold] (507) at (-12, 2.85) {$r$};
		\node [style=none] (680) at (-9.5, 6.75) {};
		\node [style=none] (681) at (-8.75, 6.75) {};
		\node [style=none] (682) at (-9.5, 5.75) {};
		\node [style=none] (683) at (-8.75, 5.75) {};
		\node [style=txtnobold] (684) at (-9.125, 6.5) {$\cdots$};
		\node [style=txtnobold] (685) at (-9.125, 6.25) {$m$};
		\node [style=none] (690) at (-7.5, 6.75) {};
		\node [style=none] (691) at (-7.5, 5.75) {};
		\node [style=none] (692) at (-6.75, 6.75) {};
		\node [style=none] (693) at (-6.75, 5.75) {};
		\node [style=txtnobold] (694) at (-7.125, 6.5) {$\cdots$};
		\node [style=txtnobold] (695) at (-7.125, 6.25) {$m$};
		\node [style=txtnobold] (696) at (-8.125, 2) {$r$};
		\node [style=txtnobold] (697) at (-8.125, 10.25) {$n$};
		\node [style=txtnobold] (698) at (-1.125, 3.25) {$m$};
		\node [style=txtnobold] (699) at (-1.125, 9.25) {$n$};
		\node [style=none] (700) at (2.625, 9.25) {};
		\node [style=none] (701) at (3.625, 9.25) {};
		\node [style=txtnobold] (702) at (3.125, 9) {$\cdots$};
		\node [style=none] (703) at (2.625, 3.25) {};
		\node [style=none] (704) at (3.625, 3.25) {};
		\node [style=txtnobold] (705) at (3.125, 3.5) {$\cdots$};
		\node [style=none] (706) at (3.125, 8.25) {};
		\node [style=none] (707) at (2.875, 8.25) {};
		\node [style=none] (708) at (3.375, 8.25) {};
		\node [style=none] (709) at (2.625, 7.75) {};
		\node [style=none] (710) at (3.625, 7.75) {};
		\node [style=none] (711) at (2.625, 7.75) {};
		\node [style=none] (712) at (3.125, 7.75) {};
		\node [style=none] (713) at (3.125, 4.25) {};
		\node [style=none] (714) at (2.875, 4.25) {};
		\node [style=none] (715) at (3.375, 4.25) {};
		\node [style=none] (716) at (2.625, 4.75) {};
		\node [style=none] (717) at (3.625, 4.75) {};
		\node [style=none] (718) at (2.625, 4.75) {};
		\node [style=none] (719) at (3.125, 4.75) {};
		\node [style=none] (720) at (3.375, 7.75) {};
		\node [style=none] (721) at (2.875, 7.75) {};
		\node [style=none] (722) at (3.375, 4.75) {};
		\node [style=none] (723) at (2.875, 4.75) {};
		\node [style=morphism] (724) at (2.125, 6.25) {$D_1'$};
		\node [style=morphism] (725) at (4.125, 6.25) {$D_\ell'$};
		\node [style=txtnobold] (726) at (3.125, 7) {$\cdots$};
		\node [style=txtnobold] (727) at (1.875, 7.5) {$q_1$};
		\node [style=txtnobold] (728) at (4.375, 7.5) {$q_\ell$};
		\node [style=none] (729) at (2.8, 7.75) {};
		\node [style=none] (730) at (1.5, 6.5) {};
		\node [style=none] (731) at (2.625, 6.5) {};
		\node [style=txtnobold] (732) at (2.25, 6.875) {$\cdots$};
		\node [style=txtnobold] (733) at (2.25, 7.1) {$r$};
		\node [style=none] (734) at (3.625, 6.5) {};
		\node [style=none] (735) at (4.625, 6.5) {};
		\node [style=none] (736) at (3.5, 7.75) {};
		\node [style=txtnobold] (737) at (4, 6.875) {$\cdots$};
		\node [style=txtnobold] (738) at (4, 7.1) {$r$};
		\node [style=none] (739) at (2.8, 4.75) {};
		\node [style=none] (740) at (3.5, 4.75) {};
		\node [style=none] (741) at (2.625, 6) {};
		\node [style=none] (742) at (1.5, 6) {};
		\node [style=txtnobold] (743) at (2.25, 5.6) {$\cdots$};
		\node [style=txtnobold] (744) at (2.375, 5.35) {$s$};
		\node [style=none] (745) at (4.625, 6) {};
		\node [style=none] (746) at (3.625, 6) {};
		\node [style=txtnobold] (747) at (4, 5.6) {$\cdots$};
		\node [style=txtnobold] (748) at (3.875, 5.35) {$s$};
		\node [style=txtnobold] (749) at (3.125, 3.25) {$s$};
		\node [style=txtnobold] (750) at (3.125, 9.25) {$r$};
		\node [style=none] (751) at (10, 10.75) {};
		\node [style=none] (752) at (11, 10.75) {};
		\node [style=txtnobold] (753) at (10.5, 10.5) {$\cdots$};
		\node [style=none] (754) at (10, 1.75) {};
		\node [style=none] (755) at (11, 1.75) {};
		\node [style=txtnobold] (756) at (10.5, 2) {$\cdots$};
		\node [style=none] (757) at (10.5, 9.75) {};
		\node [style=none] (758) at (10.25, 9.75) {};
		\node [style=none] (759) at (10.75, 9.75) {};
		\node [style=none] (760) at (10, 9.25) {};
		\node [style=none] (761) at (11, 9.25) {};
		\node [style=none] (762) at (10, 9.25) {};
		\node [style=none] (763) at (10.5, 9.25) {};
		\node [style=none] (764) at (10.5, 2.75) {};
		\node [style=none] (765) at (10.25, 2.75) {};
		\node [style=none] (766) at (10.75, 2.75) {};
		\node [style=none] (767) at (10, 3.25) {};
		\node [style=none] (768) at (11, 3.25) {};
		\node [style=none] (769) at (10, 3.25) {};
		\node [style=none] (770) at (10.5, 3.25) {};
		\node [style=none] (771) at (10.75, 9.25) {};
		\node [style=none] (772) at (10.25, 9.25) {};
		\node [style=none] (773) at (10.75, 3.25) {};
		\node [style=none] (774) at (10.25, 3.25) {};
		\node [style=morphism] (775) at (7.5, 6.25) {$D_1$};
		\node [style=txtnobold] (777) at (10.75, 6.25) {$\cdots$};
		\node [style=txtnobold] (778) at (8.75, 9) {$p_1q_1$};
		\node [style=none] (780) at (10.175, 9.25) {};
		\node [style=none] (781) at (6.875, 6.5) {};
		\node [style=none] (782) at (7.75, 6.5) {};
		\node [style=txtnobold] (783) at (7.625, 6.875) {$\cdots$};
		\node [style=txtnobold] (784) at (7.625, 7.1) {$n$};
		\node [style=none] (787) at (10.875, 9.25) {};
		\node [style=none] (790) at (10.175, 3.25) {};
		\node [style=none] (791) at (10.875, 3.25) {};
		\node [style=none] (792) at (7.75, 6) {};
		\node [style=none] (793) at (6.875, 6) {};
		\node [style=txtnobold] (794) at (7.625, 5.6) {$\cdots$};
		\node [style=txtnobold] (795) at (7.75, 5.35) {$m$};
		\node [style=txtnobold] (800) at (10.5, 1.5) {$m+s$};
		\node [style=txtnobold] (801) at (10.5, 11) {$n+r$};
		\node [style=morphism] (802) at (9, 6.25) {$D_1'$};
		\node [style=none] (803) at (8.375, 6.5) {};
		\node [style=none] (804) at (9.5, 6.5) {};
		\node [style=txtnobold] (805) at (9, 6.875) {$\cdots$};
		\node [style=txtnobold] (806) at (9.125, 7.1) {$r$};
		\node [style=none] (807) at (9.5, 6) {};
		\node [style=none] (808) at (8.375, 6) {};
		\node [style=txtnobold] (809) at (9, 5.6) {$\cdots$};
		\node [style=txtnobold] (810) at (9, 5.35) {$s$};
		\node [style=morphism] (811) at (12.25, 6.25) {$D_k$};
		\node [style=none] (812) at (11.875, 6.5) {};
		\node [style=none] (813) at (12.75, 6.5) {};
		\node [style=txtnobold] (814) at (12.25, 6.875) {$\cdots$};
		\node [style=txtnobold] (815) at (12.125, 7.1) {$n$};
		\node [style=none] (816) at (12.75, 6) {};
		\node [style=none] (817) at (11.875, 6) {};
		\node [style=txtnobold] (818) at (12.25, 5.6) {$\cdots$};
		\node [style=txtnobold] (819) at (12.25, 5.35) {$m$};
		\node [style=morphism] (820) at (13.75, 6.25) {$D_\ell'$};
		\node [style=none] (821) at (13.375, 6.5) {};
		\node [style=none] (822) at (14.375, 6.5) {};
		\node [style=txtnobold] (823) at (13.75, 6.875) {$\cdots$};
		\node [style=txtnobold] (824) at (13.625, 7.1) {$r$};
		\node [style=none] (825) at (14.375, 6) {};
		\node [style=none] (826) at (13.375, 6) {};
		\node [style=txtnobold] (827) at (13.75, 5.6) {$\cdots$};
		\node [style=txtnobold] (828) at (13.5, 5.35) {$s$};
		\node [style=txtnobold] (829) at (12.5, 9) {$p_kq_\ell$};
		\node [style=none] (830) at (-15, -1.25) {};
		\node [style=none] (831) at (-14, -1.25) {};
		\node [style=txtnobold] (832) at (-14.5, -1.5) {$\cdots$};
		\node [style=none] (833) at (-15, -7.25) {};
		\node [style=none] (834) at (-14, -7.25) {};
		\node [style=txtnobold] (835) at (-14.5, -7) {$\cdots$};
		\node [style=none] (836) at (-14.5, -2.25) {};
		\node [style=none] (837) at (-14.75, -2.25) {};
		\node [style=none] (838) at (-14.25, -2.25) {};
		\node [style=none] (839) at (-15, -2.75) {};
		\node [style=none] (840) at (-14, -2.75) {};
		\node [style=none] (841) at (-15, -2.75) {};
		\node [style=none] (842) at (-14.5, -2.75) {};
		\node [style=none] (843) at (-14.5, -6.25) {};
		\node [style=none] (844) at (-14.75, -6.25) {};
		\node [style=none] (845) at (-14.25, -6.25) {};
		\node [style=none] (846) at (-15, -5.75) {};
		\node [style=none] (847) at (-14, -5.75) {};
		\node [style=none] (848) at (-15, -5.75) {};
		\node [style=none] (849) at (-14.5, -5.75) {};
		\node [style=none] (850) at (-14.25, -2.75) {};
		\node [style=none] (851) at (-14.75, -2.75) {};
		\node [style=none] (852) at (-14.25, -5.75) {};
		\node [style=none] (853) at (-14.75, -5.75) {};
		\node [style=morphism] (854) at (-15.5, -4.25) {$D_1$};
		\node [style=morphism] (855) at (-13.5, -4.25) {$D_2$};
		\node [style=txtnobold] (857) at (-15.75, -3) {$p_1$};
		\node [style=txtnobold] (858) at (-13.25, -3) {$p_2$};
		\node [style=none] (859) at (-14.825, -2.75) {};
		\node [style=none] (860) at (-16.125, -4) {};
		\node [style=none] (861) at (-15, -4) {};
		\node [style=txtnobold] (862) at (-15.375, -3.625) {$\cdots$};
		\node [style=txtnobold] (863) at (-15.375, -3.4) {$n$};
		\node [style=none] (864) at (-14, -4) {};
		\node [style=none] (865) at (-13, -4) {};
		\node [style=none] (866) at (-14.125, -2.75) {};
		\node [style=txtnobold] (867) at (-13.625, -3.625) {$\cdots$};
		\node [style=txtnobold] (868) at (-13.625, -3.4) {$n$};
		\node [style=none] (869) at (-14.825, -5.75) {};
		\node [style=none] (870) at (-14.125, -5.75) {};
		\node [style=none] (871) at (-15, -4.5) {};
		\node [style=none] (872) at (-16.125, -4.5) {};
		\node [style=txtnobold] (873) at (-15.375, -4.9) {$\cdots$};
		\node [style=txtnobold] (874) at (-15.25, -5.15) {$m$};
		\node [style=none] (875) at (-13, -4.5) {};
		\node [style=none] (876) at (-14, -4.5) {};
		\node [style=txtnobold] (877) at (-13.625, -4.9) {$\cdots$};
		\node [style=txtnobold] (878) at (-13.75, -5.15) {$m$};
		\node [style=txtnobold] (879) at (-14.5, -7.25) {$m$};
		\node [style=txtnobold] (880) at (-14.5, -1.25) {$n$};
		\node [style=none] (881) at (-9.75, -1.25) {};
		\node [style=none] (882) at (-8.75, -1.25) {};
		\node [style=txtnobold] (883) at (-9.25, -1.5) {$\cdots$};
		\node [style=none] (884) at (-9.75, -7.25) {};
		\node [style=none] (885) at (-8.75, -7.25) {};
		\node [style=txtnobold] (886) at (-9.25, -7) {$\cdots$};
		\node [style=none] (887) at (-9.25, -2.25) {};
		\node [style=none] (888) at (-9.5, -2.25) {};
		\node [style=none] (889) at (-9, -2.25) {};
		\node [style=none] (890) at (-9.75, -2.75) {};
		\node [style=none] (891) at (-8.75, -2.75) {};
		\node [style=none] (892) at (-9.75, -2.75) {};
		\node [style=none] (893) at (-9.25, -2.75) {};
		\node [style=none] (894) at (-9.25, -6.25) {};
		\node [style=none] (895) at (-9.5, -6.25) {};
		\node [style=none] (896) at (-9, -6.25) {};
		\node [style=none] (897) at (-9.75, -5.75) {};
		\node [style=none] (898) at (-8.75, -5.75) {};
		\node [style=none] (899) at (-9.75, -5.75) {};
		\node [style=none] (900) at (-9.25, -5.75) {};
		\node [style=none] (901) at (-9, -2.75) {};
		\node [style=none] (902) at (-9.5, -2.75) {};
		\node [style=none] (903) at (-9, -5.75) {};
		\node [style=none] (904) at (-9.5, -5.75) {};
		\node [style=morphism] (905) at (-10.25, -4.25) {$D_2$};
		\node [style=morphism] (906) at (-8.25, -4.25) {$D_1$};
		\node [style=txtnobold] (908) at (-10.5, -3) {$p_2$};
		\node [style=txtnobold] (909) at (-8, -3) {$p_1$};
		\node [style=none] (910) at (-9.575, -2.75) {};
		\node [style=none] (911) at (-10.875, -4) {};
		\node [style=none] (912) at (-9.75, -4) {};
		\node [style=txtnobold] (913) at (-10.125, -3.625) {$\cdots$};
		\node [style=txtnobold] (914) at (-10.125, -3.4) {$n$};
		\node [style=none] (915) at (-8.75, -4) {};
		\node [style=none] (916) at (-7.75, -4) {};
		\node [style=none] (917) at (-8.875, -2.75) {};
		\node [style=txtnobold] (918) at (-8.375, -3.625) {$\cdots$};
		\node [style=txtnobold] (919) at (-8.375, -3.4) {$n$};
		\node [style=none] (920) at (-9.575, -5.75) {};
		\node [style=none] (921) at (-8.875, -5.75) {};
		\node [style=none] (922) at (-9.75, -4.5) {};
		\node [style=none] (923) at (-10.875, -4.5) {};
		\node [style=txtnobold] (924) at (-10.125, -4.9) {$\cdots$};
		\node [style=txtnobold] (925) at (-10, -5.15) {$m$};
		\node [style=none] (926) at (-7.75, -4.5) {};
		\node [style=none] (927) at (-8.75, -4.5) {};
		\node [style=txtnobold] (928) at (-8.375, -4.9) {$\cdots$};
		\node [style=txtnobold] (929) at (-8.5, -5.15) {$m$};
		\node [style=txtnobold] (930) at (-9.25, -7.25) {$m$};
		\node [style=txtnobold] (931) at (-9.25, -1.25) {$n$};
		\node [style=none] (932) at (-5.25, -1.25) {};
		\node [style=none] (933) at (-4.25, -1.25) {};
		\node [style=txtnobold] (934) at (-4.75, -1.5) {$\cdots$};
		\node [style=none] (935) at (-5.25, -7.25) {};
		\node [style=none] (936) at (-4.25, -7.25) {};
		\node [style=txtnobold] (937) at (-4.75, -7) {$\cdots$};
		\node [style=none] (938) at (-4.75, -2.25) {};
		\node [style=none] (939) at (-5, -2.25) {};
		\node [style=none] (940) at (-4.5, -2.25) {};
		\node [style=none] (941) at (-5.25, -2.75) {};
		\node [style=none] (942) at (-4.25, -2.75) {};
		\node [style=none] (943) at (-5.25, -2.75) {};
		\node [style=none] (944) at (-4.75, -2.75) {};
		\node [style=none] (945) at (-4.75, -6.25) {};
		\node [style=none] (946) at (-5, -6.25) {};
		\node [style=none] (947) at (-4.5, -6.25) {};
		\node [style=none] (948) at (-5.25, -5.75) {};
		\node [style=none] (949) at (-4.25, -5.75) {};
		\node [style=none] (950) at (-5.25, -5.75) {};
		\node [style=none] (951) at (-4.75, -5.75) {};
		\node [style=none] (952) at (-4.5, -2.75) {};
		\node [style=none] (953) at (-5, -2.75) {};
		\node [style=none] (954) at (-4.5, -5.75) {};
		\node [style=none] (955) at (-5, -5.75) {};
		\node [style=morphism] (956) at (-4.75, -4.25) {$D$};
		\node [style=txtnobold] (959) at (-5.5, -3.25) {$1$};
		\node [style=txtnobold] (981) at (-4.75, -7.25) {$m$};
		\node [style=txtnobold] (982) at (-4.75, -1.25) {$n$};
		\node [style=none] (983) at (-5, -4) {};
		\node [style=none] (984) at (-5, -4.5) {};
		\node [style=none] (985) at (-4.5, -4) {};
		\node [style=none] (986) at (-4.5, -4.5) {};
		\node [style=txtnobold] (987) at (-4.75, -3.5) {$\cdots$};
		\node [style=txtnobold] (988) at (-4.75, -5) {$\cdots$};
		\node [style=txtnobold] (989) at (-3.5, -4.25) {$\stackrel{(e\delta)}{=}$};
		\node [style=morphism] (990) at (-2.25, -4.25) {$D$};
		\node [style=none] (991) at (-2, -4.5) {};
		\node [style=none] (992) at (-2.5, -4.5) {};
		\node [style=none] (993) at (-2.75, -5.5) {};
		\node [style=none] (994) at (-1.75, -5.5) {};
		\node [style=none] (995) at (-2.5, -4) {};
		\node [style=none] (996) at (-2, -4) {};
		\node [style=none] (997) at (-2.75, -3) {};
		\node [style=none] (998) at (-1.75, -3) {};
		\node [style=txtnobold] (999) at (-2.25, -3) {$n$};
		\node [style=txtnobold] (1000) at (-2.25, -3.25) {$\cdots$};
		\node [style=txtnobold] (1001) at (-2.25, -5.25) {$\cdots$};
		\node [style=txtnobold] (1002) at (-2.25, -5.5) {$m$};
		\node [style=none] (1003) at (2, -1.25) {};
		\node [style=none] (1004) at (3, -1.25) {};
		\node [style=txtnobold] (1005) at (2.5, -1.5) {$\cdots$};
		\node [style=none] (1006) at (2, -7.25) {};
		\node [style=none] (1007) at (3, -7.25) {};
		\node [style=txtnobold] (1008) at (2.5, -7) {$\cdots$};
		\node [style=none] (1009) at (2.5, -2.25) {};
		\node [style=none] (1010) at (2.25, -2.25) {};
		\node [style=none] (1011) at (2.75, -2.25) {};
		\node [style=none] (1012) at (2, -2.75) {};
		\node [style=none] (1013) at (3, -2.75) {};
		\node [style=none] (1014) at (2, -2.75) {};
		\node [style=none] (1015) at (2.5, -2.75) {};
		\node [style=none] (1016) at (2.5, -6.25) {};
		\node [style=none] (1017) at (2.25, -6.25) {};
		\node [style=none] (1018) at (2.75, -6.25) {};
		\node [style=none] (1019) at (2, -5.75) {};
		\node [style=none] (1020) at (3, -5.75) {};
		\node [style=none] (1021) at (2, -5.75) {};
		\node [style=none] (1022) at (2.5, -5.75) {};
		\node [style=none] (1023) at (2.75, -2.75) {};
		\node [style=none] (1024) at (2.25, -2.75) {};
		\node [style=none] (1025) at (2.75, -5.75) {};
		\node [style=none] (1026) at (2.25, -5.75) {};
		\node [style=morphism] (1027) at (1.5, -4.25) {$D$};
		\node [style=morphism] (1028) at (3.5, -4.25) {$D$};
		\node [style=txtnobold] (1030) at (1.25, -3) {$p_1$};
		\node [style=txtnobold] (1031) at (3.75, -3) {$p_2$};
		\node [style=none] (1032) at (2.175, -2.75) {};
		\node [style=none] (1033) at (1.125, -4) {};
		\node [style=none] (1034) at (2, -4) {};
		\node [style=txtnobold] (1035) at (1.625, -3.625) {$\cdots$};
		\node [style=txtnobold] (1036) at (1.675, -3.4) {$n$};
		\node [style=none] (1037) at (3, -4) {};
		\node [style=none] (1038) at (4, -4) {};
		\node [style=none] (1039) at (2.875, -2.75) {};
		\node [style=txtnobold] (1040) at (3.375, -3.625) {$\cdots$};
		\node [style=txtnobold] (1041) at (3.375, -3.4) {$n$};
		\node [style=none] (1042) at (2.175, -5.75) {};
		\node [style=none] (1043) at (2.875, -5.75) {};
		\node [style=none] (1044) at (2, -4.5) {};
		\node [style=none] (1045) at (1.125, -4.5) {};
		\node [style=txtnobold] (1046) at (1.625, -4.9) {$\cdots$};
		\node [style=txtnobold] (1047) at (1.75, -5.15) {$m$};
		\node [style=none] (1048) at (4, -4.5) {};
		\node [style=none] (1049) at (3, -4.5) {};
		\node [style=txtnobold] (1050) at (3.375, -4.9) {$\cdots$};
		\node [style=txtnobold] (1051) at (3.25, -5.15) {$m$};
		\node [style=txtnobold] (1052) at (2.5, -7.25) {$m$};
		\node [style=txtnobold] (1053) at (2.5, -1.25) {$n$};
		\node [style=none] (1054) at (5.625, -1.25) {};
		\node [style=none] (1055) at (6.625, -1.25) {};
		\node [style=txtnobold] (1056) at (6.125, -1.5) {$\cdots$};
		\node [style=none] (1057) at (5.625, -7.25) {};
		\node [style=none] (1058) at (6.625, -7.25) {};
		\node [style=txtnobold] (1059) at (6.125, -7) {$\cdots$};
		\node [style=none] (1060) at (6.125, -2.25) {};
		\node [style=none] (1061) at (5.875, -2.25) {};
		\node [style=none] (1062) at (6.375, -2.25) {};
		\node [style=none] (1063) at (5.625, -2.75) {};
		\node [style=none] (1064) at (6.625, -2.75) {};
		\node [style=none] (1065) at (5.625, -2.75) {};
		\node [style=none] (1066) at (6.125, -2.75) {};
		\node [style=none] (1067) at (6.125, -6.25) {};
		\node [style=none] (1068) at (5.875, -6.25) {};
		\node [style=none] (1069) at (6.375, -6.25) {};
		\node [style=none] (1070) at (5.625, -5.75) {};
		\node [style=none] (1071) at (6.625, -5.75) {};
		\node [style=none] (1072) at (5.625, -5.75) {};
		\node [style=none] (1073) at (6.125, -5.75) {};
		\node [style=none] (1074) at (6.375, -2.75) {};
		\node [style=none] (1075) at (5.875, -2.75) {};
		\node [style=none] (1076) at (6.375, -5.75) {};
		\node [style=none] (1077) at (5.875, -5.75) {};
		\node [style=morphism] (1078) at (6.125, -4.25) {$D$};
		\node [style=txtnobold] (1079) at (7.375, -3) {$p_1+p_2$};
		\node [style=txtnobold] (1080) at (6.125, -7.25) {$m$};
		\node [style=txtnobold] (1081) at (6.125, -1.25) {$n$};
		\node [style=none] (1082) at (5.875, -4) {};
		\node [style=none] (1083) at (5.875, -4.5) {};
		\node [style=none] (1084) at (6.375, -4) {};
		\node [style=none] (1085) at (6.375, -4.5) {};
		\node [style=txtnobold] (1086) at (6.125, -3.25) {$\cdots$};
		\node [style=txtnobold] (1087) at (6.125, -5.25) {$\cdots$};
		\node [style=txtnobold] (1088) at (4.75, -4.25) {$\stackrel{(e+)}{=}$};
		\node [style=none] (1089) at (10.5, -1.25) {};
		\node [style=none] (1090) at (11.5, -1.25) {};
		\node [style=txtnobold] (1091) at (11, -1.5) {$\cdots$};
		\node [style=none] (1092) at (10.5, -7.25) {};
		\node [style=none] (1093) at (11.5, -7.25) {};
		\node [style=txtnobold] (1094) at (11, -7) {$\cdots$};
		\node [style=none] (1095) at (11, -2.25) {};
		\node [style=none] (1096) at (10.75, -2.25) {};
		\node [style=none] (1097) at (11.25, -2.25) {};
		\node [style=none] (1098) at (10.5, -2.75) {};
		\node [style=none] (1099) at (11.5, -2.75) {};
		\node [style=none] (1100) at (10.5, -2.75) {};
		\node [style=none] (1101) at (11, -2.75) {};
		\node [style=none] (1102) at (11, -6.25) {};
		\node [style=none] (1103) at (10.75, -6.25) {};
		\node [style=none] (1104) at (11.25, -6.25) {};
		\node [style=none] (1105) at (10.5, -5.75) {};
		\node [style=none] (1106) at (11.5, -5.75) {};
		\node [style=none] (1107) at (10.5, -5.75) {};
		\node [style=none] (1108) at (11, -5.75) {};
		\node [style=none] (1109) at (11.25, -2.75) {};
		\node [style=none] (1110) at (10.75, -2.75) {};
		\node [style=none] (1111) at (11.25, -5.75) {};
		\node [style=none] (1112) at (10.75, -5.75) {};
		\node [style=morphism] (1113) at (10, -4.25) {$D_1$};
		\node [style=morphism] (1114) at (12, -4.25) {$D_2$};
		\node [style=txtnobold] (1115) at (9.75, -3) {$p$};
		\node [style=txtnobold] (1116) at (12.25, -3) {$0$};
		\node [style=none] (1117) at (10.675, -2.75) {};
		\node [style=none] (1118) at (9.375, -4) {};
		\node [style=none] (1119) at (10.5, -4) {};
		\node [style=txtnobold] (1120) at (10.125, -3.625) {$\cdots$};
		\node [style=txtnobold] (1121) at (10.145, -3.4) {$n$};
		\node [style=none] (1122) at (11.5, -4) {};
		\node [style=none] (1123) at (12.5, -4) {};
		\node [style=none] (1124) at (11.375, -2.75) {};
		\node [style=txtnobold] (1125) at (11.875, -3.625) {$\cdots$};
		\node [style=txtnobold] (1126) at (11.875, -3.4) {$n$};
		\node [style=none] (1127) at (10.675, -5.75) {};
		\node [style=none] (1128) at (11.375, -5.75) {};
		\node [style=none] (1129) at (10.5, -4.5) {};
		\node [style=none] (1130) at (9.375, -4.5) {};
		\node [style=txtnobold] (1131) at (10.125, -4.9) {$\cdots$};
		\node [style=txtnobold] (1132) at (10.25, -5.15) {$m$};
		\node [style=none] (1133) at (12.5, -4.5) {};
		\node [style=none] (1134) at (11.5, -4.5) {};
		\node [style=txtnobold] (1135) at (11.875, -4.9) {$\cdots$};
		\node [style=txtnobold] (1136) at (11.75, -5.15) {$m$};
		\node [style=txtnobold] (1137) at (11, -7.25) {$m$};
		\node [style=txtnobold] (1138) at (11, -1.25) {$n$};
		\node [style=none] (1139) at (13.875, -1.25) {};
		\node [style=none] (1140) at (14.875, -1.25) {};
		\node [style=txtnobold] (1141) at (14.375, -1.5) {$\cdots$};
		\node [style=none] (1142) at (13.875, -7.25) {};
		\node [style=none] (1143) at (14.875, -7.25) {};
		\node [style=txtnobold] (1144) at (14.375, -7) {$\cdots$};
		\node [style=none] (1145) at (14.375, -2.25) {};
		\node [style=none] (1146) at (14.125, -2.25) {};
		\node [style=none] (1147) at (14.625, -2.25) {};
		\node [style=none] (1148) at (13.875, -2.75) {};
		\node [style=none] (1149) at (14.875, -2.75) {};
		\node [style=none] (1150) at (13.875, -2.75) {};
		\node [style=none] (1151) at (14.375, -2.75) {};
		\node [style=none] (1152) at (14.375, -6.25) {};
		\node [style=none] (1153) at (14.125, -6.25) {};
		\node [style=none] (1154) at (14.625, -6.25) {};
		\node [style=none] (1155) at (13.875, -5.75) {};
		\node [style=none] (1156) at (14.875, -5.75) {};
		\node [style=none] (1157) at (13.875, -5.75) {};
		\node [style=none] (1158) at (14.375, -5.75) {};
		\node [style=none] (1159) at (14.625, -2.75) {};
		\node [style=none] (1160) at (14.125, -2.75) {};
		\node [style=none] (1161) at (14.625, -5.75) {};
		\node [style=none] (1162) at (14.125, -5.75) {};
		\node [style=morphism] (1163) at (14.375, -4.25) {$D_1$};
		\node [style=txtnobold] (1164) at (13.625, -3) {$p$};
		\node [style=txtnobold] (1165) at (14.375, -7.25) {$m$};
		\node [style=txtnobold] (1166) at (14.375, -1.25) {$n$};
		\node [style=none] (1167) at (14.125, -4) {};
		\node [style=none] (1168) at (14.125, -4.5) {};
		\node [style=none] (1169) at (14.625, -4) {};
		\node [style=none] (1170) at (14.625, -4.5) {};
		\node [style=txtnobold] (1171) at (14.375, -3.25) {$\cdots$};
		\node [style=txtnobold] (1172) at (14.375, -5.25) {$\cdots$};
		\node [style=txtnobold] (1173) at (13.25, -4.25) {$\stackrel{(e0)}{=}$};
		\node [style=bigtxt] (1174) at (-12, 12.75) {Additional enriched ruleset};
		\node [style=none] (1175) at (-17.25, -8.25) {};
		\node [style=none] (1176) at (-17.25, 12.75) {};
		\node [style=none] (1177) at (-15.5, 12.75) {};
		\node [style=none] (1178) at (16, -8.25) {};
		\node [style=none] (1179) at (16, 12.75) {};
		\node [style=none] (1180) at (-8.5, 12.75) {};
	\end{pgfonlayer}
	\begin{pgfonlayer}{edgelayer}
		\draw [style=distribution] (62.center)
			 to (60.center)
			 to (58.center)
			 to (56.center)
			 to (57.center)
			 to (59.center)
			 to (61.center)
			 to cycle;
		\draw [style=distribution] (69.center)
			 to (67.center)
			 to (65.center)
			 to (63.center)
			 to (64.center)
			 to (66.center)
			 to (68.center)
			 to cycle;
		\draw [style=none, bend right=15] (70.center) to (57.center);
		\draw [style=none, bend left=15] (71.center) to (58.center);
		\draw [bend left=15] (315.center) to (314.center);
		\draw [bend left=15] (316.center) to (74.center);
		\draw [bend left=15] (321.center) to (320.center);
		\draw [bend right=15] (319.center) to (73.center);
		\draw [bend right=15] (327.center) to (324.center);
		\draw [bend right=15] (326.center) to (76.center);
		\draw [bend left=15] (330.center) to (325.center);
		\draw [bend left=15] (331.center) to (75.center);
		\draw [style=distribution] (341.center)
			 to (339.center)
			 to (337.center)
			 to (335.center)
			 to (336.center)
			 to (338.center)
			 to (340.center)
			 to cycle;
		\draw [style=distribution] (348.center)
			 to (346.center)
			 to (344.center)
			 to (342.center)
			 to (343.center)
			 to (345.center)
			 to (347.center)
			 to cycle;
		\draw [bend left=15] (359.center) to (358.center);
		\draw [bend left=15] (360.center) to (350.center);
		\draw [bend left=15] (365.center) to (364.center);
		\draw [bend right=15] (363.center) to (349.center);
		\draw [bend right=15] (371.center) to (368.center);
		\draw [bend right=15] (370.center) to (352.center);
		\draw [bend left=15] (374.center) to (369.center);
		\draw [bend left=15] (375.center) to (351.center);
		\draw [style=distribution] (425.center)
			 to (427.center)
			 to (428.center)
			 to (426.center)
			 to (424.center)
			 to (422.center)
			 to (423.center)
			 to cycle;
		\draw [style=distribution] (433.center)
			 to (431.center)
			 to (429.center)
			 to (430.center)
			 to (432.center)
			 to (434.center)
			 to (435.center)
			 to cycle;
		\draw [bend left=15] (446.center) to (445.center);
		\draw [bend left=15] (447.center) to (437.center);
		\draw [bend left=15] (452.center) to (451.center);
		\draw [bend right=15] (450.center) to (436.center);
		\draw [bend right=15] (458.center) to (455.center);
		\draw [bend right=15] (457.center) to (439.center);
		\draw [bend left=15] (461.center) to (456.center);
		\draw [bend left=15] (462.center) to (438.center);
		\draw [style=distribution] (466.center)
			 to (468.center)
			 to (470.center)
			 to (471.center)
			 to (469.center)
			 to (467.center)
			 to (465.center)
			 to cycle;
		\draw [style=distribution] (474.center)
			 to (472.center)
			 to (473.center)
			 to (475.center)
			 to (477.center)
			 to (478.center)
			 to (476.center)
			 to cycle;
		\draw [bend left=15] (489.center) to (488.center);
		\draw [bend left=15] (490.center) to (480.center);
		\draw [bend left=15] (495.center) to (494.center);
		\draw [bend right=15] (493.center) to (479.center);
		\draw [bend right=15] (501.center) to (498.center);
		\draw [bend right=15] (500.center) to (482.center);
		\draw [bend left=15] (504.center) to (499.center);
		\draw [bend left=15] (505.center) to (481.center);
		\draw [bend right=15] (64.center) to (466.center);
		\draw [bend left=15] (65.center) to (467.center);
		\draw [bend right=15] (473.center) to (48.center);
		\draw [bend left=15] (474.center) to (49.center);
		\draw (680.center) to (682.center);
		\draw (681.center) to (683.center);
		\draw (690.center) to (691.center);
		\draw (692.center) to (693.center);
		\draw [bend right=15] (343.center) to (103.center);
		\draw [bend left=15] (344.center) to (104.center);
		\draw [bend right=15] (96.center) to (336.center);
		\draw [bend left=15] (97.center) to (337.center);
		\draw [bend right=15] (430.center) to (210.center);
		\draw [bend left=15] (431.center) to (211.center);
		\draw [bend right=15] (164.center) to (423.center);
		\draw [bend left=15] (165.center) to (424.center);
		\draw [style=distribution] (709.center)
			 to (711.center)
			 to (712.center)
			 to (710.center)
			 to (708.center)
			 to (706.center)
			 to (707.center)
			 to cycle;
		\draw [style=distribution] (717.center)
			 to (715.center)
			 to (713.center)
			 to (714.center)
			 to (716.center)
			 to (718.center)
			 to (719.center)
			 to cycle;
		\draw [bend left=15] (730.center) to (729.center);
		\draw [bend left=15] (731.center) to (721.center);
		\draw [bend left=15] (736.center) to (735.center);
		\draw [bend right=15] (734.center) to (720.center);
		\draw [bend right=15] (742.center) to (739.center);
		\draw [bend right=15] (741.center) to (723.center);
		\draw [bend left=15] (745.center) to (740.center);
		\draw [bend left=15] (746.center) to (722.center);
		\draw [bend right=15] (714.center) to (703.center);
		\draw [bend left=15] (715.center) to (704.center);
		\draw [bend right=15] (700.center) to (707.center);
		\draw [bend left=15] (701.center) to (708.center);
		\draw [style=distribution] (759.center)
			 to (757.center)
			 to (758.center)
			 to (760.center)
			 to (762.center)
			 to (763.center)
			 to (761.center)
			 to cycle;
		\draw [style=distribution] (765.center)
			 to (767.center)
			 to (769.center)
			 to (770.center)
			 to (768.center)
			 to (766.center)
			 to (764.center)
			 to cycle;
		\draw [bend left=15] (781.center) to (780.center);
		\draw [bend right=15] (793.center) to (790.center);
		\draw [bend right=15] (765.center) to (754.center);
		\draw [bend left=15] (766.center) to (755.center);
		\draw [bend right=15] (751.center) to (758.center);
		\draw [bend left=15] (752.center) to (759.center);
		\draw [bend left=15] (782.center) to (780.center);
		\draw [bend left=15] (803.center) to (772.center);
		\draw [bend left=15] (804.center) to (772.center);
		\draw [bend right=15] (792.center) to (790.center);
		\draw [bend right=15] (808.center) to (774.center);
		\draw [bend right=15] (807.center) to (774.center);
		\draw [bend right=15] (822.center) to (787.center);
		\draw [bend right=15] (821.center) to (787.center);
		\draw [bend right=15] (813.center) to (771.center);
		\draw [bend right=15] (812.center) to (771.center);
		\draw [bend left=15] (825.center) to (791.center);
		\draw [bend left=15] (826.center) to (791.center);
		\draw [bend left=15] (816.center) to (773.center);
		\draw [bend left=15] (817.center) to (773.center);
		\draw [style=distribution] (839.center)
			 to (841.center)
			 to (842.center)
			 to (840.center)
			 to (838.center)
			 to (836.center)
			 to (837.center)
			 to cycle;
		\draw [style=distribution] (847.center)
			 to (845.center)
			 to (843.center)
			 to (844.center)
			 to (846.center)
			 to (848.center)
			 to (849.center)
			 to cycle;
		\draw [bend left=15] (860.center) to (859.center);
		\draw [bend left=15] (861.center) to (851.center);
		\draw [bend left=15] (866.center) to (865.center);
		\draw [bend right=15] (864.center) to (850.center);
		\draw [bend right=15] (872.center) to (869.center);
		\draw [bend right=15] (871.center) to (853.center);
		\draw [bend left=15] (875.center) to (870.center);
		\draw [bend left=15] (876.center) to (852.center);
		\draw [bend right=15] (844.center) to (833.center);
		\draw [bend left=15] (845.center) to (834.center);
		\draw [bend right=15] (830.center) to (837.center);
		\draw [bend left=15] (831.center) to (838.center);
		\draw [style=distribution] (891.center)
			 to (889.center)
			 to (887.center)
			 to (888.center)
			 to (890.center)
			 to (892.center)
			 to (893.center)
			 to cycle;
		\draw [style=distribution] (894.center)
			 to (895.center)
			 to (897.center)
			 to (899.center)
			 to (900.center)
			 to (898.center)
			 to (896.center)
			 to cycle;
		\draw [bend left=15] (911.center) to (910.center);
		\draw [bend left=15] (912.center) to (902.center);
		\draw [bend left=15] (917.center) to (916.center);
		\draw [bend right=15] (915.center) to (901.center);
		\draw [bend right=15] (923.center) to (920.center);
		\draw [bend right=15] (922.center) to (904.center);
		\draw [bend left=15] (926.center) to (921.center);
		\draw [bend left=15] (927.center) to (903.center);
		\draw [bend right=15] (895.center) to (884.center);
		\draw [bend left=15] (896.center) to (885.center);
		\draw [bend right=15] (881.center) to (888.center);
		\draw [bend left=15] (882.center) to (889.center);
		\draw [style=distribution] (942.center)
			 to (940.center)
			 to (938.center)
			 to (939.center)
			 to (941.center)
			 to (943.center)
			 to (944.center)
			 to cycle;
		\draw [style=distribution] (945.center)
			 to (946.center)
			 to (948.center)
			 to (950.center)
			 to (951.center)
			 to (949.center)
			 to (947.center)
			 to cycle;
		\draw [bend right=15] (946.center) to (935.center);
		\draw [bend left=15] (947.center) to (936.center);
		\draw [bend right=15] (932.center) to (939.center);
		\draw [bend left=15] (933.center) to (940.center);
		\draw [bend right=15] (984.center) to (955.center);
		\draw [bend right=15] (953.center) to (983.center);
		\draw [bend left=15] (986.center) to (954.center);
		\draw [bend right=15] (985.center) to (952.center);
		\draw [bend right=15] (997.center) to (995.center);
		\draw [bend left=15] (998.center) to (996.center);
		\draw [bend right=15] (992.center) to (993.center);
		\draw [bend left=15] (991.center) to (994.center);
		\draw [style=distribution] (1013.center)
			 to (1011.center)
			 to (1009.center)
			 to (1010.center)
			 to (1012.center)
			 to (1014.center)
			 to (1015.center)
			 to cycle;
		\draw [style=distribution] (1019.center)
			 to (1021.center)
			 to (1022.center)
			 to (1020.center)
			 to (1018.center)
			 to (1016.center)
			 to (1017.center)
			 to cycle;
		\draw [bend left=15] (1033.center) to (1032.center);
		\draw [bend left=15] (1034.center) to (1024.center);
		\draw [bend left=15] (1039.center) to (1038.center);
		\draw [bend right=15] (1037.center) to (1023.center);
		\draw [bend right=15] (1045.center) to (1042.center);
		\draw [bend right=15] (1044.center) to (1026.center);
		\draw [bend left=15] (1048.center) to (1043.center);
		\draw [bend left=15] (1049.center) to (1025.center);
		\draw [bend right=15] (1017.center) to (1006.center);
		\draw [bend left=15] (1018.center) to (1007.center);
		\draw [bend right=15] (1003.center) to (1010.center);
		\draw [bend left=15] (1004.center) to (1011.center);
		\draw [style=distribution] (1065.center)
			 to (1066.center)
			 to (1064.center)
			 to (1062.center)
			 to (1060.center)
			 to (1061.center)
			 to (1063.center)
			 to cycle;
		\draw [style=distribution] (1071.center)
			 to (1069.center)
			 to (1067.center)
			 to (1068.center)
			 to (1070.center)
			 to (1072.center)
			 to (1073.center)
			 to cycle;
		\draw [bend right=15] (1068.center) to (1057.center);
		\draw [bend left=15] (1069.center) to (1058.center);
		\draw [bend right=15] (1054.center) to (1061.center);
		\draw [bend left=15] (1055.center) to (1062.center);
		\draw [bend right=15] (1083.center) to (1077.center);
		\draw [bend right=15] (1075.center) to (1082.center);
		\draw [bend left=15] (1085.center) to (1076.center);
		\draw [bend right=15] (1084.center) to (1074.center);
		\draw [style=distribution] (1101.center)
			 to (1099.center)
			 to (1097.center)
			 to (1095.center)
			 to (1096.center)
			 to (1098.center)
			 to (1100.center)
			 to cycle;
		\draw [style=distribution] (1105.center)
			 to (1107.center)
			 to (1108.center)
			 to (1106.center)
			 to (1104.center)
			 to (1102.center)
			 to (1103.center)
			 to cycle;
		\draw [bend left=15] (1118.center) to (1117.center);
		\draw [bend left=15] (1119.center) to (1110.center);
		\draw [bend left=15] (1124.center) to (1123.center);
		\draw [bend right=15] (1122.center) to (1109.center);
		\draw [bend right=15] (1130.center) to (1127.center);
		\draw [bend right=15] (1129.center) to (1112.center);
		\draw [bend left=15] (1133.center) to (1128.center);
		\draw [bend left=15] (1134.center) to (1111.center);
		\draw [bend right=15] (1103.center) to (1092.center);
		\draw [bend left=15] (1104.center) to (1093.center);
		\draw [bend right=15] (1089.center) to (1096.center);
		\draw [bend left=15] (1090.center) to (1097.center);
		\draw [style=distribution] (1150.center)
			 to (1151.center)
			 to (1149.center)
			 to (1147.center)
			 to (1145.center)
			 to (1146.center)
			 to (1148.center)
			 to cycle;
		\draw [style=distribution] (1156.center)
			 to (1154.center)
			 to (1152.center)
			 to (1153.center)
			 to (1155.center)
			 to (1157.center)
			 to (1158.center)
			 to cycle;
		\draw [bend right=15] (1153.center) to (1142.center);
		\draw [bend left=15] (1154.center) to (1143.center);
		\draw [bend right=15] (1139.center) to (1146.center);
		\draw [bend left=15] (1140.center) to (1147.center);
		\draw [bend right=15] (1168.center) to (1162.center);
		\draw [bend right=15] (1160.center) to (1167.center);
		\draw [bend left=15] (1170.center) to (1161.center);
		\draw [bend right=15] (1169.center) to (1159.center);
		\draw (1176.center) to (1175.center);
		\draw (1176.center) to (1177.center);
		\draw (1175.center) to (1178.center);
		\draw (1179.center) to (1178.center);
		\draw (1180.center) to (1179.center);
	\end{pgfonlayer}
\end{tikzpicture}

%% file: figures/exampleSV.tikz
\begin{tikzpicture}
	\begin{pgfonlayer}{nodelayer}
		\node [style=Green Node] (1) at (0.5, 0.75) {$\pi$};
		\node [style=Red Node] (2) at (-0.5, 0) {$\pi$};
		\node [style=Red Node] (3) at (0.5, -0.75) {$\pi$};
		\node [style=Green Node] (4) at (1.25, 0) {$\pi$};
		\node [style=txtnobold] (5) at (-1.5, 2) {$1-p$};
		\node [style=txtnobold] (6) at (1.2, 2) {$\frac{p}{3}$};
		\node [style=txtnobold] (7) at (0.19, 2) {$\frac{p}{3}$};
		\node [style=txtnobold] (8) at (-0.55, 2) {$\frac{p}{3}$};
		\node [style=none] (9) at (0, 3) {};
		\node [style=none] (10) at (-0.25, 3) {};
		\node [style=none] (11) at (0.25, 3) {};
		\node [style=none] (12) at (-0.5, 2.5) {};
		\node [style=none] (13) at (0.5, 2.5) {};
		\node [style=none] (14) at (-0.5, 2.5) {};
		\node [style=none] (15) at (0, 2.5) {};
		\node [style=none] (16) at (0.25, 2.5) {};
		\node [style=none] (17) at (-0.25, 2.5) {};
		\node [style=none] (18) at (0, -3) {};
		\node [style=none] (19) at (-0.25, -3) {};
		\node [style=none] (20) at (0.25, -3) {};
		\node [style=none] (21) at (-0.5, -2.5) {};
		\node [style=none] (22) at (0.5, -2.5) {};
		\node [style=none] (23) at (-0.5, -2.5) {};
		\node [style=none] (24) at (0, -2.5) {};
		\node [style=none] (25) at (0.25, -2.5) {};
		\node [style=none] (26) at (-0.25, -2.5) {};
		\node [style=none] (27) at (-0.4, -2.5) {};
		\node [style=none] (28) at (-0.4, 2.5) {};
		\node [style=none] (29) at (0.375, 2.5) {};
		\node [style=none] (30) at (0.375, -2.5) {};
		\node [style=txtnobold] (32) at (0.25, 0) {$i$};
		\node [style=Empty red] (34) at (0, 4.5) {};
		\node [style=Empty red] (35) at (0, -4) {};
		\node [style=Empty green] (36) at (3, -4) {};
		\node [style=h] (37) at (3, -3.25) {};
		\node [style=h] (38) at (3, -4.75) {};
		\node [style=Empty red] (39) at (3, -2.25) {};
		\node [style=Empty red] (40) at (3, -5.75) {};
		\node [style=none] (41) at (0, -5.75) {$\ground$};
		\node [style=none] (42) at (0, -5.75) {};
		\node [style=h] (43) at (0, 3.75) {};
	\end{pgfonlayer}
	\begin{pgfonlayer}{edgelayer}
		\draw (3) to (1);
		\draw [style=distribution] (15.center)
			 to (13.center)
			 to (11.center)
			 to (9.center)
			 to (10.center)
			 to (12.center)
			 to (14.center)
			 to cycle;
		\draw [style=distribution] (21.center)
			 to (23.center)
			 to (24.center)
			 to (22.center)
			 to (20.center)
			 to (18.center)
			 to (19.center)
			 to cycle;
		\draw [in=75, out=-90] (3) to (25.center);
		\draw [in=105, out=-90, looseness=0.75] (2) to (26.center);
		\draw [in=90, out=-105, looseness=0.75] (17.center) to (2);
		\draw [in=90, out=-75] (16.center) to (1);
		\draw [in=150, out=-150, looseness=0.75] (28.center) to (27.center);
		\draw [in=75, out=-30, looseness=0.75] (29.center) to (4);
		\draw [in=45, out=-75, looseness=0.75] (4) to (30.center);
		\draw (35) to (36);
		\draw (36) to (37);
		\draw (37) to (39);
		\draw (36) to (38);
		\draw (38) to (40);
		\draw (18.center) to (35);
		\draw (35) to (42.center);
		\draw (34) to (43);
		\draw (43) to (9.center);
	\end{pgfonlayer}
\end{tikzpicture}

%% file: content/discussion.tex
\section{Discussion and Future Work}
\label{sec:discussion}
In this work, we have shown how to construct freely enriched symmetric monoidal categories
over the algebras of a monoidal monad on $\SetC$ that satisfies $F \dashv \emt(I,-)$
for $F$ the free $T$-algebra functor and $I$ the unit of the monoidal structure, which is
the case in particular when $T$ is an affine monad.
We have then taken this construction and developed a graphical language that captures the
additional algebraic structure of the morphisms for the case of the Distribution monad.
We then show how we can use this to study classical probabilistic processes in quantum systems,
a highly relevant type of operation for near-term quantum applications.
In particular, we extend the ZX-calculus to make it a language for reasoning in
an enriched version of $\cpm$.

We believe that this work opens several directions for future research.
The most evident one is to prove \emph{completeness} of the enriched diagrams,
which in turn would facilitate automated implementations for tasks such as simulation
of noisy quantum systems, fine-tuned quantum circuit optimization techniques for
specific quantum devices, or comparison of the effectiveness of different Quantum
Error Mitigation techniques.
An interesting venue would be to use enrichment over $\mm$ to reason about
\emph{quantum circuit pre- and post-processing} techniques,
such as \emph{circuit cutting}~\cite{reducechop-2023,Peng-2020},
in which quantum circuits are ``split'' into linear combinations of smaller
ones that are executed separately.
We also believe that it could be possible to integrate monads that capture quantum behaviours
into our construction to represent in
enriched ZX-diagrams phenomena such as \emph{superposition of execution orders},
like what is done in the Many-Worlds calculus~\cite{chardonnet2022manyworlds}.

Strongly related to completeness is to have presentations of the diagrams in terms of generators
and equations.
We achieved this by hand in \Cref{sec:zx} by using that the algebras for the distribution monad
can be presented as convex algebras with a family of operations $+_{p}$.
The question is then what the analogue of convex monads is when using algebras presented by
Lawvere theories or sketches~\cite{HP07:CategoryTheoreticUnderstanding,Manes76:AlgebraicTheories}.

We are also interested in finding other monads that could capture interesting processes
outside of the quantum realm.
For example, the \emph{non-empty powerset monad} could be used to encode
non-deterministic operations and be used for reasoning about a third party operating on a shared
quantum system.

\subsubsection*{Acknowledgements}
This work was funded by the European Union under Grant Agreement 101080142, EQUALITY project.
AV was partly supported by project PRG 946 funded by the Estonian Research Council.
The authors would like to thank anonymous reviewers for pointing out
reference~\cite{banaschewskiNelson1976} regarding the tensor product of convex algebras. 

%%% Local Variables:
%%% mode: latex
%%% TeX-master: "../enriched-diagrams"
%%% End:

%% file: content/appendix.tex
\section{Proof of~\Cref{thm:pseudo-monoid-adjunction}}\label{appProofThm2}

First, we write down the diagrams for the natural isomorphisms of the pseudo-monoid that is the
image of $(\eC, \odot, \odotUnit, \alpha, \lambda, \rho, \sigma)$ under $\PMon(G)$, following the
diagrams of~\cite[Section 3]{DS97:MonoidalBicategoriesHopf}.

We have that the associator under $\PMon(G)$ is $\PMon(G)(\alpha)=G\alpha$, which corresponds
to the natural isomorphism between the outer layers of the
following diagram, for $A\in \obj{\eC}$:

% https://q.uiver.app/#q=WzAsOSxbMCwwLCJHQSBcXHRpbWVzIEdBIFxcdGltZXMgR0EiXSxbMiwwLCJHQSBcXHRpbWVzIEcoQVxcb3RpbWVzIEEpIl0sWzQsMCwiR0EgXFx0aW1lcyBHQSJdLFswLDEsIkcoQVxcb3RpbWVzIEEpIFxcdGltZXMgR0EgICJdLFs0LDEsIkcoQVxcb3RpbWVzIEEpIl0sWzIsMSwiRyhBXFxvdGltZXMgQVxcb3RpbWVzIEEpIl0sWzIsMiwiRyhBXFxvdGltZXMgQSkiXSxbNCwyLCJHQSJdLFswLDIsIkdBIFxcdGltZXMgR0EiXSxbMCwxLCJHX3tcXHRleHR7SWR9X0F9XFx0aW1lcyBHXjIiXSxbMSwyLCJHX3tcXHRleHR7SWR9X0F9XFx0aW1lcyBHKFxcb2RvdCkiXSxbMCwzLCJHXjJcXHRpbWVzIEdfe1xcdGV4dHtJZH1fQX0iLDJdLFsyLDQsIkdeMiJdLFszLDUsIkdeMiIsMl0sWzEsNSwiR14yIl0sWzUsNCwiRyhcXHRleHR7SWR9X0FcXG90aW1lcyBcXG9kb3QpIl0sWzUsNiwiRyhcXG9kb3QgXFxvdGltZXMgXFx0ZXh0e0lkfV9BKSIsMl0sWzYsNywiRyhcXG9kb3QpIiwyXSxbNCw3LCJHKFxcb2RvdCkiXSxbOCw2LCJHXjIiLDJdLFszLDgsIkcoXFxvZG90KVxcdGltZXMgR197XFx0ZXh0e0lkfV9BfSIsMl0sWzE2LDE4LCJHXFxhbHBoYSIsMix7InNob3J0ZW4iOnsic291cmNlIjo0MCwidGFyZ2V0Ijo0MH19XV0=
\[\begin{tikzcd}
	{GA \times GA \times GA} && {GA \times G(A\otimes A)} && {GA \times GA} \\
	{G(A\otimes A) \times GA  } && {G(A\otimes A\otimes A)} && {G(A\otimes A)} \\
	{GA \times GA} && {G(A\otimes A)} && GA
	\arrow["{G_{\text{Id}_A}\times G^2}", from=1-1, to=1-3]
	\arrow["{G_{\text{Id}_A}\times G(\odot)}", from=1-3, to=1-5]
	\arrow["{G^2\times G_{\text{Id}_A}}"', from=1-1, to=2-1]
	\arrow["{G^2}", from=1-5, to=2-5]
	\arrow["{G^2}"', from=2-1, to=2-3]
	\arrow["{G^2}", from=1-3, to=2-3]
	\arrow["{G(\text{Id}_A\otimes \odot)}", from=2-3, to=2-5]
	\arrow[""{name=0, anchor=center, inner sep=0}, "{G(\odot \otimes \text{Id}_A)}"',
        from=2-3, to=3-3]
	\arrow["{G(\odot)}"', from=3-3, to=3-5]
	\arrow[""{name=1, anchor=center, inner sep=0}, "{G(\odot)}", from=2-5, to=3-5]
	\arrow["{G^2}"', from=3-1, to=3-3]
	\arrow["{G(\odot)\times G_{\text{Id}_A}}"', from=2-1, to=3-1]
	\arrow["G\alpha"', shorten <=37pt, shorten >=37pt, Rightarrow, from=0, to=1]
\end{tikzcd}\]

The inner squares commute from associativity and naturality of $G^2$.

From there we also see that we can define the multiplication and unit functors under
$\PMon(G)$ to be $\PMon(G)(\odot):= G(\odot)\circ G^2$ and $\PMon(G)(U):= G(U)\circ G^0$.

Similarly, the left unitor under $\PMon(G)$ is $\PMon(G)(\lambda)=G\lambda$ by following the diagram

% https://q.uiver.app/#q=WzAsOCxbMCwwLCJHQSJdLFswLDEsIiogXFx0aW1lcyBHQSJdLFswLDIsIkdJIFxcdGltZXMgR0EiXSxbMCwzLCJHQSBcXHRpbWVzIEdBIl0sWzIsMiwiRyhJXFxvdGltZXMgQSkiXSxbMiwzLCJHKEFcXG90aW1lcyBBKSJdLFszLDMsIkdBIl0sWzIsMCwiR0EiXSxbMCwxLCJcXGNvbmciLDJdLFsxLDIsIkdeMCBcXHRpbWVzIEdfe1xcdGV4dHtJZH1fQX0iLDJdLFsyLDMsIkcoVSkgXFx0aW1lcyBHX3tcXHRleHR7SWR9X0F9IiwyXSxbMiw0LCJHXjIiLDJdLFs0LDUsIkcoVVxcb3RpbWVzIFxcdGV4dHtJZH1fQSkiLDJdLFszLDUsIkdeMiIsMl0sWzUsNiwiRyhcXG9kb3QpIiwyXSxbMCw3LCIiLDIseyJvZmZzZXQiOi0xLCJzdHlsZSI6eyJoZWFkIjp7Im5hbWUiOiJub25lIn19fV0sWzAsNywiIiwyLHsic3R5bGUiOnsiaGVhZCI6eyJuYW1lIjoibm9uZSJ9fX1dLFs3LDQsIkcoXFxjb25nKSIsMl0sWzcsNiwiR197XFx0ZXh0e0lkfV9BfSIsMCx7ImN1cnZlIjotMn1dLFs0LDE4LCJHXFxsYW1iZGEiLDAseyJzaG9ydGVuIjp7InRhcmdldCI6MjB9fV1d
\[\begin{tikzcd}
    GA && GA \\
    {* \times GA} \\
    {GI \times GA} && {G(I\otimes A)} \\
    {GA \times GA} && {G(A\otimes A)} & GA
    \arrow["\cong"', from=1-1, to=2-1]
    \arrow["{G^0 \times G_{\text{Id}_A}}"', from=2-1, to=3-1]
    \arrow["{G(U) \times G_{\text{Id}_A}}"', from=3-1, to=4-1]
    \arrow["{G^2}"', from=3-1, to=3-3]
    \arrow["{G(U\otimes \text{Id}_A)}"', from=3-3, to=4-3]
    \arrow["{G^2}"', from=4-1, to=4-3]
    \arrow["{G(\odot)}"', from=4-3, to=4-4]
    \arrow[shift left, no head, from=1-1, to=1-3]
    \arrow[no head, from=1-1, to=1-3]
    \arrow["{G(\cong)}"', from=1-3, to=3-3]
    \arrow[""{name=0, anchor=center, inner sep=0}, "{G_{\text{Id}_A}}", curve={height=-12pt},
    from=1-3, to=4-4]
    \arrow["G\lambda", shorten >=4pt, Rightarrow, from=3-3, to=0]
  \end{tikzcd}\]

Where the top and bottom squares commute from coherence and naturality of $G^2$, respectively.
We can find $\PMon(G)(\rho)=G\rho$ in the same way.
The braiding $\sigma$ gets mapped to $\PMon(G)(\sigma)=G\sigma$:

% https://q.uiver.app/#q=WzAsNSxbMCwwLCJHQSBcXHRpbWVzIEdBIl0sWzIsMCwiR0EgXFx0aW1lcyBHQSJdLFswLDEsIkcoQSBcXG90aW1lcyBBKSJdLFsyLDEsIkcoQSBcXG90aW1lcyBBKSJdLFsxLDIsIkdBIl0sWzAsMSwiU19cXHRpbWVzIl0sWzAsMiwiR14yIiwyXSxbMSwzLCJHXjIiXSxbMiwzLCJHKFNfXFxvdGltZXMpIl0sWzIsNCwiRyhcXG90aW1lcykiLDJdLFszLDQsIkcoXFxvdGltZXMpIl0sWzIsMywiR1xcc2lnbWEiLDIseyJvZmZzZXQiOjUsInNob3J0ZW4iOnsic291cmNlIjo0MCwidGFyZ2V0Ijo0MH0sImxldmVsIjoyfV1d
\[\begin{tikzcd}
	{GA \times GA} && {GA \times GA} \\
	{G(A \otimes A)} && {G(A \otimes A)} \\
	& GA
	\arrow["{S_\times}", from=1-1, to=1-3]
	\arrow["{G^2}"', from=1-1, to=2-1]
	\arrow["{G^2}", from=1-3, to=2-3]
	\arrow["{G(S_\otimes)}", from=2-1, to=2-3]
	\arrow["{G(\otimes)}"', from=2-1, to=3-2]
	\arrow["{G(\otimes)}", from=2-3, to=3-2]
	\arrow["G\sigma"', shift right=5, shorten <=18pt, shorten >=18pt, Rightarrow,
        from=2-1, to=2-3]
\end{tikzcd}\]

where $S_\otimes, S_\times$ are the braidings in $(\VCat, \otimes, \otimesUnit)$ and
$(\WCat, \times, \timesUnit)$, respectively.
It follows that $\PMon(G)(\sigma)$ is symmetric if $S_\otimes, S_\times$ are symmetric.

$\PMon(G)$ maps pseudo-monoid homomorphisms
$(h, h^{0}, h^{2}): (\eC, \odot_{1}, U_{1}) \to (\eD, \odot_{2}, U_{2})$ to
$$(G(h), G(h^{0}), G(h^{2})): (G\eC, G(\odot_{1}), G(U_{1})) \to (G\eD, G(\odot_{2}), G(U_{2}))$$
as seen from the following unit and multiplication diagrams, for $A\in \obj{\eC}, B\in \obj{\eD}$:

% https://q.uiver.app/#q=WzAsMTIsWzAsMCwiKiJdLFswLDEsIkdJIl0sWzAsMiwiR0IiXSxbMSwwLCIqIl0sWzEsMSwiR0kiXSxbMSwyLCJHQSJdLFs1LDAsIkdBIFxcdGltZXMgR0EiXSxbNSwxLCJHKEFcXG90aW1lcyBBKSJdLFs1LDIsIkdBIl0sWzcsMCwiR0IgXFx0aW1lcyBHQiJdLFs3LDEsIkcoQlxcb3RpbWVzIEIpIl0sWzcsMiwiR0IiXSxbMCwxLCJHXjAiLDJdLFsxLDIsIkcoVV8yKSIsMl0sWzMsNCwiR14wIl0sWzAsMywiIiwwLHsic3R5bGUiOnsiaGVhZCI6eyJuYW1lIjoibm9uZSJ9fX1dLFswLDMsIiIsMSx7Im9mZnNldCI6MSwic3R5bGUiOnsiaGVhZCI6eyJuYW1lIjoibm9uZSJ9fX1dLFsxLDQsIiIsMCx7Im9mZnNldCI6MSwic3R5bGUiOnsiaGVhZCI6eyJuYW1lIjoibm9uZSJ9fX1dLFsxLDQsIiIsMCx7InN0eWxlIjp7ImhlYWQiOnsibmFtZSI6Im5vbmUifX19XSxbNCw1LCJHKFVfMSkiXSxbNSwyLCJHKGgpIl0sWzYsNywiR14yIiwyXSxbNyw4LCJHKFxcb2RvdF8xKSIsMl0sWzYsOSwiRyhoKSBcXHRpbWVzIEcoaCkiXSxbOSwxMCwiR14yIl0sWzcsMTAsIkcoaCBcXG90aW1lcyBoKSJdLFs4LDExLCJHKGgpIiwyXSxbMTAsMTEsIkcoXFxvZG90XzIpIl0sWzEzLDE5LCJHKGheMCkiLDIseyJzaG9ydGVuIjp7InNvdXJjZSI6MzAsInRhcmdldCI6MzB9fV0sWzIyLDI3LCJHKGheMikiLDIseyJzaG9ydGVuIjp7InNvdXJjZSI6MzAsInRhcmdldCI6MzB9fV1d
\[\begin{tikzcd}
    {*} & {*} &&&& {GA \times GA} && {GB \times GB} \\
    GI & GI &&&& {G(A\otimes A)} && {G(B\otimes B)} \\
    GB & GA &&&& GA && GB
    \arrow["{G^0}"', from=1-1, to=2-1]
    \arrow[""{name=0, anchor=center, inner sep=0}, "{G(U_2)}"', from=2-1, to=3-1]
    \arrow["{G^0}", from=1-2, to=2-2]
    \arrow[no head, from=1-1, to=1-2]
    \arrow[shift right, no head, from=1-1, to=1-2]
    \arrow[shift right, no head, from=2-1, to=2-2]
    \arrow[no head, from=2-1, to=2-2]
    \arrow[""{name=1, anchor=center, inner sep=0}, "{G(U_1)}", from=2-2, to=3-2]
    \arrow["{G(h)}", from=3-2, to=3-1]
    \arrow["{G^2}"', from=1-6, to=2-6]
    \arrow[""{name=2, anchor=center, inner sep=0}, "{G(\odot_1)}"', from=2-6, to=3-6]
    \arrow["{G(h) \times G(h)}", from=1-6, to=1-8]
    \arrow["{G^2}", from=1-8, to=2-8]
    \arrow["{G(h \otimes h)}", from=2-6, to=2-8]
    \arrow["{G(h)}"', from=3-6, to=3-8]
    \arrow[""{name=3, anchor=center, inner sep=0}, "{G(\odot_2)}", from=2-8, to=3-8]
    \arrow["{G(h^0)}"', shorten <=10pt, shorten >=10pt, Rightarrow, from=0, to=1]
    \arrow["{G(h^2)}"', shorten <=25pt, shorten >=25pt, Rightarrow, from=2, to=3]
  \end{tikzcd}\]

Lastly, we show that if $G$ has a left adjoint $F$, then we also have $\PMon(F) \dashv \PMon(G)$.
Let
$(\eC, \odot_{1}, U_{1})\in \PMon(\WC, \times, \timesUnit),
(\eD, \odot_{2}, U_{2}) \in \PMon(\VC, \otimes, \otimesUnit)$,
and let $A\in\obj{\eC},B\in\obj{\eD}$.
To show this, we show that the hom-adjunction $\alpha_{A,B}: \VC(FA,B)\cong \WC(A,GB):\beta_{A,B}$
induces
$$\tilde{\alpha}_{A,B}: \PMon(\VC)(\PMon(F)A,B) \cong \PMon(\WC)(A,\PMon(G)B):\tilde{\beta}_{A,B}.$$

A pseudo-monoid homomorphism $(h, h^{0}, h^{2})\in \PMon(\VC)(\PMon(F)A,B)$ has types
\begin{equation}
  \label{eq:PMhomThm2Proof}
  \begin{split}
    & h: FA\to B,\\
    & h^{0}: U_{2}\to h \circ F(U_{1}) \circ F^{0},\\
    & h^{2}: \odot_{2} \circ (h \otimes h) \to h  \circ F(\odot_{1}) \circ F^{2}.
  \end{split}
\end{equation}

In the same way, the image of $(h, h^{0}, h^{2})$ under $\tilde{\alpha}$ has types
\begin{equation}
  \label{eq:PMhomAdjThm2Proof}
  \begin{split}
    & \tilde{\alpha}(h): A\to GB,\\
    & \tilde{\alpha}(h^{0}): G(U_2) \circ G^{0} \to \tilde{\alpha}(h) \circ U_1,\\
    & \tilde{\alpha}(h^{2}): G(\odot_2) \circ G^{2} \circ (\tilde{\alpha}(h) \times
      \tilde{\alpha}(h)) \to \tilde{\alpha}(h) \circ \odot_1.
  \end{split}
\end{equation}

We can set $\tilde{\alpha}(h)=\alpha(h)$.
Whiskering (depicted here with a bullet $\bullet$) $h^0$
from~\Cref{eq:PMhomThm2Proof} with $(F^{0})^{-1}$ and applying $\alpha$ to it yields
\begin{equation}
  \label{eq:AdjunctH0Thm2Proof}
  \alpha( h^{0} \bullet  (F^{0})^{-1}): \alpha(U_2 \circ (F^{0})^{-1}) \to \alpha(h \circ F(U_{1})
  \circ F^0 \circ (F^0)^{-1}).
\end{equation}

Similarly, whiskering $h^{2}$ with $(F^2)^{-1}$ and applying $\alpha$ to it gives us
\begin{equation}
  \label{eq:AdjunctH2Thm2Proof}
  \alpha( h^{2} \bullet  (F^2)^{-1}):
  \alpha(\odot_{2} \circ (h \otimes h) \circ (F^2)^{-1}) \to \alpha(h  \circ F(\odot_{1})
  \circ F^2 \circ (F^2)^{-1}  ).
\end{equation}

If we have
\begin{equation}\label{eq:ReqThm2Proof}
  (F^{0})^{-1} = \beta(G^{0})
  \quad \text{and} \quad
  (F^{2})^{-1} = \beta(G^{2}) \comp F(\eta \times \eta),
\end{equation}
then we can simplify both sides of~\Cref{eq:AdjunctH0Thm2Proof} to 
$$\alpha(U_2 \circ (F^{0})^{-1}) = G(U_2) \circ \alpha((F^{0})^{-1}) = G(U_2) \circ G^{0} $$
and $\alpha(h \circ F(U_{1})) = \alpha(h) \circ U_{1}$ using naturality of $\alpha$
and~\Cref{eq:ReqThm2Proof}.

Similarly for~\Cref{eq:AdjunctH2Thm2Proof} we get
\begin{equation}
  \label{eq:AlphaH2Thm2Proof}
  \begin{split}
    \alpha(\odot_{2} \circ (h \otimes h) \circ (F^2)^{-1}) & = G(\odot_{2})
    \circ G(h \otimes h) \circ G^2 \circ (\eta \times \eta) \\
    & = G(\odot_{2}) \circ G^2 \circ (Gh \times Gh) \circ (\eta \times \eta) \\
    & = G(\odot_{2}) \circ G^2 \circ (Gh\circ \eta \times Gh\circ \eta) \\
    & = G(\odot_{2}) \circ G^2 \circ (\alpha(h) \times \alpha(h))
  \end{split}
\end{equation}
and $\alpha(h  \circ F(\odot_{1})) = \alpha(h) \circ \odot_{1}$ by naturality of $\alpha$
and \Cref{eq:ReqThm2Proof}.

This shows us that defining
$\tilde{\alpha}(h) := \alpha(h), \tilde{\alpha}(h^{0}) := \alpha( h^{0} \bullet  (F^{0})^{-1})$
and $ \tilde{\alpha}(h^{2}) := \alpha( h^{2} \bullet  (F^2)^{-1})$
gives us the necessary natural transformation.

We can do the converse for $\tilde{\beta}$.
This time the pseudo-monoid homomorphism $(h, h^{0}, h^{2})\in \PMon(\WC)(A,\PMon(G)B)$ has types
\begin{equation}
  \begin{split}
    & h: A\to GB,\\
    & h^{0}: G(U_2) \circ G^{0} \to h \circ U_1\\
    & h^{2}: G(\odot_2) \circ G^{2} \circ (h \times h) \to h \circ \odot_1.
  \end{split}
\end{equation}
Setting $\tilde{\beta}(h) := \beta(h),$ we can define
$$
  \tilde{\beta}(h^0) := \beta(h^0 \bullet F^0):
\beta(G(U_2)\circ G^0)\circ F^0 \to \beta(h\circ U_1) \circ F^0, 
$$
where
$$
\beta(G(U_2)\circ G^0)\circ F^0 = U_2,\quad \quad
\beta(h\circ U_1) \circ F^0 = \tilde{\beta}(h)\circ F(U_1)\circ F^0
$$
following similar steps as what we did for the derivation of $\tilde{\alpha}$.
Similarly we define
$$\tilde{\beta}(h^2) := \beta(h^2 \bullet F^2):
\beta(G(\odot_2)\circ G^2 \circ (h\times h))\circ F^2  \to \beta(h\circ \odot_1)  \circ F^2$$
where
\begin{gather*}
  \beta(G(\odot_2)\circ G^2 \circ (h\times h))\circ F^2 =
  \odot_2\circ (\tilde{\beta}(h)\otimes \tilde{\beta}(h)),\\
  \beta(h\circ \odot_1)  \circ F^2= \tilde{\beta}(h)\circ F(\odot_1)\circ F^2
\end{gather*}
also following similar steps as what previously done.
These definitions of $\tilde{\alpha},\tilde{\beta}$ form an isomorphism,
which concludes the proof that the adjunction gets lifted to include monoidal structure.

\section{Monoidal Structure in $\emd$}\label{appMonoidalAlg}

The tensor product of algebras of $\md$ represents bi-convex maps
following~\cite{banaschewskiNelson1976}, which we give in detail in this section.
For algebras $a \from \Distr(A) \to A$, $b \from \Distr(B) \to B$ and $c \from \Distr(C) \to C$
we denote by $\BiConv{a}{b}{c}$ the set
\begin{equation*}
  \setDef*{f \from A \times B \to C}{
    \all{\sigma}\all{\gamma} f(a(\sigma), b(\gamma)) = c\parens*{\sum_{x,y} \sigma(x)\gamma(y)[f(x,y)]}}
  \tag{$\ast$}
\end{equation*}
of bi-convex maps.
One can show that the coequalizer~\eqref{eq:coeq} for $T = \md$ is given universally by the property
\begin{equation*}
  \emd(a \otimes b, c) \cong \BiConv{a}{b}{c}
\end{equation*}
as follows.
First of all, by definition of the equalizer and by $\Distr$ being a monad, we have the following
natural isomorphisms.
\begin{align*}
  \emd(a \otimes b, c)
  & \cong
    \setDef{g \in \emd(F(A \times B),c)}{
    g \comp \Distr(a \times b) = g \comp \mu \comp \Distr\nabla} \\
  & \cong
    \setDef{f \from A \times B \to C}{
    c \comp \Distr f \comp \Distr(a \times b) = c \comp \Distr f \comp \mu \comp \Distr\nabla}
    \tag{$\dagger$}
\end{align*}
So it remains to show that the two conditions ($\ast$) and ($\dagger$) on maps
$f \from A \times B \to C$ are equivalent.
A calculation shows for $\rho \in \Distr(\Distr(A) \times \Distr(B))$ that %, $x \in A$ and $y \in B$ that
\begin{equation*}
  % (\mu \comp \Distr\nabla)(\rho)(x,y) = \sum_{\sigma, \gamma}\rho(\sigma, \gamma) \sigma(x) \gamma(y)
  (\mu \comp \Distr\nabla)(\rho) = \sum_{x,y}\sum_{\sigma, \gamma}\rho(\sigma, \gamma) \sigma(x) \gamma(y) [(x,y)]
\end{equation*}
and
\begin{equation*}
  % \Distr(a \times b)(\rho)(x,y) = \sum_{\substack{\sigma, \gamma\\ a(\sigma) = x \\ b(\gamma) = y}} \rho(\sigma, \gamma)
  \Distr(a \times b)(\rho) = \sum_{\sigma, \gamma} \rho(\sigma, \gamma) [(a(\sigma), b(\gamma))] \, .
\end{equation*}
Since, for $\rho = \eta(\sigma, \gamma)$ we have
\begin{equation*}
  (\mu \comp \Distr\nabla)(\rho) = \sum_{x,y}\sigma(x) \gamma(y) [(x,y)]
\end{equation*}
and
\begin{equation*}
  \Distr(a \times b)(\rho) = \eta(a(\sigma), b(\gamma)) \, ,
\end{equation*}
it is immediately clear that condition ($\dagger$) implies ($\ast$).
The other way around, since $c$ is an algebra and $\mu$ a natural transformation, we have that
\begin{equation*}
  c \comp \Distr f \comp \mu \comp \Distr\nabla = c \comp \Distr c \comp \Distr^{2} f \comp \Distr\nabla
\end{equation*}
By functoriality of $\Distr$, the condition ($\dagger$) thus holds if
$f \comp (a \times b) = c \comp \Distr f \comp \nabla$, which is exactly condition ($\ast$).
Hence, the two conditions imply each other and algebra morphisms out of $a \otimes b$ are the same
as bi-convex map out of $A \times B$.

\section{Soundness of New Ruleset}\label{sec:appendixSoundness}
We show that the interpretation $\llbracket \cdot \rrbracket_\md$ of LHS and RHS are the same:
\begin{itemize}
\item \textbf{(es):}
  \begin{equation*}
    \begin{split}
      \text{LHS: } \llangle \left( \sum_{j=1}^\ell q_j [D_j']\right) \circ \left(\sum_{i=1}^k p_i [D_i]\right)  \rrangle & =
                                                                                                            \llangle \sum_{j=1}^\ell q_j [D_j']\rrangle \circ \llangle\sum_{i=1}^k p_i [D_i]  \rrangle\\
                                                                                                          & = \left(\sum_{j=1}^\ell q_j \llbracket D_j'\rrbracket\right) \circ \left(\sum_{i=1}^k p_i \llbracket D_i\rrbracket\right)\\
                                                                                                          & =  \sum_{j=1}^\ell\sum_{i=1}^k q_jp_i \llbracket D_j'\rrbracket \circ \llbracket D_i\rrbracket\\
                                                                                                          & =  \sum_{j=1}^\ell\sum_{i=1}^k q_jp_i \llbracket D_j' \circ D_i\rrbracket
    \end{split}
  \end{equation*}
  \begin{equation*}
    \begin{split}
      \text{RHS: } \llangle \sum_{j=1}^\ell\sum_{i=1}^k q_jp_i [ D_j' \circ D_i]\rrangle & =
                                                                                                   \sum_{j=1}^\ell\sum_{i=1}^k q_jp_i \llbracket D_j' \circ D_i\rrbracket
    \end{split}
  \end{equation*}
  
\item \textbf{(ep):} Same as \textbf{(es)}, but substituting composition with tensor product.

\item \textbf{(ec):}
  \begin{equation*}
      \text{LHS: } \llangle p_1 [D_1] + p_2[D_2]  \rrangle  = p_1 \llbracket D_1\rrbracket + p_2 \llbracket D_2\rrbracket
  \end{equation*}
  \begin{equation*}
      \text{RHS: } \llangle p_2 [D_2] + p_1[D_1]  \rrangle  = p_2 \llbracket D_2\rrbracket + p_1 \llbracket D_1\rrbracket
  \end{equation*}

\item \textbf{(e$\delta$):}
  \begin{equation*}
    \text{LHS: } \llangle 1 [D] \rrangle  = \llbracket D\rrbracket
  \end{equation*}
  \begin{equation*}
    \text{RHS: } \llbracket D\rrbracket
  \end{equation*}
  
\item \textbf{(e$+$):}
  \begin{equation*}
    \text{LHS: } \llangle p_1 [D] + p_2[D]  \rrangle  = p_1 \llbracket D\rrbracket + p_2 \llbracket D\rrbracket = (p_1 + p_2) \llbracket D\rrbracket
  \end{equation*}
  \begin{equation*}
    \text{RHS: } \llangle (p_1+p_2)[D]  \rrangle  = (p_1 + p_2) \llbracket D\rrbracket
  \end{equation*}

\item \textbf{(e$0$):}
  \begin{equation*}
    \text{LHS: } \llangle p [D_1] + 0[D_2]  \rrangle  = p \llbracket D_1\rrbracket + 0 \llbracket D_2\rrbracket = p \llbracket D_1\rrbracket
  \end{equation*}
  \begin{equation*}
    \text{RHS: } \llangle p [D_1]  \rrangle  = p \llbracket D_1\rrbracket
  \end{equation*}  
\end{itemize}

\section{Diagrammatic Study of Symmetry Verification}\label{app:SV}

To have the correct scalar factors in the diagram, we have to take a look
at Example~\eqref{eq:zxExample}.
This tells us the scalar factors we need to apply to the diagram in order to represent
exactly the diagrams for $\ket{0},\bra{0},$ and CNOT.
We then have the following scalar-corrected diagram:

\begin{equation*}\label{eq:scaledSV}
  \input{figures/scaledSV.tikz}
\end{equation*}

Then, applying diagrammatic rules to the diagram results in the following:

\begin{equation*}\label{eq:derivedSV}
  \input{figures/derivedSV.tikz}
\end{equation*}

We then use one rule regarding the discard operator from~\cite{caretteGround}, stating
that we can remove the dangling diagram at the bottom if we multiply the enitre diagram by $\sqrt{2}$:

\begin{equation*}\label{eq:derivedSV2}
  \input{figures/derivedSV2.tikz}
\end{equation*}

Calculating the interpretation of the last diagram amounts to calculating the interpretation of
each scalar diagram in each branch, multiply them by the probabilities and scalar factor, and
add them together.
More interestingly, we see that there are two possibilities, either a scalar $Z$ spider with phase
$0$ or with phase $\pi$, which respectively have interpretation $2$ and $0$ following~\Cref{fig:zxInterpretation}.
The interpretation of the last diagram then gives us that $\tr(P_{+1}\rho_{noisy}) = 1 - \frac{2p}{3}$.

What this is expressing is that our symmetry commutes with the first two branches
(the identity and the $X$ error) and anti-commutes with the other two, which in practice
tells us that measuring this symmetry only mitigates some part of the depolarizing channel.
Then, by representing different symmetries we can use this form of diagrammatic reasoning
to find and compare symmetries for different error channels.

%%% Local Variables:
%%% mode: latex
%%% TeX-master: "../enriched-diagrams"
%%% End:

%% file: figures/scaledSV.tikz
\begin{tikzpicture}
	\begin{pgfonlayer}{nodelayer}
		\node [style=Green Node] (0) at (2.5, 0.75) {$\pi$};
		\node [style=Red Node] (1) at (1.5, 0) {$\pi$};
		\node [style=Red Node] (2) at (2.5, -0.75) {$\pi$};
		\node [style=Green Node] (3) at (3.25, 0) {$\pi$};
		\node [style=txtnobold] (4) at (0.5, 2) {$1-p$};
		\node [style=txtnobold] (5) at (3.2, 2) {$\frac{p}{3}$};
		\node [style=txtnobold] (6) at (2.19, 2) {$\frac{p}{3}$};
		\node [style=txtnobold] (7) at (1.45, 2) {$\frac{p}{3}$};
		\node [style=none] (8) at (2, 3) {};
		\node [style=none] (9) at (1.75, 3) {};
		\node [style=none] (10) at (2.25, 3) {};
		\node [style=none] (11) at (1.5, 2.5) {};
		\node [style=none] (12) at (2.5, 2.5) {};
		\node [style=none] (13) at (1.5, 2.5) {};
		\node [style=none] (14) at (2, 2.5) {};
		\node [style=none] (15) at (2.25, 2.5) {};
		\node [style=none] (16) at (1.75, 2.5) {};
		\node [style=none] (17) at (2, -3) {};
		\node [style=none] (18) at (1.75, -3) {};
		\node [style=none] (19) at (2.25, -3) {};
		\node [style=none] (20) at (1.5, -2.5) {};
		\node [style=none] (21) at (2.5, -2.5) {};
		\node [style=none] (22) at (1.5, -2.5) {};
		\node [style=none] (23) at (2, -2.5) {};
		\node [style=none] (24) at (2.25, -2.5) {};
		\node [style=none] (25) at (1.75, -2.5) {};
		\node [style=none] (26) at (1.6, -2.5) {};
		\node [style=none] (27) at (1.6, 2.5) {};
		\node [style=none] (28) at (2.375, 2.5) {};
		\node [style=none] (29) at (2.375, -2.5) {};
		\node [style=txtnobold] (30) at (2.25, 0) {$i$};
		\node [style=Empty red] (32) at (2, 4.5) {};
		\node [style=Empty red] (33) at (2, -4) {};
		\node [style=Empty green] (34) at (5, -4) {};
		\node [style=h] (35) at (5, -3.25) {};
		\node [style=h] (36) at (5, -4.75) {};
		\node [style=Empty red] (37) at (5, -2.25) {};
		\node [style=Empty red] (38) at (5, -5.75) {};
		\node [style=none] (39) at (2, -5.75) {$\ground$};
		\node [style=none] (40) at (2, -5.75) {};
		\node [style=txtnobold] (41) at (4.5, 3.5) {$\left(\frac{1}{\sqrt{2}}\right)^2$};
		\node [style=h] (42) at (2, 3.75) {};
	\end{pgfonlayer}
	\begin{pgfonlayer}{edgelayer}
		\draw (2) to (0);
		\draw [style=distribution] (14.center)
			 to (12.center)
			 to (10.center)
			 to (8.center)
			 to (9.center)
			 to (11.center)
			 to (13.center)
			 to cycle;
		\draw [style=distribution] (20.center)
			 to (22.center)
			 to (23.center)
			 to (21.center)
			 to (19.center)
			 to (17.center)
			 to (18.center)
			 to cycle;
		\draw [in=75, out=-90] (2) to (24.center);
		\draw [in=105, out=-90, looseness=0.75] (1) to (25.center);
		\draw [in=90, out=-105, looseness=0.75] (16.center) to (1);
		\draw [in=90, out=-75] (15.center) to (0);
		\draw [in=150, out=-150, looseness=0.75] (27.center) to (26.center);
		\draw [in=75, out=-30, looseness=0.75] (28.center) to (3);
		\draw [in=45, out=-75, looseness=0.75] (3) to (29.center);
		\draw (33) to (34);
		\draw (34) to (35);
		\draw (35) to (37);
		\draw (34) to (36);
		\draw (36) to (38);
		\draw (17.center) to (33);
		\draw (33) to (40.center);
		\draw (32) to (42);
		\draw (42) to (8.center);
	\end{pgfonlayer}
\end{tikzpicture}

%% file: figures/derivedSV.tikz
\begin{tikzpicture}
	\begin{pgfonlayer}{nodelayer}
		\node [style=Green Node] (0) at (-2, 0.75) {$\pi$};
		\node [style=Red Node] (1) at (-3, 0) {$\pi$};
		\node [style=Red Node] (2) at (-2, -0.75) {$\pi$};
		\node [style=Green Node] (3) at (-1.25, 0) {$\pi$};
		\node [style=txtnobold] (4) at (-4, 2) {$1-p$};
		\node [style=txtnobold] (5) at (-1.3, 2) {$\frac{p}{3}$};
		\node [style=txtnobold] (6) at (-2.31, 2) {$\frac{p}{3}$};
		\node [style=txtnobold] (7) at (-3.05, 2) {$\frac{p}{3}$};
		\node [style=none] (8) at (-2.5, 3) {};
		\node [style=none] (9) at (-2.75, 3) {};
		\node [style=none] (10) at (-2.25, 3) {};
		\node [style=none] (11) at (-3, 2.5) {};
		\node [style=none] (12) at (-2, 2.5) {};
		\node [style=none] (13) at (-3, 2.5) {};
		\node [style=none] (14) at (-2.5, 2.5) {};
		\node [style=none] (15) at (-2.25, 2.5) {};
		\node [style=none] (16) at (-2.75, 2.5) {};
		\node [style=none] (17) at (-2.5, -3) {};
		\node [style=none] (18) at (-2.75, -3) {};
		\node [style=none] (19) at (-2.25, -3) {};
		\node [style=none] (20) at (-3, -2.5) {};
		\node [style=none] (21) at (-2, -2.5) {};
		\node [style=none] (22) at (-3, -2.5) {};
		\node [style=none] (23) at (-2.5, -2.5) {};
		\node [style=none] (24) at (-2.25, -2.5) {};
		\node [style=none] (25) at (-2.75, -2.5) {};
		\node [style=none] (26) at (-2.9, -2.5) {};
		\node [style=none] (27) at (-2.9, 2.5) {};
		\node [style=none] (28) at (-2.125, 2.5) {};
		\node [style=none] (29) at (-2.125, -2.5) {};
		\node [style=txtnobold] (30) at (-2.25, 0) {$i$};
		\node [style=Empty red] (32) at (-2.5, 4.5) {};
		\node [style=Empty red] (33) at (-2.5, -4) {};
		\node [style=Empty green] (34) at (0.5, -4) {};
		\node [style=h] (35) at (0.5, -3.25) {};
		\node [style=h] (36) at (0.5, -4.75) {};
		\node [style=Empty red] (37) at (0.5, -2.25) {};
		\node [style=Empty red] (38) at (0.5, -5.75) {};
		\node [style=none] (39) at (-2.5, -5.75) {$\ground$};
		\node [style=none] (40) at (-2.5, -5.75) {};
		\node [style=txtnobold] (41) at (0, 3.5) {$\left(\frac{1}{\sqrt{2}}\right)^2$};
		\node [style=Green Node] (42) at (6, 0.75) {$\pi$};
		\node [style=Red Node] (43) at (5, 0) {$\pi$};
		\node [style=Red Node] (44) at (6, -0.75) {$\pi$};
		\node [style=Green Node] (45) at (6.75, 0) {$\pi$};
		\node [style=txtnobold] (46) at (4, 2) {$1-p$};
		\node [style=txtnobold] (47) at (6.7, 2) {$\frac{p}{3}$};
		\node [style=txtnobold] (48) at (5.69, 2) {$\frac{p}{3}$};
		\node [style=txtnobold] (49) at (4.95, 2) {$\frac{p}{3}$};
		\node [style=none] (50) at (5.5, 3) {};
		\node [style=none] (51) at (5.25, 3) {};
		\node [style=none] (52) at (5.75, 3) {};
		\node [style=none] (53) at (5, 2.5) {};
		\node [style=none] (54) at (6, 2.5) {};
		\node [style=none] (55) at (5, 2.5) {};
		\node [style=none] (56) at (5.5, 2.5) {};
		\node [style=none] (57) at (5.75, 2.5) {};
		\node [style=none] (58) at (5.25, 2.5) {};
		\node [style=none] (59) at (5.5, -3) {};
		\node [style=none] (60) at (5.25, -3) {};
		\node [style=none] (61) at (5.75, -3) {};
		\node [style=none] (62) at (5, -2.5) {};
		\node [style=none] (63) at (6, -2.5) {};
		\node [style=none] (64) at (5, -2.5) {};
		\node [style=none] (65) at (5.5, -2.5) {};
		\node [style=none] (66) at (5.75, -2.5) {};
		\node [style=none] (67) at (5.25, -2.5) {};
		\node [style=none] (68) at (5.1, -2.5) {};
		\node [style=none] (69) at (5.1, 2.5) {};
		\node [style=none] (70) at (5.875, 2.5) {};
		\node [style=none] (71) at (5.875, -2.5) {};
		\node [style=txtnobold] (72) at (5.75, 0) {$i$};
		\node [style=Empty red] (74) at (5.5, 4.5) {};
		\node [style=Empty red] (75) at (5.5, -4) {};
		\node [style=Empty green] (76) at (8, -4) {};
		\node [style=none] (81) at (5.5, -5.75) {$\ground$};
		\node [style=none] (82) at (5.5, -5.75) {};
		\node [style=txtnobold] (83) at (8, 3.5) {$\left(\frac{1}{\sqrt{2}}\right)^2$};
		\node [style=txtnobold] (84) at (2, 0) {$\stackrel{(h),(f)}{=}$};
		\node [style=Green Node] (85) at (13.5, 0.75) {$\pi$};
		\node [style=Red Node] (86) at (12.5, 0) {$\pi$};
		\node [style=Red Node] (87) at (13.5, -0.75) {$\pi$};
		\node [style=Green Node] (88) at (14.25, 0) {$\pi$};
		\node [style=txtnobold] (89) at (11.5, 2) {$1-p$};
		\node [style=txtnobold] (90) at (14.2, 2) {$\frac{p}{3}$};
		\node [style=txtnobold] (91) at (13.19, 2) {$\frac{p}{3}$};
		\node [style=txtnobold] (92) at (12.45, 2) {$\frac{p}{3}$};
		\node [style=none] (93) at (13, 3) {};
		\node [style=none] (94) at (12.75, 3) {};
		\node [style=none] (95) at (13.25, 3) {};
		\node [style=none] (96) at (12.5, 2.5) {};
		\node [style=none] (97) at (13.5, 2.5) {};
		\node [style=none] (98) at (12.5, 2.5) {};
		\node [style=none] (99) at (13, 2.5) {};
		\node [style=none] (100) at (13.25, 2.5) {};
		\node [style=none] (101) at (12.75, 2.5) {};
		\node [style=none] (102) at (13, -3) {};
		\node [style=none] (103) at (12.75, -3) {};
		\node [style=none] (104) at (13.25, -3) {};
		\node [style=none] (105) at (12.5, -2.5) {};
		\node [style=none] (106) at (13.5, -2.5) {};
		\node [style=none] (107) at (12.5, -2.5) {};
		\node [style=none] (108) at (13, -2.5) {};
		\node [style=none] (109) at (13.25, -2.5) {};
		\node [style=none] (110) at (12.75, -2.5) {};
		\node [style=none] (111) at (12.6, -2.5) {};
		\node [style=none] (112) at (12.6, 2.5) {};
		\node [style=none] (113) at (13.375, 2.5) {};
		\node [style=none] (114) at (13.375, -2.5) {};
		\node [style=txtnobold] (115) at (13.25, 0) {$i$};
		\node [style=Empty red] (117) at (13, 4.5) {};
		\node [style=none] (122) at (13, -6.25) {$\ground$};
		\node [style=none] (123) at (13, -6.25) {};
		\node [style=txtnobold] (124) at (15.5, 3.5) {$\left(\frac{1}{\sqrt{2}}\right)^3$};
		\node [style=txtnobold] (125) at (9.5, 0) {$\stackrel{(\pi)}{=}$};
		\node [style=Empty green] (126) at (13, -5.25) {};
		\node [style=Empty green] (127) at (13, -4) {};
		\node [style=h] (128) at (-2.5, 3.75) {};
		\node [style=h] (129) at (5.5, 3.75) {};
		\node [style=h] (130) at (13, 3.75) {};
	\end{pgfonlayer}
	\begin{pgfonlayer}{edgelayer}
		\draw (2) to (0);
		\draw [style=distribution] (8.center)
			 to (9.center)
			 to (11.center)
			 to (13.center)
			 to (14.center)
			 to (12.center)
			 to (10.center)
			 to cycle;
		\draw [style=distribution] (21.center)
			 to (19.center)
			 to (17.center)
			 to (18.center)
			 to (20.center)
			 to (22.center)
			 to (23.center)
			 to cycle;
		\draw [in=75, out=-90] (2) to (24.center);
		\draw [in=105, out=-90, looseness=0.75] (1) to (25.center);
		\draw [in=90, out=-105, looseness=0.75] (16.center) to (1);
		\draw [in=90, out=-75] (15.center) to (0);
		\draw [in=150, out=-150, looseness=0.75] (27.center) to (26.center);
		\draw [in=75, out=-30, looseness=0.75] (28.center) to (3);
		\draw [in=45, out=-75, looseness=0.75] (3) to (29.center);
		\draw (33) to (34);
		\draw (34) to (35);
		\draw (35) to (37);
		\draw (34) to (36);
		\draw (36) to (38);
		\draw (17.center) to (33);
		\draw (33) to (40.center);
		\draw (44) to (42);
		\draw [style=distribution] (50.center)
			 to (51.center)
			 to (53.center)
			 to (55.center)
			 to (56.center)
			 to (54.center)
			 to (52.center)
			 to cycle;
		\draw [style=distribution] (63.center)
			 to (61.center)
			 to (59.center)
			 to (60.center)
			 to (62.center)
			 to (64.center)
			 to (65.center)
			 to cycle;
		\draw [in=75, out=-90] (44) to (66.center);
		\draw [in=105, out=-90, looseness=0.75] (43) to (67.center);
		\draw [in=90, out=-105, looseness=0.75] (58.center) to (43);
		\draw [in=90, out=-75] (57.center) to (42);
		\draw [in=150, out=-150, looseness=0.75] (69.center) to (68.center);
		\draw [in=75, out=-30, looseness=0.75] (70.center) to (45);
		\draw [in=45, out=-75, looseness=0.75] (45) to (71.center);
		\draw (75) to (76);
		\draw (59.center) to (75);
		\draw (75) to (82.center);
		\draw (87) to (85);
		\draw [style=distribution] (94.center)
			 to (96.center)
			 to (98.center)
			 to (99.center)
			 to (97.center)
			 to (95.center)
			 to (93.center)
			 to cycle;
		\draw [style=distribution] (103.center)
			 to (105.center)
			 to (107.center)
			 to (108.center)
			 to (106.center)
			 to (104.center)
			 to (102.center)
			 to cycle;
		\draw [in=75, out=-90] (87) to (109.center);
		\draw [in=105, out=-90, looseness=0.75] (86) to (110.center);
		\draw [in=90, out=-105, looseness=0.75] (101.center) to (86);
		\draw [in=90, out=-75] (100.center) to (85);
		\draw [in=150, out=-150, looseness=0.75] (112.center) to (111.center);
		\draw [in=75, out=-30, looseness=0.75] (113.center) to (88);
		\draw [in=45, out=-75, looseness=0.75] (88) to (114.center);
		\draw (102.center) to (127);
		\draw (126) to (123.center);
		\draw (32) to (128);
		\draw (128) to (8.center);
		\draw (74) to (129);
		\draw (129) to (50.center);
		\draw (117) to (130);
		\draw (130) to (93.center);
	\end{pgfonlayer}
\end{tikzpicture}

%% file: figures/derivedSV2.tikz
\begin{tikzpicture}
	\begin{pgfonlayer}{nodelayer}
		\node [style=Green Node] (0) at (0.5, 0.75) {$\pi$};
		\node [style=Red Node] (1) at (-0.5, 0) {$\pi$};
		\node [style=Red Node] (2) at (0.5, -0.75) {$\pi$};
		\node [style=Green Node] (3) at (1.25, 0) {$\pi$};
		\node [style=txtnobold] (4) at (-1.5, 2) {$1-p$};
		\node [style=txtnobold] (5) at (1.2, 2) {$\frac{p}{3}$};
		\node [style=txtnobold] (6) at (0.19, 2) {$\frac{p}{3}$};
		\node [style=txtnobold] (7) at (-0.55, 2) {$\frac{p}{3}$};
		\node [style=none] (8) at (0, 3) {};
		\node [style=none] (9) at (-0.25, 3) {};
		\node [style=none] (10) at (0.25, 3) {};
		\node [style=none] (11) at (-0.5, 2.5) {};
		\node [style=none] (12) at (0.5, 2.5) {};
		\node [style=none] (13) at (-0.5, 2.5) {};
		\node [style=none] (14) at (0, 2.5) {};
		\node [style=none] (15) at (0.25, 2.5) {};
		\node [style=none] (16) at (-0.25, 2.5) {};
		\node [style=none] (17) at (0, -3) {};
		\node [style=none] (18) at (-0.25, -3) {};
		\node [style=none] (19) at (0.25, -3) {};
		\node [style=none] (20) at (-0.5, -2.5) {};
		\node [style=none] (21) at (0.5, -2.5) {};
		\node [style=none] (22) at (-0.5, -2.5) {};
		\node [style=none] (23) at (0, -2.5) {};
		\node [style=none] (24) at (0.25, -2.5) {};
		\node [style=none] (25) at (-0.25, -2.5) {};
		\node [style=none] (26) at (-0.4, -2.5) {};
		\node [style=none] (27) at (-0.4, 2.5) {};
		\node [style=none] (28) at (0.375, 2.5) {};
		\node [style=none] (29) at (0.375, -2.5) {};
		\node [style=txtnobold] (30) at (0.25, 0) {$i$};
		\node [style=Empty red] (32) at (0, 4.5) {};
		\node [style=txtnobold] (35) at (1.75, 3.5) {$\frac{1}{2}$};
		\node [style=Empty green] (37) at (0, -4) {};
		\node [style=txtnobold] (38) at (3, 0) {$\stackrel{(es)}{=}$};
		\node [style=Green Node] (39) at (8.5, 0.5) {$\pi$};
		\node [style=Red Node] (40) at (7, -0.25) {$\pi$};
		\node [style=Red Node] (41) at (8.5, -1) {$\pi$};
		\node [style=Green Node] (42) at (10, -0.25) {$\pi$};
		\node [style=txtnobold] (43) at (4.75, 2) {$1-p$};
		\node [style=txtnobold] (44) at (8.2, 2) {$\frac{p}{3}$};
		\node [style=txtnobold] (45) at (9.69, 2) {$\frac{p}{3}$};
		\node [style=txtnobold] (46) at (6.7, 2) {$\frac{p}{3}$};
		\node [style=none] (47) at (7.75, 3.5) {};
		\node [style=none] (48) at (7.5, 3.5) {};
		\node [style=none] (49) at (8, 3.5) {};
		\node [style=none] (50) at (7.25, 3) {};
		\node [style=none] (51) at (8.25, 3) {};
		\node [style=none] (52) at (7.25, 3) {};
		\node [style=none] (53) at (7.75, 3) {};
		\node [style=none] (54) at (8, 3) {};
		\node [style=none] (55) at (7.5, 3) {};
		\node [style=none] (56) at (7.75, -3) {};
		\node [style=none] (57) at (7.5, -3) {};
		\node [style=none] (58) at (8, -3) {};
		\node [style=none] (59) at (7.25, -2.5) {};
		\node [style=none] (60) at (8.25, -2.5) {};
		\node [style=none] (61) at (7.25, -2.5) {};
		\node [style=none] (62) at (7.75, -2.5) {};
		\node [style=none] (63) at (8, -2.5) {};
		\node [style=none] (64) at (7.5, -2.5) {};
		\node [style=none] (65) at (7.35, -2.5) {};
		\node [style=none] (66) at (7.35, 3) {};
		\node [style=none] (67) at (8.125, 3) {};
		\node [style=none] (68) at (8.125, -2.5) {};
		\node [style=txtnobold] (69) at (8.25, -0.25) {$i$};
		\node [style=txtnobold] (72) at (11, 2.75) {$\frac{1}{2}$};
		\node [style=Empty green] (73) at (5.5, -2) {};
		\node [style=Empty green] (74) at (7, -2) {};
		\node [style=Empty green] (75) at (8.5, -2) {};
		\node [style=Empty green] (76) at (10, -2) {};
		\node [style=Empty red] (78) at (5.5, 2.5) {};
		\node [style=Empty red] (80) at (7, 2.5) {};
		\node [style=Empty red] (82) at (8.5, 2.5) {};
		\node [style=Empty red] (84) at (10, 2.5) {};
		\node [style=h] (85) at (0, 3.75) {};
		\node [style=h] (86) at (5.5, 1.5) {};
		\node [style=h] (87) at (7, 1.5) {};
		\node [style=h] (88) at (8.5, 1.5) {};
		\node [style=h] (89) at (10, 1.5) {};
		\node [style=txtnobold] (90) at (-3.75, -7.75) {$\stackrel{(h)}{=}$};
		\node [style=Green Node] (91) at (1.75, -7.25) {$\pi$};
		\node [style=Red Node] (92) at (0.25, -8) {$\pi$};
		\node [style=Red Node] (93) at (1.75, -8.75) {$\pi$};
		\node [style=Green Node] (94) at (3.25, -8) {$\pi$};
		\node [style=txtnobold] (95) at (-2, -6.25) {$1-p$};
		\node [style=txtnobold] (96) at (1.45, -6.25) {$\frac{p}{3}$};
		\node [style=txtnobold] (97) at (2.94, -6.25) {$\frac{p}{3}$};
		\node [style=txtnobold] (98) at (-0.05, -6.25) {$\frac{p}{3}$};
		\node [style=none] (99) at (1, -4.5) {};
		\node [style=none] (100) at (0.75, -4.5) {};
		\node [style=none] (101) at (1.25, -4.5) {};
		\node [style=none] (102) at (0.5, -5) {};
		\node [style=none] (103) at (1.5, -5) {};
		\node [style=none] (104) at (0.5, -5) {};
		\node [style=none] (105) at (1, -5) {};
		\node [style=none] (106) at (1.25, -5) {};
		\node [style=none] (107) at (0.75, -5) {};
		\node [style=none] (108) at (1, -10.75) {};
		\node [style=none] (109) at (0.75, -10.75) {};
		\node [style=none] (110) at (1.25, -10.75) {};
		\node [style=none] (111) at (0.5, -10.25) {};
		\node [style=none] (112) at (1.5, -10.25) {};
		\node [style=none] (113) at (0.5, -10.25) {};
		\node [style=none] (114) at (1, -10.25) {};
		\node [style=none] (115) at (1.25, -10.25) {};
		\node [style=none] (116) at (0.75, -10.25) {};
		\node [style=none] (117) at (0.6, -10.25) {};
		\node [style=none] (118) at (0.6, -5) {};
		\node [style=none] (119) at (1.375, -5) {};
		\node [style=none] (120) at (1.375, -10.25) {};
		\node [style=txtnobold] (121) at (1.5, -8) {$i$};
		\node [style=txtnobold] (122) at (4.25, -5) {$\frac{1}{2}$};
		\node [style=Empty green] (123) at (-1.25, -9.75) {};
		\node [style=Empty green] (124) at (0.25, -9.75) {};
		\node [style=Empty green] (125) at (1.75, -9.75) {};
		\node [style=Empty green] (126) at (3.25, -9.75) {};
		\node [style=Empty green] (127) at (-1.25, -5.75) {};
		\node [style=Empty green] (128) at (0.25, -5.75) {};
		\node [style=Empty green] (129) at (1.75, -5.75) {};
		\node [style=Empty green] (130) at (3.25, -5.75) {};
		\node [style=txtnobold] (131) at (5.25, -7.75) {$\stackrel{(f),(\pi)}{=}$};
		\node [style=Green Node] (132) at (10.5, -7.75) {$\pi$};
		\node [style=Green Node] (135) at (11.75, -7.75) {$\pi$};
		\node [style=txtnobold] (136) at (7.25, -7.25) {$1-p$};
		\node [style=txtnobold] (137) at (9.95, -7.25) {$\frac{p}{3}$};
		\node [style=txtnobold] (138) at (11.19, -7.25) {$\frac{p}{3}$};
		\node [style=txtnobold] (139) at (8.7, -7.25) {$\frac{p}{3}$};
		\node [style=none] (140) at (9.75, -6) {};
		\node [style=none] (141) at (9.5, -6) {};
		\node [style=none] (142) at (10, -6) {};
		\node [style=none] (143) at (9.25, -6.5) {};
		\node [style=none] (144) at (10.25, -6.5) {};
		\node [style=none] (145) at (9.25, -6.5) {};
		\node [style=none] (146) at (9.75, -6.5) {};
		\node [style=none] (147) at (10, -6.5) {};
		\node [style=none] (148) at (9.5, -6.5) {};
		\node [style=none] (149) at (9.75, -9.5) {};
		\node [style=none] (150) at (9.5, -9.5) {};
		\node [style=none] (151) at (10, -9.5) {};
		\node [style=none] (152) at (9.25, -9) {};
		\node [style=none] (153) at (10.25, -9) {};
		\node [style=none] (154) at (9.25, -9) {};
		\node [style=none] (155) at (9.75, -9) {};
		\node [style=none] (156) at (10, -9) {};
		\node [style=none] (157) at (9.5, -9) {};
		\node [style=none] (158) at (9.35, -9) {};
		\node [style=none] (159) at (9.35, -6.5) {};
		\node [style=none] (160) at (10.125, -6.5) {};
		\node [style=none] (161) at (10.125, -9) {};
		\node [style=txtnobold] (162) at (10, -7.75) {$i$};
		\node [style=txtnobold] (163) at (11.25, -6.25) {$\frac{1}{2}$};
		\node [style=Empty green] (164) at (7.5, -7.75) {};
		\node [style=Empty green] (165) at (9, -7.75) {};
	\end{pgfonlayer}
	\begin{pgfonlayer}{edgelayer}
		\draw (2) to (0);
		\draw [style=distribution] (9.center)
			 to (11.center)
			 to (13.center)
			 to (14.center)
			 to (12.center)
			 to (10.center)
			 to (8.center)
			 to cycle;
		\draw [style=distribution] (18.center)
			 to (20.center)
			 to (22.center)
			 to (23.center)
			 to (21.center)
			 to (19.center)
			 to (17.center)
			 to cycle;
		\draw [in=75, out=-90] (2) to (24.center);
		\draw [in=105, out=-90, looseness=0.75] (1) to (25.center);
		\draw [in=90, out=-105, looseness=0.75] (16.center) to (1);
		\draw [in=90, out=-75] (15.center) to (0);
		\draw [in=150, out=-150, looseness=0.75] (27.center) to (26.center);
		\draw [in=75, out=-30, looseness=0.75] (28.center) to (3);
		\draw [in=45, out=-75, looseness=0.75] (3) to (29.center);
		\draw (17.center) to (37);
		\draw (41) to (39);
		\draw [style=distribution] (51.center)
			 to (49.center)
			 to (47.center)
			 to (48.center)
			 to (50.center)
			 to (52.center)
			 to (53.center)
			 to cycle;
		\draw [style=distribution] (57.center)
			 to (59.center)
			 to (61.center)
			 to (62.center)
			 to (60.center)
			 to (58.center)
			 to (56.center)
			 to cycle;
		\draw (40) to (74);
		\draw (41) to (75);
		\draw (42) to (76);
		\draw (32) to (85);
		\draw (85) to (8.center);
		\draw (78) to (86);
		\draw (86) to (73);
		\draw (80) to (87);
		\draw (87) to (40);
		\draw (82) to (88);
		\draw (88) to (39);
		\draw (89) to (42);
		\draw (84) to (89);
		\draw (93) to (91);
		\draw [style=distribution] (103.center)
			 to (101.center)
			 to (99.center)
			 to (100.center)
			 to (102.center)
			 to (104.center)
			 to (105.center)
			 to cycle;
		\draw [style=distribution] (109.center)
			 to (111.center)
			 to (113.center)
			 to (114.center)
			 to (112.center)
			 to (110.center)
			 to (108.center)
			 to cycle;
		\draw (92) to (124);
		\draw (93) to (125);
		\draw (94) to (126);
		\draw (127) to (123);
		\draw (128) to (92);
		\draw (129) to (91);
		\draw (130) to (94);
		\draw [style=distribution] (144.center)
			 to (142.center)
			 to (140.center)
			 to (141.center)
			 to (143.center)
			 to (145.center)
			 to (146.center)
			 to cycle;
		\draw [style=distribution] (150.center)
			 to (152.center)
			 to (154.center)
			 to (155.center)
			 to (153.center)
			 to (151.center)
			 to (149.center)
			 to cycle;
	\end{pgfonlayer}
\end{tikzpicture}

%% file: content/notation.tex
\section{Notation}
\label{sec:notation}

\begin{center}
  \begin{tabular}{l|l}
    Notation & Meaning \\\hline
    $\SetC$ & Category of sets \\
    $\emt$ & Eilenberg-Moore category of a monad \\
    $\CatC$ & Category of (small) categories \\
    $\obj{\CC}$ & Objects of the category $\CC$ \\
    $\eC$, $\enr{\Cat{M}}$ & Enriched categories \\
    $\timesUnit$ & Monoidal unit of $\times$ (singleton set in $\SetC$) \\
    $\otimesUnit$ & Monoidal unit of $\otimes$ \\
    $\ecirc$ & Composition in an enriched category \\
    $\odot$ & Tensor product in a monoidal enriched category \\
    $\otimes^\mt$ & Tensor product in $\emt$ \\
    $\VCat$, $\ECatC{\Cat{M}}$ & 2-category of enriched category
  \end{tabular}
\end{center}

%%% Local Variables:
%%% mode: latex
%%% TeX-master: "../enriched-diagrams"
%%% End: